\documentclass[11pt]{article}

\usepackage{fullpage}
\usepackage{microtype}
\usepackage[utf8]{inputenc}
\usepackage{amsfonts,amsmath,amsthm} 
\usepackage{mathtools} 
\usepackage{graphicx} 
\usepackage{listings} 
\usepackage{hyperref}
\usepackage{url}
\usepackage{xspace}
\usepackage{subcaption}
\usepackage{bussproofs} \EnableBpAbbreviations
\usepackage[disable]{todonotes}
\usepackage[inline]{enumitem}
\usepackage{multicol}
\usepackage[notextcomp]{stix} 
\usepackage{tikz} \usetikzlibrary{arrows}
\usepackage{comment}

\newtheorem{theorem}{Theorem}
\newtheorem{lemma}[theorem]{Lemma}

\newtheorem{remark}{Remark}
\newtheorem{example}{Example}

\DeclareUnicodeCharacter{207A}{^+}
\DeclareUnicodeCharacter{207B}{^-}
\DeclareUnicodeCharacter{2191}{\mathop{\uparrow}}
\DeclareUnicodeCharacter{2193}{\mathop{\downarrow}}
\DeclareUnicodeCharacter{2205}{\varnothing}
\DeclareUnicodeCharacter{2208}{\in}
\DeclareUnicodeCharacter{2209}{\notin}
\DeclareUnicodeCharacter{222A}{\cup}
\DeclareUnicodeCharacter{2286}{\subseteq}
\DeclareUnicodeCharacter{22A5}{\bot}

\newcommand{\ie}{i.e.,\ }

\newcommand{\eg}{e.g.}
\newcommand{\resp}{resp.\ }

\newcommand{\N}{\ensuremath{\mathbb{N}}}
\newcommand\tactic[1]{\texttt{#1}}

\newtheorem*{coqrmk}{Coq Remark}

\newcommand\sep{\qquad}

\newlength{\offset}
\newenvironment{listdef*}{
  \begin{trivlist}
  \item[]
    \begin{tabular}{@{\hspace{\offset}}l@{\qquad}r@{}r@{}ll@{}}
}{
    \end{tabular}
  \end{trivlist}
}

\def\Longarrow#1#2{\setbox1=\hbox{$\ \scriptstyle{#1}\ $}
                   \setbox2=\hbox{$\ \scriptstyle{#2}\ $}
                   \ifnum\wd1>\wd2
                      \mathop{\hbox to
                        \wd1{\rightarrowfill}}\limits_{\box1}^{\box2}
                   \else
                     \mathop{\hbox to
                        \wd2{\rightarrowfill}}\limits_{\box1}^{\box2}
                   \fi
                  }

\def\Longconstrarrow#1#2{\setbox1=\hbox{$\ \scriptstyle{#1}\ $}
                   \setbox2=\hbox{$\ \scriptstyle{#2}\ $}
                   \ifnum\wd1>\wd2
                      \lhook\mathrel{\mkern-10mu}\mathop{\hbox to
                        \wd1{\rightarrowfill}}\limits_{\box1}^{\box2}
                   \else
                     \lhook\mathrel{\mkern-10mu}\mathop{\hbox to
                        \wd2{\rightarrowfill}}\limits_{\box1}^{\box2}
                   \fi
                  }

\newcommand{\SetProgSymbols}{\catcode`\#=12\catcode`\%=12\catcode`\_=12}

\newenvironment{newprogram}
{\begin{quote}\tt\obeylines\obeyspaces\SetProgSymbols%
\addtolength{\parskip}{0pt}}
{\end{quote}\par\noindent}

{\obeyspaces\gdef {\ }}

\def\picture #1 by #2 (#3){
  \makebox[#1][l]{
    \special{pictfile #3} 
    \rule{0pt}{#2}
    }
  }

\def\scaledpicture #1 by #2 (#3 scaled #4){{
  \dimen0=#1 \dimen1=#2
  \divide\dimen0 by 1000 \multiply\dimen0 by #4
  \divide\dimen1 by 1000 \multiply\dimen1 by #4
  \picture \dimen0 by \dimen1 (#3 scaled #4)}
  }

\newcount\DrawingMode
\newcommand{\Drawing}[4]{
          \ifnum \DrawingMode=0
                 \scaledpicture #1 by #2 (#3 scaled #4)
          \else
                 {\epsfig{file={#3.ps}, 
                          width=\ifdim#1>\textwidth 
                                      \textwidth 
                                 \else 
                                      #1 
                                 \fi}}
          \fi}

\newcommand{\id}[1]{{\tt #1}} 
\newcommand{\Tnothing}{\id{nothing}}
\newcommand{\Tpause}{\id{pause}}
\newcommand{\Tgotopause}{\id{gotopause}}
\newcommand{\Tgoto}{\id{goto}}
\newcommand{\Texit}[1]{\id{exit}~{#1}}
\newcommand{\Tawait}[1]{\id{await}~{#1}}
\newcommand{\Tawimm}[1]{\id{await}\ \id{immediate}~{#1}}
\newcommand{\Temit}[1]{\id{emit}~{#1}}
\newcommand{\Tif}[3]{\id{if}~#1~\id{then}~#2~\id{else}~#3~\id{end}}
\newcommand{\Tsequence}[2]{{#1}\,;\,{#2}}
\newcommand{\Tloop}[1]{\id{loop}~#1~\id{end}}
\newcommand{\Tparallel}[2]{#1\, \id{||}\, #2}
\newcommand{\Ttrap}[2]{\id{trap}~#1~\id{in}~#2~\id{end}}
\newcommand{\Tsuspend}[2]{\id{suspend}~#2~\id{when}~#1}

\newcommand{\Tsignaldecl}[2]{\id{signal}~#1~\id{in}~#2~\id{end}}

\newcommand{\Sk}[1]{{#1}}
\newcommand{\Snothing}{\Sk 0}
\newcommand{\Spause}{\Sk 1}

\newcommand{\Sawimm}[1]{\id{@} {#1}}
\newcommand{\Semit}[1]{!{#1}}
\newcommand{\Sif}[3]{{#1}\,?\,{#2}\,,\,{#3}}
\newcommand{\Ssequence}[2]{{#1}\,;\,{#2}}
\newcommand{\Sloop}[1]{{#1}\mathop{*}}
\newcommand{\Sparallel}[2]{{#1} \mathbin{|} {#2}}
\newcommand{\Strap}[1]{\left\{{#1}\right\}}
\newcommand{\Sexit}[1]{{#1}}
\newcommand{\Sshift}[1]{\operatorname{\uparrow}{#1}}
\newcommand{\Ssuspend}[2]{{#1} \operatorname{\mathord{\supset}} {#2}}

\newcommand{\Ssignaldecl}[2]{{#2} \backslash {#1}}

\newcommand{\Tpar}[2]{\Tparallel{#1}{#2}}
\newcommand{\Sseq}[2]{\Ssequence{#1}{#2}}
\newcommand{\Spar}[2]{\Sparallel{#1}{#2}}

\newcommand\state[1]{\widehat{#1}}
\newcommand{\activatedpause}{\state{\Spause}}

\newcommand\term[1]{\overline{#1}}

\definecolor{darkgreen}{rgb}{0,0.7,0}
\newcommand\Sin[1][]{\mathop{#1 \text{\scriptsize $\squarelrblack$}}}

\newcommand\SWin[1][]{\mathop{#1 \text{\scriptsize $\square$}}}
\newcommand\SBin[1][]{\mathop{#1 \text{\scriptsize $\blacksquare$}}}
\newcommand\go{\SWin[\color{darkgreen}]\!^+}
\newcommand\nogo{\SWin[\color{red}]\!^-}

\newcommand\res{\SBin[\color{darkgreen}]\!^+}
\newcommand\nores{\SBin[\color{red}]\!^-}
\newcommand\gores{\Sin[\color{darkgreen}]\!^+}
\newcommand\nogores{\Sin[\color{red}]\!^-}
\DeclareMathOperator\inCsymbol{in}
\newcommand\inC[1]{\inCsymbol(#1)}
\newcommand\WIRE[3]{\operatorname{\text{WIRE}}(#1, #2, #3)}
\newcommand\PAUSE[3]{\operatorname{\text{PAUSE}}(#1, #2, #3)}
\newcommand\band{\mathbin{\&}}

\newcommand\Sout{\operatorname{\text{\scriptsize $\circlerighthalfblack$}}}
\newcommand\black[1]{\operatorname{\text{\scriptsize $\mdlgblkcircle_{#1}$}}}
\newcommand\white[1]{\operatorname{\text{\scriptsize $\mdwhtcircle_{#1}$}}}
\newcommand\whiteAll{\white{∅}}
\DeclareMathOperator\outCsymbol{out}
\newcommand\outC[1]{\outCsymbol(#1)}
\DeclareMathOperator\outtoCsymbol{\text{out2C}}
\newcommand\outtoC[1]{\outtoCsymbol(#1)}

\newcommand\VC[2]{\ensuremath{\operatorname{\text{\texttt{valid\_coloring}}}(#1, #2)}}
\newcommand\VD[2]{\ensuremath{\operatorname{\text{\texttt{valid\_dom}}}(#1, #2)}}

\newcommand\changeI[2]{[#1] #2}

\newcommand\fromstmtName{from\_stmt}
\newcommand\fromstmt[1]{\operatorname{\text{\fromstmtName}}(#1)}
\newcommand\fromstateName{from\_state}
\newcommand\fromstate[1]{\operatorname{\text{\fromstateName}}(#1)}
\newcommand\totermName{to\_term}
\newcommand\toterm[2]{\operatorname{\text{\totermName}}(#1, #2)}
\newcommand\toeventName{to\_event}
\newcommand\toevent[2]{\operatorname{\text{\toeventName}}(#1, #2)}

\newcommand{\SingletonEvent}[1]{\{{#1}^+\}}
\newcommand{\addEvent}[3]{#3 * #1^{#2}}
\newcommand{\restrictEvent}[2]{{#1 \setminus #2}}
\newcommand\CtoKName{\id{C2K}}
\newcommand\TotalName{\id{Total}}
\newcommand\CtoK[1]{\CtoKName({#1})}
\newcommand\Total[1]{\TotalName({#1})}
\DeclareMathOperator\dom{dom}

\newlength{\rulevspace}
\newlength{\rulehspace}
\newcommand\SigP{sig${}^+$}
\newcommand\SigM{sig${}^-$}
\newcommand\CSigP{C-sig${}^+$}
\newcommand\CSigM{C-sig${}^-$}

\newcommand\deltafun[2]{\delta(#1, #2)}

\makeatletter
\newcommand{\xdoubleheadrightarrow}[3][]{%
  #2\joinrel
  \ext@arrow 0359\rightarrowfill@ {#1}{#3}%
  \mathrel{\mspace{-20mu}}\rightarrow
}
\makeatother

\newcommand\LBS[6][]{#2 \, \xrightarrow[#3]{\; #4, \, #5 \;}_{\text{\tiny #1}} #6}
\newcommand\CBS[5]{#1 \lhook\,\joinrel\xrightarrow[#2]{\; #3, \, #4 \;} #5}

\newcommand\CSS[6][]{#2 \xdoubleheadrightarrow[#3]{\lhook\,}{\; #4, \, #5 \;}_{\text{\tiny #1}} #6}
\newcommand\CSSs[6][s]{#2 \xdoubleheadrightarrow[#3]{\lhook\,}{\; #4, \, #5 \;}_{\text{\tiny #1}} #6}
\newcommand\CSSr[6][r]{#2 \xdoubleheadrightarrow[#3]{\lhook\,}{\; #4, \, #5 \;}_{\text{\tiny #1}} #6}
\newcommand\micro[4][]{#3 \, \xrightarrow[#2\,]{}_{\text{\tiny #1}} #4}
\newcommand\microexp[5][]{#4 \, \xrightarrow[#3\,]{}_{\text{\tiny #1}}^{#2} #5}
\newcommand\micros[4][]{\microexp[#1]{*}{#2}{#3}{#4}}

\newcommand\MAXFirstName{Max}

\DeclareMathOperator*{\MAXop}{\MAXFirstName}
\newcommand\MAX[2]{\MAXop({#1},{#2})}

\newcommand{\Kdown}[1]{\operatorname{\downarrow}{#1}}
\newcommand{\Kup}[1]{\operatorname{\uparrow}{#1}}
\newcommand{\SuspNow}[2]{\operatorname{SuspNow}({#1},{#2})}
\newcommand{\synchronize}[4]{\operatorname{synchronizer}_{{#1},{#2}}(#3, #4)}

\newcommand\baseName{\ensuremath{\cal B}}
\newcommand\base[1]{\mathop{\baseName}({#1})}

\newcommand\expandName{\ensuremath{\cal E}}
\newcommand\expand[1]{\mathop{\expandName}({#1})}

\newcommand{\CanName}{{\textit{Can}}}

\newcommand{\MustName}{{\textit{Must}}}

\newcommand{\CanAux}[6]{\CanName_{#1}^{#2}{#5}{#3},{#4}{#6}}
\newcommand{\Can}[3]{\CanAux{}{#1}{#2}{#3}{(}{)}}

\newcommand{\CanK}[3]{\CanAux{k}{#1}{#2}{#3}{(}{)}}
\newcommand{\CanS}[3]{\CanAux{s}{#1}{#2}{#3}{(}{)}}

\newcommand{\MustAux}[5]{\MustName_{#1}{#4}{#2},{#3}{#5}}
\newcommand{\Must}[2]{\MustAux{}{#1}{#2}{(}{)}}

\newcommand{\MustK}[2]{\MustAux{k}{#1}{#2}{(}{)}}
\newcommand{\MustS}[2]{\MustAux{s}{#1}{#2}{(}{)}}

\newcommand\baseNetAddress{./html/}
\newcommand\linkSymbol{\includegraphics[width=8mm]{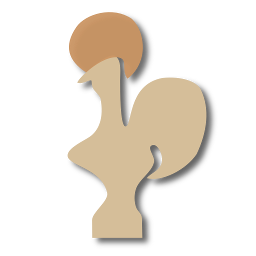}}
\usepackage{marginnote}
\newcommand\link[2][]{\href{\baseNetAddress#2.html#1}{\linkSymbol}}
\newcommand{\marginlink}[3][]%
           {\marginnote{\href{\baseNetAddress#2.html#1}{\linkSymbol}}[#3]}

\title{Towards a Coq-verified Chain of Esterel Semantics}
\author{Lionel~\textsc{Rieg} \and Gérard~\textsc{Berry}}
\date{\today}

\begin{document}
\maketitle

\begin{abstract}
  This paper focuses on formally specifying and verifying the
  chain of formal semantics of the Esterel synchronous programming
  language using the Coq proof assistant. In particular, in addition
  to the standard logical (LBS) semantics, constructive semantics
  (CBS) and constructive state semantics (CSS), we introduce a novel
  microstep semantics that gets rid of the Must/Can potential function
  pair of the constructive semantics and can be viewed as an abstract
  version of Esterel's circuit semantics used by compilers to generate
  software code and hardware designs. Excluding the loop construct from
  Esterel, the paper also provides formal proofs in Coq of the
  equivalence between the CBS and CSS semantics and of the refinement of
  the CSS by the microstep semantics.
\end{abstract}

\section{Introduction}
\label{sec:introduction}

The long-term goal of the research presented here is twofold: first,
formally validate the chain of semantics of the synchronous reactive
language Esterel~\cite{Berry:FoundationsOfEsterel} that leads from its
definition by formal semantics to its implementation; second, build a
formally proven compiler from Esterel to clocked digital circuits or C
code, in the spirit of the CompCert verified C
compiler~\cite{Leroy-Compcert-Coq}.  The tools we use for this is the
chain of Plotkin-style Structural Operational Semantics (SOS) semantics
developed for Esterel between 1984 and 2000 and published in the
web book~\cite{Berry:ConstructiveSemanticsOfPureEsterel}, which
was not submitted to an editor because it was
lacking formal proofs.  The subject being quite tricky, the second
author thought that the proofs should be done on computer --- which is what
we contribute to here.

We also introduce here a new operational semantics that gets closer to
the circuit semantics of
~\cite{Berry:ConstructiveSemanticsOfPureEsterel}, on which the Esterel
v5 compiler and later the industrial v7 compiler were based. We use
the Coq proof assistant~\cite{Coq-8.15.2} as the formal verification
environment.

This work is still partial : we limit ourselves to the loop-free part
of Kernel Esterel~\cite{Berry:FoundationsOfEsterel}, the core language
that only deals with pure signaling but still gathers all the
technical difficulties. The full Esterel language has loops and
involves data-valued signals and variables. We already know how to
efficiently handle loops with a method presented
in~\cite{Berry:ConstructiveSemanticsOfPureEsterel}, used in the
Esterel v5 and v7 compilers, or alternatively with Oliver Tardieu's
method published in~\cite{TardieuDeSimone:LoopsInEsterel}.
Data handling should not add extra difficulties since the only
addition in the semantics and compilers is the insertion of additional
dependencies in the control flow, which can be technically handled
almost in the same way as Kernel Esterel's simpler pure signal
dependencies.

\subsection{Synchrony, determinism, concurrency: a short history of Esterel}

Esterel, born in 1983~\cite{Esterel:BerryMoisanRigault}, is the oldest
member of the family of synchronous and deterministic concurrent
programming languages. It was initially dedicated to programming
discrete control systems found in industry such as airplanes,
automotive, process control in factories, etc., for which determinism
is a key feature.  While being concurrent, Esterel strongly differs
from classical asynchronous and non-deterministic concurrent languages
because it is synchronous and fully deterministic. It was introduced
just before Lustre~\cite{Halbwachs:IEEE:Lustre} and Signal~\cite{Signal}, two other
synchronous deterministic languages also dedicated to industrial
process control but with a different style : they are data-flow
oriented functions languages dedicated to controlling processes with
fairly constant control laws, while Esterel was dedicated to
control-intensive processes whose behavior keeps changing over
time. All these synchronous languages are based on well-defined and
formal semantics.

The Esterel language and formal semantics were developed together in a
series of steps, from a first semantics used in 1984 for the
Esterel~v2 compiler~\cite{Esterel:These:Cosserat} to the final
constructive semantics in 2000 that lead to the academic Esterel
v5~compiler~\cite{EsterelV591}. This compiler was used in academic or
industrial research for safety-critical embedded systems in a wider
variety of domains compared to the initial goal:
avionics~\cite{Berry:avionics},
robotics~\cite{Esterel:EspiauCosteManiere}, communication
protocols~\cite{Sethi:Terminal}, and the synthesis of efficient
synchronous digital
circuits~\cite{Esterel:TouatiBerry,BerryKishinevskySingh,SentovichTomaBerry}.
Another very efficient but more limited compiler was
developed for Esterel by Stephen Edwards at Columbia
University~\cite{Edwards-compiler}.  All this work is described in
detail in the \emph{Compiling Esterel} book~\cite{CompilingEsterel}.

From 2001 to 2009, an enriched v7 version was developed by the
second author and his team at the Esterel Technologies company, with
the very same kernel semantics but many language extensions and a compiler
able to generate both C programs for software applications and Verilog/VHDL
circuits for hardware applications. Esterel v7 was used in R\&D and
production by companies such as Intel, Texas Instruments, ST
microelectronics, NXP, etc. But the 2008 crisis unfortunately canceled
its development and usage.

The Esterel v5 academic compiler and debugging environment as well
as the Edwards Columbia compiler are still freely available.\footnote{Respectively
from \url{https://esterel.org/files/Html/Downloads/Downloads.htm} and
\url{http://www1.cs.columbia.edu/~sedwards/cec/}.}
Furthermore, the Esterel v5 semantics and compiling technology has
been transported to HipHop~\cite{HipHop}, a reactive extension of
Serrano's Hop language~\cite{Hop} which is itself an extension of
JavaScript that specifies client/server computations in a single file
and simplifies the design of web pages. HipHop is dedicated to program complex
behaviors of autonomous devices and man/machine Web-based interfaces.

In parallel, Lustre was industrialized by the French company Verilog
under a graphical form called SCADE\footnote{For Safety Critical
Applications Development Environment}, used for avionics by Airbus,
for the Lyon automatic subway control, for the French nuclear plants
safety systems, etc.  In 2006, SCADE and Esterel (with some
restrictions) were unified by Esterel Technologies into the SCADE~6
language and toolset~\cite{Scade-6}; the SCADE~6 compiler was certified at level DO-178C, the highest certification level for avionics.
Esterel Technologies was bought in
2012 by Ansys, and SCADE~6 is now widely used by several hundred
industrial customers for certified embedded systems and other
safety-critical applications.

\subsection{Signals in Esterel}

Let us first focus on how and why concurrency and determinism are
reconciled by Esterel. The shared objects between concurrent
statements are called \emph{signals}. A signal carries a presence~/
absence \emph{status} and optionally a unique \emph{data value}. The
execution of a program consists of a series of discrete
\emph{reactions}, where each reaction handles an input signal vector
from the environment and generates an output signal vector to the
environment in a deterministic way and conceptually in zero time. The
program can also use an arbitrary number of local signals declared by
a specific ``\Tsignaldecl{S}{\dots}'' scoping statement.

In this paper, as in~\cite{Berry:ConstructiveSemanticsOfPureEsterel},
we only deal with \emph{pure signals} that only carry a Boolean
\emph{present~/~absent status}. In each reaction, a pure signal
\id{S} is \emph{absent} by default; it is set \emph{present} only if
the environment says so for an input signal or if it is explicitly
emitted within the program by an ``\Temit{S}'' statement that
broadcasts this status in the signal scope. Valued signals would introduce no
major difficulties in the semantics and compiling technology since
they simply add data dependencies that can be handled in almost the
same way as signal dependencies, see~\cite{CompilingEsterel}, but they
would make the formalization bigger.

\subsection{Causality issues}

Synchrony immediately leads to causality issues we can explain
with the following three examples:

\begin{newprogram}
P1: if S then emit S end
P2: if S then nothing else emit S end
P3: if S then emit S else emit S end
\end{newprogram}%
\id{P1} looks non-deterministic : if \id{S} is present, it is emitted,
which is logically fine; if \id{S} is absent, it is not emitted,
logically fine as well.  \id{P2} is clearly nonsensical: if \id{S} is
present, it is not emitted, contradicting the aforementioned basic
rule; if \id{S} is absent, it is emitted, which sets it present, again
a contradiction.  But \id{P3} is more problematic: in classical
logic, one would say that \id{S} present is not contradictory since it
makes it emitted, while \id{S} absent is contradictory since it makes
it emitted as well. Thus, classical logic would accept \id{S} present
as the unique solution, by the excluded middle rule. But we don't
accept this viewpoint for Esterel because it is neither natural nor
understandable in big programs. We want \emph{causal and constructive}
reasoning.

An easy way to reject all three cases above is to perform a topological
sorting of signals, forbidding a signal to depend on itself even by a
long chain of other signals that ends up with itself. This was adopted
by the first formal semantics of Esterel in
1984~\cite{Esterel:BerryCosserat} and it remains the rule for many
aforementioned synchronous languages. But our main industrial user
in the 1990's, a Dassault Aviation
team, complained because it was developing large modular Esterel
program for safety-critical avionics that sometimes led for
good reasons to the following dependency structure:

\begin{newprogram}
if S1 then
   if S2 then emit S3 end
else
   if S3 then emit S2 end
end
\end{newprogram}%
Here, the \id{S2} and \id{S3} signals
cannot be topologically ordered since the two \id{if} statements
generate a trivially cyclic dependency graph, but there is no
causality problem since the \id{then} and \id{else} branches cannot be
executed together in a single reaction.

Gonthier proposed a clever but partial
solution~\cite{Esterel:These:Gonthier} that was implemented by the Esterel team
together
with many other improvements in the Esterel v3 compiler he designed
with the second author.  Gonthier's solution did handle many cases but
not all of them, sometimes not being able to tell the programmer
whether her program is correct or incorrect --- this
definitely had to be improved for industrial usage.

Finally, the
causality problem has been fully solved by the \emph{constructive
semantics}, first presented in
\cite{Berry:ConstructiveSemanticsOfPureEsterel}, the one detailed and
proven correct here. It is used by the Esterel~v5 and v7 compilers.

\begin{remark}
  One could think that \id{P2} above provokes some sort of deadlock,
  as found when programming in asynchronous concurrent languages.
  Asynchronous deadlock are relatively easy to make and hard to detect
  at runtime, especially if there are partial, letting other threads
  of the program run normally.
  Therefore, in the asynchronous domain, it is highly recommended to
  stick to provably deadlock-free asynchronous protocols.
  The situation is different and much simpler in synchronous
  languages, since deadlocks are always detected and forbidden at
  compile-time.
  But the goals of both kinds of languages are quite different.
\end{remark}

\subsection{Switching to the circuit semantics and implementation}

In 1989, the second author was invited by Jean Vuillemin in the
Digital Equipment Paris (DEC) Research Lab to work on control parts of
computation-oriented circuits implemented on the first Xilinx FPGAs
(rewirable-by-software hardware circuits). They realized that
Esterel's synchrony hypothesis was well-known for acyclic synchronous
hardware circuits: electrical fronts propagate in wires in an asynchronous way,
but with a deterministic result computed in a predictable and very
short time; at the end of each clock cycle, all wires are electrically
stable to voltages that exactly correspond to the solution of the
circuit's Boolean equations.

Since this is exactly Esterel's viewpoint, they could directly implement
Esterel on synchronous circuits or FPGAs, each reaction lasting one
clock cycle with data computations attached to particular gates.
Furthermore, the generated circuits could be easily and
reasonably efficiently simulated to build a software code usable
either as a circuit simulator or directly in software
applications. The Esterel~v4 compiler followed, with the benefit of
definitively solving the generated code size explosion problem often
encountered with Esterel~v3. The software generation could now scale
up to very large problems.  Then, very efficient optimizers and
verification methods for the generated circuits were developed based on
Boolean Decision Diagrams (BDDs)~\cite{BDD}, the optimized circuits
turned out to be often smaller and more efficient than
when designed as usual in Verilog or VHDL, and the generated of
software code was improved in the same way.

Esterel v4 was limited to the acyclic case, a natural condition in
the hardware field. However, a seminal paper by
Malik~\cite{MalikCyclic} studied cyclic hardware circuits and showed
that some of them did compute deterministically and in bounded time,
as for acyclic ones.  In that case, they could even be much more
symmetric and natural than any equivalent acyclic circuit.

The second author then realized that the Esterel cycle analysis
problem could also be stated on cyclic circuits's Boolean equations
and that it could be solved at that level by replacing classical
Boolean logic by \emph{constructive Boolean logic}, which is a logic
that only propagates known Boolean values through logic operators,
without ever using excluded middle reasoning.  The adaptation of this
constructivity idea to Esterel led him to the final \emph{constructive
semantics} of Esterel described in the web
book~\cite{Berry:ConstructiveSemanticsOfPureEsterel}.  A new circuit
generation backend able to generate correct cyclic circuits was built
for the Esterel~v4 compiler, then renamed Esterel~v5, and a BDD-based
tool for statically determining if the cyclic circuit is constructive
and optionally to generate an acyclic equivalent (but possibly much
bigger) was added to it. The circuit implementation was later used for
the Esterel~v7 industrial compiler to Verilog and VHDL.

In 2012, Michael Mendler fully justified our viewpoint by proving the
reciprocal : a cyclic circuit yields the expected result in a finite
time for all gates and wires delay \emph{if and only if} this result is
computable from the Boolean equations by using only constructive
Boolean logic~\cite{MendlerShipleBerry:ConstructiveCircuits}, a fact
conjectured for long by the second author; in our opinion, this definitely
shows that constructive circuits exactly correspond to Esterel's
constructive semantics.

\subsection{Summary of the final Esterel semantics chain}
\label{sec:esterel-semantics-chain}

Each of the following formal semantics we analyze in this book is presented
in Plotkin's SOS (Structural Operational Semantics) style, which is
based on logical rules that are very versatile and quite naturally
fitted to formal analysis with systems such as Coq. All of them but
the last one are detailed
in~\cite{Berry:ConstructiveSemanticsOfPureEsterel}.  The last
\emph{microstep semantics} is new and was created by the first author
to mimic the behavior of circuits in a SOS-style. It gets very close
to the level of synchronous circuits, which serve as the final
implementation of Esterel, but without the difficulty that circuits
are graphs, a structure more difficult to manipulate in Coq. All Coq
analysis and proofs in this paper are due to the first author.

The \emph{behavioral semantics}, introduced
in~\cite{Esterel:BerryCosserat}, defines logical constraints that
ensure reactivity, i.e., existence of a solution, but not determinism.
It is presented as a set of SOS rules, where each SOS transition
computes the output signals and transforms a program text into a new
program text ready for the next reaction.  This semantics is based on
natural constraints, but still too loose in the sense that it accepts
some programs that happen to be deterministic ``just by chance'', \ie
in a non-causal way such as for \id{P3} above. (A slightly different
version due to Tardieu~\cite{phd-tardieu} does ensure determinism, but
adding some technical complexity. We do not find it relevant for this
work, for which the logical semantics is just a historical starting step).
Nevertheless, it was enough to build first an interpreter and then the
first Esterel~v2 academic compiler~\cite{Esterel:These:Couronne} by
adding a partial dependency ordering on signals that ensures
determinism in quite strict a way.  This v2 compiler was based on the
transformation of a program into a finite automaton, following
Brzozowski derivative-based construction of finite automata from
regular expressions~\cite{BrzozowskiDerivatives}, later improved in
\cite{BerrySethi}.

The \emph{potentials-based semantics}
\cite{Esterel:BerryGonthier}, again due to G. Gonthier,
adds sufficient
conditions to accept many more programs, including cyclic ones, but
not all those that should be considered as correct --- which was a problem for users. The \emph{state
semantics}, also due to Gonthier, refines it by replacing rewriting
the full program text done in the potential-based semantics by simply
moving simple state markers in the program's source code. It led to a much
faster implementation with the Esterel v4 compiler, the first one to be
industrialized. But users rightly complained that when a program was
rejected, the compiler could not tell whether it was the program's
insufficiency or a real programming error.
These semantics are now obsolete.

The final solution came with the \emph{constructive semantics}
\cite{Berry:ConstructiveSemanticsOfPureEsterel}, which refines the logical
semantics by imposing proof-theoretic constructivity constraints that
ensure a causal and disciplined flow of information in the program;
this implies determinism in a natural way.  It is presented as a
richer set of SOS rules, and has become the true reference semantics
of the language. It was similarly joined with a \emph{state constructive semantics} that moves markers in the source code.

The \emph{constructive circuit semantics} goes further by transforming
a Pure~Esterel program into a flat Boolean circuit (i.e. system of
equations) with the very same behavior, provided one uses constructive
logic when evaluating the circuit's equations (no excluded middle). As
said before, using constructive Boolean logic is equivalent to letting
electricity find the right solution even in the presence of cycles, as
proved in \cite{MendlerShipleBerry:ConstructiveCircuits}. The state
constructive semantics is the basis of the industrial v5 and v7
compilers~\cite{EsterelV7} to hardware circuit designs and software code, by
implementing the circuits's Boolean equations in C for software.  Furthermore,
the circuit and C code can be heavily optimized using modern circuit
CAD tools.

Finally, the \emph{microstep semantics} presented here and due to the
first author precisely describes the fine-grain
step-by-step propagation of information in an Esterel program during
one reaction, according to the current state computed from the last
reaction and the current input.
A first version was published in~\cite{CompilingEsterel}, but
the version presented here is technically much closer to the circuit semantics.
It is presented in a SOS-style amenable to Coq proofs,
while circuits are graphs, yet harder to manipulate in Coq.

This paper only deals with the behavioral, constructive, state
constructive, and new microstep semantics. We prove the expected
refinement or equivalence relations between these semantics, but only
on Kernel Esterel without loops, since loops create a schizophrenia
problem (see Section~\ref{sec:reincarnation}) that was solved in a
way we have not yet transported to Coq. This work is enough to deduce
a correct-by-construction interpreter for loop-free Kernel Esterel
(actually one for each semantics). But going to an efficient compiler
will require handling the circuit semantics, which will be the goal of
a subsequent paper.

\subsection{Related work}

After 1985, research inspired by Esterel took place for other
synchronous languages in several labs: Reactive~C~\cite{Reactive-C} by
F. Boussinot, in the same lab as Esterel; Argos~\cite{Argos} by
F. Maraninchi's, a graphical formalism developed in the Lustre group
and inspired by both Esterel and D. Harel’s Statecharts
(which were not quite synchronous in our sense); SyncChart by C. André
at Nice University, an Argos-inspired graphical language that
generated Esterel code~\cite{SyncCharts}; Reactive
ML~\cite{MandelPouzet-PPDP-2005} by L. Mandel and M. Pouzet at Ecole
Normale Supérieure Paris, which adds Esterel-like statements to the
functional OCaml language; Quartz~\cite{Schneider:Quartz} by
K. Schneider's team in Germany, now at the core of
the Averest System\footnote{https;//www.averest.org.};
SCL~\cite{SCL} and
SCCharts~\cite{vonHanxledenDM+14} by R. von Hanxleden and M. Mendler
also in Germany; finally, some Ptolemy II domains~\cite{PtolemyII} and
more recently Lingua Franca~\cite{LinguaFranca} by E.~A.~Lee and his
team at Berkeley University.

Considering that synchronous languages target in particular
safety-critical systems, formal verification has been carried out
since their beginning.  The first formal verification efforts were
targeted toward verifying program properties, to ensure the absence of
bugs in the program but not in the compiler. For instance,
verification of properties of Esterel programs was done with the
Auto/Autograph verifier~\cite{AutoAutograph}, and for Lustre programs
with the Kind2 model-checker~\cite{Kind2};
the Averest toolset for
Quartz supports formal verification with BDD and SAT techniques.

Verification of the compilation from Signal to C was done by Van Chan Ngô \emph{et al.}~\cite{Signal-translation-validation-1,
Signal-translation-validation-2} using translation validation, whereas
the compilation from Lustre to C was verified by Bourke \emph{et
al.}~\cite{Bourke-BDLPR-2017} using a direct proof.  Closer to our
purpose, the first preliminary attempt~\cite{kaplan2000} at
formalizing the Esterel v5 compiler in Coq only considered the first
instant.  It also disregarded loops because of schizophrenia, just
like we do here (see future work in Section~\ref{sec:reincarnation}).
Nevertheless, it was enough to uncover bugs in the initial constructive
semantics attempt.  A
more complete formal proof of a compiler of the  Quartz language
to circuits was done in the proof assistant HOL4 by Schneider \emph{et
al.}~\cite{Schneider:Quartz}.  Their semantics is not defined
in Plotkin's SOS-style used by Esterel but is
instead based on logical predicates, which are closer to automata
and circuits.  They also handle data, which we do not.

Two microstep semantics have previously been defined for Esterel in the literature.
The first one is in the ``Compiling Esterel'' book~\cite{CompilingEsterel}.
It is slightly higher-level than our microstep semantics as it only propagates control and does not perform a detailed computation of completion codes.
It handles data but still uses the \CanName{} function for the local signal rules.
Overall, it may be more suited for reasoning and explaining reactions but it is less suited to a translation to digital circuits.
The other one~\cite{Quartz-microstep} is defined on the Quartz language, a variant of Esterel.
It relates the original Quartz semantics given in term of logical predicates to a new SOS semantics for Quartz, in the spirit of the SOS semantics of Esterel.
Its objective is to explain the reactions of Quartz programs, not to be closer to the circuit semantics.
Considering its focus and level of details, we do not consider it to be a microstep semantics, as it is much closer to the state semantics of Section~\ref{sec:state-semantics} than to our microstep semantics or the one in~\cite{CompilingEsterel}.

\subsection{The contents of this paper}

This paper proves the correctness of the Esterel semantics chain
presented in~\cite{Berry:ConstructiveSemanticsOfPureEsterel} and described above
for
the loop-free case, using the
Coq proof assistant.  This is essentially the work of the first
author.  After describing Kernel Esterel and its informal semantics in
Section~\ref{sec:kernel-esterel} and the Coq proof assistant and
high-level features of the formalization in
Section~\ref{sec:coq-repr}, the following semantics are described:
first the behavioral semantics, namely the logical one
(Section~\ref{sec:logical-semantics}) and the constructive one
(Section~\ref{sec:constructive-semantics}); then the constructive
state semantics (Section~\ref{sec:state-semantics}), and finally a new
fine-grained microstep semantics
(Section~\ref{sec:microstep-semantics}).
We also prove the relations between these semantics (simulation, equivalence), see Figure~\ref{fig:esterel-semantics-chain}.
Finally, we detail some aspects of the Coq formalization in Section~\ref{sec:coq-proofs} before concluding in Section~\ref{sec:conclusion}.

The final operational microstep semantics presented at the end of the
chain is not yet the original circuit semantics, but it is technically
very close to it.
As said before, the difference lies in the difficulty of dealing with
graph manipulations with Coq.
This very last step remains to be completed.
Once this final simulation to the circuit semantics is completed, this
chain of simulations will prove that the translation of Esterel
programs to circuit
of~\cite{Berry:ConstructiveSemanticsOfPureEsterel}, that is, the
compiler, is correct: it preserves the semantics of Esterel programs
down to the generated digital circuits.
Nevertheless, this paper is a major step towards the final
justification and publication of the
book~\cite{Berry:ConstructiveSemanticsOfPureEsterel} with formal
proofs included as Coq programs executable by anybody who wants to
check them.

\begin{figure}
  \begin{center}
    \begin{tikzpicture}
      \node (LBS) at (-3, 0) {LBS};
      \node (CBS) at (0, 0) {CBS};
      \node (CSS) at (3.5, 0) {CSS};
      \node (micro) at (8, 0) {micro};
      \draw[thick, <->] (CBS) -- (CSS)
        node[midway, above] {\href{{\baseNetAddress}Proofs.CBS_CSS.html#sCSS_CBS}{\texttt{s{\small /}rCSS\_CBS}}}
        node[midway, below] {\href{{\baseNetAddress}Proofs.CBS_CSS.html#CBS_sCSS}{\texttt{CBS\_s{\small /}rCSS}}};
      \draw[thick, <-] (LBS) -- (CBS)
        node[midway, above] {\href{{\baseNetAddress}Proofs.CBS_LBS.html#CBS_LBS}{\texttt{CBS\_LBS}}};
      \draw[thick, ->] (CSS) -- (micro)
        node[midway, above] {\href{{\baseNetAddress}Proofs.CSS_Micro.html#sCSS_microsteps}{\texttt{s{\small /}rCSS\_microsteps}}};
    \end{tikzpicture}
  \end{center}
  \caption{The chain of Kernel Esterel semantics. Arrows represent simulation.}
  \label{fig:esterel-semantics-chain}
\end{figure}



\section{Kernel Esterel} 
\label{sec:kernel-esterel}

\subsection{Signals, instants, and synchrony}

Kernel Esterel (also called Pure Esterel) is a concurrent reactive
language that deals with \emph{pure signals} denoted by identifiers,
abbreviated into \emph{signals} in this paper.
An \emph{event} is a set of present signals.
A Kernel
Esterel program $P$ is defined by an \emph{interface}, which defines
the set $\cal I$ of its \emph{input signals} and the set $\cal O$ of
its \emph{output} signals, and a \emph{body} which is an executable
statement that may declare locally scoped signals.

The execution of a Kernel Esterel program consists of successive
\emph{reactions} to \emph{input events} $E$ that generate \emph{output
events} $E'$ according to the behavior of the body. The
\emph{synchrony hypothesis} stipulates that reactions have
conceptually no duration: inputs do not vary and outputs are produced
instantly in a deterministic way, local signals bearing a unique
status within their scope. To stress that point, we call reactions
\emph{instants}. The synchrony hypothesis is the key to reconciling
concurrency and determinism.
 
At each instant,
each signal \id{S} carries a unique \emph{present} or
\emph{absent} \emph{status}.  A signal \id{S} is set \emph{present} either if
it is an input set present by the program environment for the current instant
or if it is a local or output signal emitted by some ``\id{emit~S}''
statement executed during the current instant. If it is neither a \emph{present}
input nor emitted by some ``\id{emit~S}'' executed statement, 
a signal \id{S} is declared \emph{absent} for the instant; thus, a pure
signal behaves as a Boolean hardware \emph{or} gate that takes value 1
when any of its inputs has value 1 or value 0 where all its inputs
take value 0.
Each signal status in the input event is instantaneously broadcast to
all active statements, which means that all of them see the same
status for the signal at this instant, be they in sequence, in
parallel, or embedded in any active statement anywhere in the program.

The status of a signal can be tested by a
``\Tif{S}{$p$}{$q$}'' statement (the keyword is
\id{present} in \cite{Berry:ConstructiveSemanticsOfPureEsterel} and
Esterel~v5, but we use \id{if} here as in Esterel~v7). In this paper,
we shall also allow test for signal absence, which simply exchanges
the \id{if} branches; this is not indispensable but simplifies the
presentation.

In actual implementations, instants are externally determined by the
environment when it sends input; what is important is that the real time of a
reaction is sufficiently small w.r.t. the timing constraints of the
controlled process for the semantics to be preserved. Implementations
may rely on atomic calls of a C function implementing the global
program behavior (as for Esterel v5), on circuit clock cycles for
translation to hardware (as for Esterel v7), on execution of
atomic-by-construction JavaScript code (as for HipHop), or even on
appropriately synchronized execution of a distributed program provided
interference between I/O and execution is avoided, as in Lingua
Franca~\cite{LinguaFranca}. Since they do not interfere with
the semantics, these practical implementations are outside the scope
of this paper.

\begin{remark}
  In Full Esterel, signal can carry data values that are not available
  in the kernel.
  This is not really a limitation: the development of Esterel
  semantics and compilers has shown that dealing with shared data
  values can be done in much the same way as dealing with the pure
  signals that carry them~\cite{CompilingEsterel}.
  The only difference is that knowing the value of a valued signal
  requires the resolution of all emitters to combine the individual
  values emitted by the active emitters using a user-specified
  associative and commutative function.
  This means that the value of a valued signal depends on the status
  of all emitters, which is exactly the same thing as for determining
  the absence of a pure signal.
  Of course, one also need to add data dependencies, but this is as
  for any classical language.
  Nevertheless, as said before, we have not yet handled valued signals
  in the Coq proof.
\end{remark}

\subsection{Kernel statements}
\label{sec:kernel-statements}

Kernel Esterel contains a small number of statements, from
which one can easily define the richer statement set of the
user-friendly full language~\cite{EsterelV591,EsterelV7}.
A statement starts at some instant and may execute either instantly,
\ie entirely within the instant, or up to some further instant where reactions
produce empty results from then on, or even indefinitely.
Its starting and ending (if any) instants define its \emph{lifetime}.

The statements can be presented in two equivalent forms: with
keywords, which make reading easier, or with mathematical symbols,
a much shorter writing for semantic rules and proofs. We use the
latter symbolic form in all technical developments.  Here are the
kernel textual statements, where~$s$ is a signal,~$T$ is a trap name,~$k$ is an
integer and $p$, $q$ are kernel statement. The statement order is not
the same as in the original Esterel papers, but it will be technically
more convenient here:

\setlength{\arraycolsep}{1cm} \[
\begin{array}{ll} \textbf{Textual form} & \textbf{Symbolic form} \\
\Tnothing & \Snothing \\
\Tpause & \Spause \\
\Temit s & \Semit s \\
\Tawimm s & \Sawimm s \\
\Tif s p q & \Sif s p q \\
\marginlink[\#stmt]{Esterel.Definitions}{0mm}
\Tsuspend p s & \Ssuspend s p \\
\Ttrap T p & \Strap p \\
\Texit {T^k} & \Sexit k \qquad\qquad (\text{with } k\ge 2) \\
& \Sshift p \\
\Tsequence p q & \Ssequence p q \\
\Tloop p & \Sloop p \\
\Tparallel p q & \Sparallel p q \\
\Tsignaldecl s p & \Ssignaldecl s p
\end{array}
\]
We slightly depart
from~\cite{Berry:ConstructiveSemanticsOfPureEsterel} in two ways:
first, we rename the \id{present} test for signals of
\cite{EsterelV591} and Esterel~v5 into $\Tif{\_}{\_}{\_}$ as in Esterel
v7~\cite{EsterelV7}; second, we add an ``\id{await~immediate}'' wait
statement to the kernel, written as $\Sawimm s$, which can be defined
from the other primitives as follows:
\begin{newprogram}
trap T in
   loop
      if $s$ then exit T else pause end
   end loop
end
\end{newprogram}%
The reason to include ``\Tawimm{$s$}'' here if that the correctness
proof of the microstep semantics
(Section~\ref{sec:microstep-semantics}) does not handle loops yet,
because they lead to the signal instantaneous reincarnation problem,
explained and solved in the last chapter
of~\cite{Berry:ConstructiveSemanticsOfPureEsterel}.  But, in order to
express suspension ``$\Tsuspend s p$'' in the constructive semantics
(Section~\ref{sec:constructive-semantics}), we need ``\Tawimm{$s$}''
to wait until~$s$ becomes absent before performing another step
of~$p$, current instant included, which is what \id{await~immediate} does
(by default, ``\Tawait{s}'' ignore $s$  when it starts).

\subsection{Completion codes and the trap encoding}
\label{sec:completion codes}

The \id{trap}-\id{exit} lexically scoped statement is unique to
Esterel (although a fairly similar construct exists in David Harel's
graphical Statecharts, but with a different semantics) and need more explanations.
It is
akin to an explicit user-driven error-handling statement in a
sequential language, as for \id{try}--\id{throw} in Java or
\id{try}--\id{except} in Python, but used positively to structure
programs. Furthermore, it is fully compatible with synchronous
concurrency and fully deterministic: executing ``\id{exit~T}''
anywhere means asking for terminating the body of the ``\id{trap~T}''
statement as soon as this body has been completely evaluated in the
instant, even if it has many parallel components, but respecting a
scope order for traps if the body is concurrent: if several traps are
simultaneously exited by concurrent statements, only the outermost one
matters and the other ones are discarded. If no trap is exited, the trap
terminates or pauses as its body does. Such a statement is a key
for constructing derived statement from the kernel and programming
large applications; it cannot really exist in asynchronous concurrent
languages.

In the symbolic forms, this seemingly complex behavior and more
generally the whole control propagation is realized using a simple
\emph{completion code} integer encoding due to Gonthier and used in
all compilers.  Code $0$ means termination, thus \id{nothing} simply
becomes code $0$ in the symbolic form; code $1$ means pausing for the
instant and waiting to be resumed at the next instant, thus \id{pause}
becomes $1$. A trap is just a marker, written $\{.\}$, and an
``\id{exit~T}'' becomes code $k$ with $k \ge 2$, which encodes exiting
the enclosing \id{trap} reached by traversing $k-2$ nested traps on
the way. This encoding greatly simplifies the semantic rules.
Although this is not necessary, we add the code as an exponent to the
trap name in our textual \id{trap}--\id{exit} examples for more clarity.

Here is how this encoding works: at each instant, each executed
statement returns such a completion code, and the composition of these
codes determines the control flow of the program in a deterministic
way. Concurrent traps are handled in the following way: each parallel statement
waits for the completion codes returned by all branches and returns as
its code the \emph{maximum} of these codes, discarding lower ones.
This means that the parallel terminates if and only if all
its branches terminate, pauses if at least one branch pauses and no
branch exits a trap, and exits the outermost trap exited by
branches if at least one exits a trap, termination, pausing, or
other traps exits being ignored. Determinism is therefore enforced. 

For example, consider the statement

\begin{newprogram}
trap T in
   exit T$^2$;
   trap U in
      exit T$^3$
   ||
      exit U$^2$
   end
end
\end{newprogram}
where ``\id{exit~T}'' is successively decorated by 2 and 3
because the second one is enclosed in ``\id{trap~U}''.
This textual code simply becomes $\Strap{\Ssequence{\Sexit{2}}{\Strap{\Spar{\Sexit{3}}{\Sexit{2}}}}}$ in the symbolic form.

\begin{remark}
  This encoding of traps is akin to the De Bruijn encoding of bound
  variables in the $\lambda$-calculus~\cite{DeBruijn:Indices}, but in
  a concurrent setting and adding 2 to account for termination (0) and
  pausing (1).
\end{remark}

\subsection{A running example}
\label{sec:running example}

Here a slight variant called \id{ABROi} of the \id{ABRO} program,
which can be considered as the ``Hello World!'' of Esterel that starts
most Esterel
descriptions~\cite{Berry:ConstructiveSemanticsOfPureEsterel}. It
features the most striking specificities of Kernel Esterel:
deterministic parallelism, instantaneous reactions to presence/absence
of signals, simultaneous reception of signals (and emission, not illustrated here),
and using a trap statement capable of instantaneously canceling
parallel activities. The reader can check that its behavior is much
harder to write in sequential languages (including Javascript), and
even with asynchronous threads or asynchronous parallel languages.

\newcommand\myspace[1]{\\ \mbox{~} \hspace{#1em}}
\begin{example}[ABROi program] \label{ex:ABROi}

The \id{ABROi} program
involves three input signals \id{A}, \id{B}, and \id{R} and one output
signal~\id{O}.  Its behavior can be specified as follows: as soon as
signals~\id{A} and~\id{B} have been both received, either in succession or
simultaneously, emit~\id{O} and do nothing from then on; restart the same
behavior afresh whenever~\id{R} is received.

Here is \id{ABROi} in plain Esterel:

\begin{newprogram}
loop
    [ await immediate A || await immediate B ];
    emit O
each R
\end{newprogram}%
The difference with classical \id{ABRO} is that
\id{O} is also emitted if \id{A} and \id{B} are simultaneously received when \id{R} occurs. This behavior would be obtained
by adding a \id{pause} in sequence right before the parallel.

Expanding with kernel statements the \id{loop}-{each~R} loop that
restarts its body each time \id{R} occurs, as first done by compilers, one
gets:

\begin{newprogram}
loop
   trap T in
      loop
         pause;
         if R then exit T$^2$ else pause end
      end
   ||
      [ await immediate A || await immediate B ];
       emit O;
      loop pause end;
      exit T$^2$
   end
end
\end{newprogram}%
The last ``$\id{exit}~\id{T}^2$'' is generated by the macro-expansion
of $\id{loop}-\id{each}$ but is unreachable here,
which is detected by compilers. We keep it in the sequel anyway.
The much more compact symbolic form will be useful for semantics rules:
$    \Sloop{
      \Strap{
        \Sparallel{\big(
          \Ssequence{\Spause}
                    {\Sloop{(\Sif{R}{\Sexit{2}}{\Spause})}}\big)}
                  {\big(\Ssequence{
                      \Ssequence{
                        \Ssequence{(
                          \Sparallel{\Sawimm{A}}
                                    {\Sawimm{B}}
                          \,)}
                          {\Semit{O}}
                        }
                        {(\Sloop{\Spause})}
                     }
                     {\Sexit{2}}
                  \big)}
      }
    }$

\noindent
noting that $(\Sloop{\Spause})$ is simply called \id{halt} in Esterel.
The reader is referred to~\cite{EsterelV591, Berry:ConstructiveSemanticsOfPureEsterel} for additional examples
\end{example}

\subsection{Intuitive Semantics of Kernel Esterel}

The intuitive semantics is defined by cases over the statements:
\begin{itemize}
\item
  The ``\Tnothing'' or $\Snothing$ statement instantly terminates, returning  completion code~$0$.
\item
  The ``\Tpause'' or $\Spause$ statement waits for the next instant: at
  its starting instant, it stops control propagation and returns
  completion code $1$; at its next execution instant, it
  terminates and returns code $0$. This is not necessarily the instant that
  follows the starting instant since the $\Tpause$ statement could be
  preempted or suspended.
\item
  At a given instant, the status of a signal $s$ is shared by all
  program components within its scope. By default $s$ is absent in a reaction.
  Any executed
  ``$\Temit s$'' or $\Semit s$ statement sets $s$ present for the instant.
  The $\Temit{}$ statement terminates instantly.
\item
  The ``\Tawimm{$s$}'' statement blocks control propagation until~$s$ is present.
  At every instant including the starting one, it terminates instantly if~$s$
  is present, or if pauses and continues waiting if~$s$ is absent.
\item
  The presence test statement ``$\Tif s p q$'' or $\Sif s p q$ instantly tests the
  status of $s$. According to the presence/absence of $s$, it selects $p$ or $q$ for
  immediate execution and behaves as it from then on. Note that the
  test is only performed at the instant in which the statement is started.
\item
  The delayed suspension ``$\Tsuspend s p$'' or $\Ssuspend s p$ statement behaves as~$p$ in its starting instant.
  In all subsequent instants during the lifetime of~$p$, it executes~$p$ whenever~$s$ is absent in the instant
  or freezes the state of~$p$ until the next instant if~$s$ is present.
  Notice that the status of~$s$ is tested at all instants except for the starting one.
  The completion code at first instant is that of $p$; for subsequent instants (if any)
  it is that of~$p$ if~$s$ is absent or~$1$ if~$s$ is present.
  Thus, termination and trap exits of~$p$ are only propagated at the first instant or if~$s$ is absent.
\item
  A loop ``\Tloop{$p$}'' or $\Sloop p$ statement instantly restarts its body~$p$ when this body terminates, and it
  propagates traps. The body $p$ is not allowed to terminate instantly when started.
  Notice that traps are the only way to exit loops (the $\it{abort}$
  statement of the full language is definable in the kernel language using traps).
\item
At each instant, the ``$\Ttrap {T} p$'' or $\Strap p$ statement executes $p$ for the instant; it
terminates or pauses if $p$ does, terminates if $p$ returns completion
code $2$ (\ie catches its exits), or returns completion
code $k-1$ if the completion code of $p$ is $k > 2$ (thus propagating
exits to the appropriate outer trap).
\item
  The textual ``\Texit {$T^k$}'' or symbolic $\Sexit k$ statement with $k\ge 2$ simply
  returns completion code $k$.
\item
  The $\Sshift p$ symbolic statement is necessary to define the macro-based
  full-language statements that place $p$ in a ``$\Strap{\_}$'' trap context.
  It simply adds $1$ to any trap completion code
  $k \ge 2$ returned by $p$. It is not needed in the textual
  form since traps are named.
\item
  For a sequence ``$\Ssequence p q$'', the statement~$q$ instantly starts if and when~$p$
  terminates.
  Trap exits by~$p$ or~$q$ are propagated. Note that both $p$ and $q$ are executed in the instant when $p$ terminates.
\item
  The textual ``$\Tparallel p q$'' or symbolic $\Sparallel p q$ parallel statement terminates
  instantly as soon as both branches have terminated; it pauses if at
  least one branch pauses and no branch raises an exit; it exits
  a trap instantly if one of its branches does; if several
  branches concurrently exit traps, it only exits the outermost exited
  trap.  In the symbolic formal
  semantics, this behavior simply reduces to the parallel statement
  $\Sparallel p q$ returning as completion code the maximum
  $\max(k_p,k_q)$ of the completion codes of its branches at each
  instant.
\item
  Finally, a local signal declaration ``\Tsignaldecl{$s$}{$p$}'' or $\Ssignaldecl s p$ declares a signal~$s$ local to its
  body~$p$, with the usual lexical binding and shadowing.
\end{itemize}

\begin{example}[Intuitive execution of the ABROi program]
  Consider the execution of the ABROi program (Example~\ref{ex:ABROi})
  during five consecutive instants where the following inputs are received, recalling that \id{halt} is simply {1*}:
  \begin{enumerate}
    \item $\{\id{B}\}$: only \id{B} is received, so ABROi keeps waiting on \id{A};
    \item $\{\id{A,B}\}$: as \id{A} is received, \id{O} is instantly emitted and ABROi reaches the \id{halt} statement --- a natural name for the ``\id{loop~pause~end}'' statement in the code of ABROi;
    \item $\{\id{B}\}$: the \id{halt} statement stalls execution, as the outer loop has not yet been restarted; \id{B} is ignored;
    \item $\{\id{R}\}$: receiving~\id{R} kills the \id{halt} statement and restarts the loop so that ABROi again waits for~\id{A},~\id{B}, and \id{R};
    \item $\{\id{A,B,R}\}$: receiving~\id{R} kills and restarts the whole body of the loop, but \id{O} is emitted since \id{A} and \id{B} are both present.
  \end{enumerate}
\end{example}




\section{Esterel semantics in Coq}
\label{sec:coq-repr}

\subsection{The Coq proof assistant}
\label{sec:coq}

A proof assistant is software that fully automates the verification of
proofs, thus building increased confidence in their correctness.  This
is of particular importance and interest for critical system software
as well as domains where proofs are notoriously difficult, such as
distributed systems.  Unlike for other formal methods such as model
checking, writing a proof is not fully automated but assistance is
usually provided in the form of tactics (``proofs steps'') and
decision procedures for some domains.  Proof assistants are also a lot
more expressive than more automated techniques and can readily work
with arbitrarily complex mathematical formulas.

Quoting the description on the Coq website (hyperlinks replaced with footnotes):
\begin{quote}
  Coq is a formal proof management system. It provides a formal language to write mathematical definitions, executable algorithms and theorems together with an environment for semi-interactive development of machine-checked proofs.
  Typical applications include the certification of properties of programming languages\footnote{\url{https://coq.inria.fr/cocorico/List\%20of\%20Coq\%20PL\%20Projects}} (e.g. the CompCert\footnote{\url{http://compcert.inria.fr/}} compiler certification project, the Verified Software Toolchain\footnote{\url{http://vst.cs.princeton.edu/}} for verification of C programs, or the Iris\footnote{\url{https://iris-project.org/}} framework for concurrent separation logic), the formalization of mathematics\footnote{\url{https://coq.inria.fr/cocorico/List\%20of\%20Coq\%20Math\%20Projects}} (e.g. the full formalization of the Feit-Thompson theorem\footnote{\url{https://hal.inria.fr/hal-00816699}}, or homotopy type theory\footnote{\url{http://homotopytypetheory.org/coq/}}), and teaching\footnote{\url{https://coq.inria.fr/cocorico/CoqInTheClassroom}}.
\end{quote}

The logical foundation of the Coq proof assistant is the Calculus of Inductive Construction~\cite{coq-paulin}, a powerful extension of the higher-order typed $\lambda$-calculus in which types and logical formulas are unified, so that writing a proof amounts to providing a $\lambda$-term of a given type.

The computational part of these terms can be extracted to functional programming languages (currently, Haskell, OCaml, and Scheme) and can be integrated as verified components in bigger programs.
For instance, the CompCert compiler can be extracted from its Coq code, thus ensuring that the executed code is indeed the verified one.
Note that the compiler of the functional language and its runtime environment must be trusted.

The trust one may have in proofs accepted by such software is limited by the trust in the software itself.
In order to reduce the trusted code base, Coq and other proof assistants are built around the so-called \emph{de Bruijn principle}, in which only a reasonably small kernel checking proof validity must be trusted, while the rest (tactics, automation, etc.) need not be.
Other ways to increase confidence in the proof assistant are to provide an independent proof checker or even to aim at building formal proofs in the tool itself\footnote{Notice that this only increases confidence and is not a definitive answer because of Gödel's second incompleteness theorem.} of the correctness of its logic and its software implementation~\cite{coqcoqcorrect}.

\subsection{The formalization of Kernel Esterel in Coq}

\subsubsection*{Restriction to loop-free Esterel}
The Coq formalization is currently restricted to Kernel Esterel without the looping construct $\Sloop p$.
Indeed, loops can interact with local signals in a subtle way, requiring a specific treatment for the microstep semantics of Section~\ref{sec:microstep-semantics} which we do not handle yet, see \cite[chap.~12]{Berry:ConstructiveSemanticsOfPureEsterel} and Section~\ref{sec:reincarnation} for details.
In order to keep a coherent body of proofs, loops were entirely removed from the whole Coq formalization, although they are easy to handle in the other semantics.

\subsubsection*{High-level view of the formalization}
The syntax of Kernel Esterel is represented in Coq as an inductive
type, making it possible to perform proofs by induction on the syntax
tree of a program.
Similarly, each of our SOS semantics will be represented by an
inductive type, making it possible to perform proof by induction on
derivation trees.
In this setting, proving simulation (\resp equivalence) between two
semantics amounts to proving that the existence of a derivation tree
in the source semantics implies (\resp is equivalent to) the existence
of a derivation tree in the target semantics.
Therefore, a typical proof goes by induction on the derivation tree
in the source semantics and builds an adequate derivation tree in the
target semantics.

The full Coq code for the syntax, semantics, and proofs presented in this paper is available as supplementary material to this article.\footnote{Note to reviewers: the link will be added later but you should already have a copy of this material.}

A fair number of proofs have no technical difficulty and rely mostly on a straightforward induction.
They will not be detailed here, where we will only concentrate on the interesting points.
We refer the interested reader to the Coq development.
This paper contains links to the html documentation generated from the Coq development, permitting a more comfortable browsing of the formalization.
If you put the file corresponding to this paper in the root directory of the Coq development, you will be able to directly use these links, represented as \linkSymbol in the margin; to access the relevant documentation.

\subsection{Representation choices in Coq}
\label{sec:coq-representation}

The Coq formalization tries to be faithful to the symbolic notation of Kernel Esterel, in particular to~\cite{Berry:ConstructiveSemanticsOfPureEsterel}, but there are still two syntactic differences.
First, the parallel composition $\Sparallel p q$ is written with the textual notation $p \mathbin{||} q$ for technical reasons: the single pipe~$|$ is used in Coq for pattern matching.
Second, borrowing standard notation for lists, extending an event~$E$ by mapping the signal~$s$ to status~$b$ (with possible overshadowing of a previous signal~$s$) is written $s^b \; +\!+ \; E$ rather than $\addEvent s b E$. 

There are also less superficial changes compared to the literature that we describe now.

\paragraph*{Inputs and outputs events}
Signal inputs and outputs are described with events (written $\cal I$ and $\cal O$ respectively) that may range over different sets of signals. (We usually have $\cal O \subseteq \cal I$ to represent the fact that any emitted signal can be instantaneously read.)
In this paper, inputs and outputs events~$E$ and~$E'$ are maps from the set of signals in scope to statuses, which makes it possible to prove that the set of visible signals is not modified during execution (see property ``Domain invariance'' in Section~\ref{sec:Esterel-semantics-properties}).
This is different from the presentation of this paper and of~\cite{Berry:ConstructiveSemanticsOfPureEsterel} where~$E'$ is a set of emitted signals, which amounts to the set of signals mapped to~$+$ in our representation.

\paragraph*{Setoid of signals}
Signals are represented by a setoid, that is, a type equipped with a dedicated equality.
In particular, signals do not use the equality of Coq and may use any coarser equivalence.
The downside of this choice is that every semantics must have an extra rule \texttt{compat} ensuring compatibility with this equivalence, that is, a rule stating that one can replace a signal by any equivalent one, both in programs and in events.
This is the last rule in all the semantics in the formalization.
Furthermore, the \tactic{inversion} tactic is no longer useful as this compatibility case always applies.
Nevertheless, it is rather straightforward to prove dedicated inversion lemmas and write a tactic mimicking \tactic{inversion}.
The benefit of this choice is that we can use extensional equality for events and we can optimize signal expressions (once we introduce them in the formalization): instead of structural equality, we can choose propositional equality, so that for instance $s \band s'$ becomes equal to $s' \band s$.
Another solution would have been to use Coq equality for everything (signals, events, statements and so on) and either add axioms ensuring that this equality is extensional on events or to use an implementation of events featuring a unique canonical representative (for instance, ordered associative lists without duplicate keys).

\paragraph*{Delta function}
Whenever the completion code~$k$ of a statement is not~$1$, there is nothing left to execute, either because we finished execution ($k = 0$) or because a trap killed the remaining state ($k \geq 2$), hence the derivative should be~$\Snothing$.
The usual rules for Esterel in the literature do not have this property: for instance the trap rule from Figure~\ref{fig:LBS} always keep the trap.
This makes rules easier to write and read but is sometimes inconvenient.
Although we stick to this tradition in this paper, the Coq formalization performs this normalization of the derivative to~$\Snothing$ whenever $k \neq 1$
by introducing a function~$\delta$ defined as follows:
\marginlink[\#delta_stmt]{Esterel.Util.SemanticsCommon}{8mm}
\[
  \deltafun k p := \begin{cases}
    p & \text{if } k = 1 \\
    \Snothing & \text{otherwise} \\
  \end{cases}
\]
Then we replace each derivative~$p'$ with~$\deltafun k {p'}$
(some rules do not require this, namely the ones for $\Sk k$, $\Semit s$, $\Sif s p q$, $\Sawimm s$, and $\Ssequence p q$ when $k_p = 0$).
Although it is not mandatory, this normalization technically and conceptually simplifies the calculations by setting all terminated terms to $\Snothing$ which makes termination checks simpler.
The same choice is made for the state semantics, the only difference being that the derivative is not $\Snothing$ but the inert form of the state.

\paragraph*{Well-formed programs\marginlink[\#valid_dom]{Esterel.Util.SemanticsCommon}{-2mm}}
\label{sec:valid-dom}

When writing a Kernel Esterel program, there are implicit well-formation rules.
For instance, it is not allowed to emit or read a signal that is not in scope.
In compiler implementations, this is ensured by typing checks.
Here, all semantic rules using signals ensure this property with a premise $s ∈ E$ which implies that $s$ is in scope.

In Coq, we also define the $\VD E p$ predicate expressing that all signals free in~$p$ are in the domain of~$E$, thus making sure that the evaluation of~$p$ is properly defined.
The definition is made by a straightforward case analysis and amounts to requiring $s ∈ E$ every time a signal is explicitly used, that is, for statements $\Semit s$, $\Sawimm s$, $\Ssuspend s p$, and $\Sif s p q$, but not for local signal declaration in which it is added to the domain of~$E$ instead.
See the Coq code for details.

\paragraph*{Other remarks related to the Coq formalization}
Other remarks or specificities of the Coq formalization will be signaled in the paper as \texttt{Coq Remark}s.
If the reader is not interested in the technical details of the formalization, they can be safely ignored.


\section{The logical semantics
\texorpdfstring{\marginlink{Esterel.Semantics.LBS}{0mm}}{}} 
\label{sec:logical-semantics}

The SOS-style logical (behavioral) semantics (LBS) defines constraints
that ensure reactivity of a program to a given set of inputs,
\ie absence of deadlock and presence/absence of all signals.
In the version given here, it does not guarantee determinism;
Tardieu~\cite{phd-tardieu} proposed a version that also guarantees determinism,
but we find it too heavy since it amounts to verifying that all non-deterministic
executions yield the same result, which may require considering an exponential number of executions with nested local signals.
Determinism will be guaranteed for much more compelling conceptual reasons
by the constructive semantics detailed in the next section.

The logical semantics defines reactions as behavioral transitions represented by SOS rules of the form
\[
  \LBS p E {E'} k {p'},
\]
where $E, E'$ are \emph{events}.
The event~$E$ denotes the \emph{inputs} of~$p$ while~$E'$ denotes its \emph{outputs}, that is, the signals emitted by~$p$. 
The integer~$k$ is the \emph{completion code} of~$p$ for the instant,
and~$p'$ is the \emph{derivative} of~$p$ by the transition, \ie the
statement replacing~$p$ for the next instant (this terminology is
inspired by the work of Brzozowski and Seger on translating regular
expressions to automata~\cite{Brzozowski-Seger}).
The rules are given in Figure~\ref{fig:LBS}.

\begin{figure}[tp]
  \begin{prooftree}
    \AXC{}
    \RightLabel{k}
    \UIC{$\LBS {\Sk{k}} E {∅} k \Snothing$}
    \DP
    \\[\rulevspace]

    \AXC{}
    \RightLabel{\Temit{}}
    \UIC{$\LBS{\Semit s} E {\SingletonEvent s} 0 \Snothing$}
    \DP
    \\[\rulevspace]

    \AXC{$s^+ ∈ E$}
    \RightLabel{awimm$^+$}
    \UIC{$\LBS{\Sawimm s} E {∅} 0 {\Snothing}$}
    \DP
    \hspace{\rulehspace}
    \RightLabel{awimm$^-$}
    \AXC{$s^- ∈ E$}
    \UIC{$\LBS{\Sawimm s} E {∅} 1 {\Sawimm s}$}
    \DP
    \\[\rulevspace]

    \AXC{$s^+ ∈ E$}
    \AXC{$\LBS p E {E'} k {p'}$}
    \RightLabel{then}
    \BIC{$\LBS{\Sif s p q} E {E'} k {p'}$}
    \DP
    \hspace{\rulehspace}
    \AXC{$s^- ∈ E$}
    \AXC{$\LBS q E {E'} k {q'}$}
    \RightLabel{else}
    \BIC{$\LBS{\Sif s p q} E {E'} k {q'}$}
    \DP
    \\[\rulevspace]

    \AXC{$\LBS p E {E'} k {p'}$}
    \RightLabel{suspend}
    \UIC{$\LBS{\Ssuspend s p} E {E'} k {\Ssequence{\Sawimm{\lnot s}}{\Ssuspend s {p'}}}$}
    \DP
    \\[\rulevspace]

    \AXC{$k \neq 0$}
    \AXC{$\LBS p E {E'} k {p'}$}
    \RightLabel{loop}
    \BIC{$\LBS{\Sloop p} E {E'} k {\Ssequence{p'}{\Sloop p}}$}
    \DP \\[\rulevspace]

    \AXC{$\LBS p E {E'} k {p'}$}
    \RightLabel{trap}
    \UIC{$\LBS{\Strap{\strut p}} E {E'}{↓k}{\Strap{p'}}$}
    \DP
    \hspace{\rulehspace}
    \AXC{$\LBS p E {E'} k {p'}$}
    \RightLabel{shift}
    \UIC{$\LBS{\Sshift p} E {E'}{\Kup k}{\Sshift{p'}}$}
    \DP
    \\[\rulevspace]

    \AXC{$k \neq 0$}
    \AXC{$\LBS p E {E'} k {p'}$}
    \RightLabel{seq$_k$}
    \BIC{$\LBS{\Ssequence p q} E {E'} k {\Ssequence{p'} q}$}
    \DP
    \hspace{\rulehspace}
    \AXC{$\LBS p E {E_p'} 0 {p'}$}
    \AXC{$\LBS q E {E_q'} k {q'}$}
    \RightLabel{seq$_0$}
    \BIC{$\LBS{\Ssequence p q} E {E_p' ∪ E_q'} k {q'}$}
    \DP
    \\[\rulevspace]

    \AXC{$\LBS p E {E_p'}{k_p}{p'}$}
    \AXC{$\LBS q E {E_q'}{k_q}{q'}$}
    \RightLabel{par}
    \BIC{$\LBS{\Sparallel p q} E {E_p' ∪ E_q'}{\max(k_p,k_q)}{\Sparallel{p'}{q'}}$}
    \DP \\[\rulevspace]

    \AXC{$\LBS p {\addEvent s + E}{E'} k {p'}$}
    \AXC{$s^+ ∈ E'$}
    \RightLabel{\SigP}
    \BIC{$\LBS{\Ssignaldecl s p} E {\restrictEvent{E'} s} k {\Ssignaldecl s {p'}}$}
    \DP
    \hspace{\rulehspace}
    \AXC{$\LBS p {\addEvent s - E}{E'} k {p'}$}
    \AXC{$s^+ \not\in E'$}
    \RightLabel{\SigM}
    \BIC{$\LBS{\Ssignaldecl s p} E {\restrictEvent{E'} s} k {\Ssignaldecl s {p'}}$}
  \end{prooftree}
  \caption{Logical (Behavioral) Semantics (LBS) rules.}
  \label{fig:LBS}
\end{figure}

Note that the rule $\LBS {\Sk{k}} E {∅} k \Snothing$ covers three
statements: first, the trivial termination of~$0$ ($\Tnothing$);
second, the pausing case for~$1$ ($\Tpause$) that returns completion
code~$1$ and has derivative~$0$ ($\Tnothing$) that will
terminate instantly at the next reaction; third, the~$k$ ($\Texit
{T^k}$) exit statement which returns completion code~$k$ and
has derivative~$\Snothing$ ($\Tnothing$).

Except for the rules {\SigP} and {\SigM} that deal with signal declaration, our rules are a straightforward implementation of the intuitive semantics.
The trap statement catches the innermost exit (the one with code 2) and turns it into termination (code 0).
This is the meaning of the $\Kdown{}$ function, with (pseudo-)inverse $\Kup{}$ used in shift.

\begin{minipage}{0.4\linewidth}
  \[
  \Kdown k :=
  \begin{cases}
    0 & \mapsto 0 \\
    1 & \mapsto 1 \\
    2 & \mapsto 0 \\
    n \quad (n \geq 3) & \mapsto n - 1 \\
  \end{cases}
  \]
\end{minipage}
\begin{minipage}{0.4\linewidth}
  \[
  \Kup k :=
  \begin{cases}
    0 & \mapsto 0 \\
    1 & \mapsto 1 \\
    n \quad (n \geq 2) & \mapsto n+1 \\
  \end{cases}
  \]
\end{minipage} \\
We have $\Kdown{(\Kup n)} = n$ for any integer $n$ but not $\Kup{(\Kdown n)} = n$ because $\Kup{(\Kdown 2)} = \Kup 0 = 0$.

Note an unusual aspect: an instantaneous but compound reaction is
intentionally defined by a single transition $\LBS p E {E'} k {p'}$, not
by a sequence of microsteps.
For instance, the sequence operator $\Ssequence p q$ chains execution
of its components in the same rule if~$p$ terminates (rule seq$_0$),
instead of chaining microsteps of the form $\LBS p E {E_p'} 0
{\Ssequence \Snothing q}$ followed by $\LBS q E {E_q'} k {q'}$ as
would a standard SOS semantics.
The union $E_p'∪ E_q'$ in the seq$_0$ rule conclusion directly
expresses the synchrony hypothesis within the reaction: control passes
from~$p$ to~$q$ in the very same reaction.
The same holds for the parallel statement (rule par).

The {\SigP} and {\SigM} rules for local signals first extend the input
event~$E$ with~$s$ mapped to a status~$b$ (possibly shadowing any
upper-scoped signal~$s$), written $\addEvent s b E$, then execute~$p$
using this new input event, and finally either restore the status of the shadowed signal~$s$ if any or remove~$s$ from the output event~$E'$, before returning it.
This restoration or removal is written $\restrictEvent{E'} s$.
The completion code and derivative statement comes from the execution of~$p$.
The side condition $s^+ \in E'$ or $s^+ \not\in E'$ ensures the correct status for~$s$ in $p$: in rule {\SigP} we assume that $s$ is received, so we check that it is indeed emitted ($s^+ \in E'$), and conversely for rule {\SigM} with~$s$ not received and not emitted.
This is called the \emph{(logical) coherence law}~\cite[chap.~3]{Berry:ConstructiveSemanticsOfPureEsterel}: a signal~$s$ is present in an instant if and only if an ``\Temit{$s$}'' statement is executed in this instant.

As an illustration of the LBS, the execution of the ABROi program is given in Figure~\ref{fig:ABROi-LBS}.

\begin{figure}
  \[
  \begin{array}{r@{\;}c@{\;}l}
    ABROi & = &
    \Sloop{
      \Strap{
        \Sparallel{\big(\Ssequence{\Spause}{\Sloop{(\Sif{R}{\Sexit{2}}{\Spause})}}\big)}
                  {\big(\Ssequence{
                      \Ssequence{
                        \Ssequence{(
                          \Sparallel{\Sawimm{A}}
                                    {\Sawimm{B}}
                          \,)}
                          {\Semit{O}}
                        }
                        {(\Sloop{\Spause})}
                     }
                     {\Sexit{2}}
                  \big)}
      }
    } \\
    & \LBS{}{\{ B \}}{\emptyset}{1}{} &
    \Ssequence{
      \Strap{
        \Sparallel{\big(\Ssequence{\Snothing}{\Sloop{(\Sif{R}{\Sexit{2}}{\Spause})}}\big)}
                  {\big(\Ssequence{
                      \Ssequence{
                        \Ssequence{(
                          \Sparallel{\Sawimm{A}}
                                    {\Snothing}
                          )}
                          {\Semit{O}}
                        }
                        {(\Sloop{\Spause})}
                    }
                    {\Sexit{2}}}
                  \big)}
      }{ABROi} \\
    & \LBS{}{\{ A, B \}}{\{ O \}}{1}{} &
    \Ssequence{
      \Strap{
        \Sparallel{\big(\Ssequence{\Snothing}{\Sloop{(\Sif{R}{\Sexit{2}}{\Spause})}}\big)}
                  {\big(\Ssequence{\Snothing}{\Ssequence{(\Sloop{\Spause})}
                               {\Sexit{2}}}\big)}
      }
    }{ABROi} \\
    & \LBS{}{\{ B \}}{\emptyset}{1}{} &
    \Ssequence{
      \Strap{
        \Sparallel{\big(\Ssequence{\Snothing}{\Sloop{(\Sif{R}{\Sexit{2}}{\Spause})}}\big)}
                  {\big(\Ssequence{\Snothing}{\Ssequence{(\Sloop{\Spause})}
                      {\Sexit{2}}}\big)}
      }
    }{ABROi} \\
    & \LBS{}{\{ R \}}{\emptyset}{1}{} &
    \Ssequence{
      \Strap{
        \Sparallel{\big(\Ssequence{\Snothing}{\Sloop{(\Sif{R}{\Sexit{2}}{\Spause})}}\big)}
                  {\big(\Ssequence{
                      \Ssequence{
                        \Ssequence{(
                          \Sparallel{\Sawimm{A}}
                                    {\Sawimm{B}}
                          )}
                          {\Semit{O}}
                        }
                        {(\Sloop{\Spause})}
                    }
                    {\Sexit{2}}
                  \big)}
      }
    }{ABROi} \\
    & \LBS{}{\{ A, B, R \}}{\{ O \}}{1}{} &
    \Ssequence{
      \Strap{
        \Sparallel{\big(\Ssequence{\Snothing}{\Sloop{(\Sif{R}{\Sexit{2}}{\Spause})}}\big)}
                  {\big(\Ssequence{\Snothing}{\Ssequence{(\Sloop{\Spause})}
                      {\Sexit{2}}}\big)}
      }
    }{ABROi} \\
  \end{array}
  \]
  \caption{Execution of ABROi in the LBS.}
  \label{fig:ABROi-LBS}
\end{figure}

\section{The Constructive Semantics
\texorpdfstring{\marginlink{Esterel.Semantics.CBS}{0mm}}{}} 
\label{sec:constructive-semantics}

The introduction showed that synchrony comes with some causality issues, in particular the program ``$\Sif{S}{\Semit S}{\Semit S}$''
(or \Tif{S}{\Temit S}{\Temit S}) is only valid in a classical setting as S is present because it is emitted but also emitted (\id{then} branch) because is it present.
In order to remove such causality loops as well as non-determinism, the constructive semantics was introduced in~\cite{Berry:ConstructiveSemanticsOfPureEsterel}.
We now focus on this semantics, which has become the reference for the language because it solves these issues but also because it provides a much better match with the circuit semantics~\cite{MendlerShipleBerry:ConstructiveCircuits}.

Altogether, the constructive (behavioral) semantics (CBS) we now present differs from the logical one on two main aspects: the statuses of signals and the signal rules.
Signals may take a third status~$⊥$ representing the absence of information: we do not yet know whether a signal is emitted or not.
This means in particular that no rule can apply to ``$\Sif s p q$'' (\Tif{$s$}{$p$}{$q$}) if $s$'s status is~$⊥$: the execution is blocked on the test.
Only the signal declaration rule deals with the~$⊥$ status.
For this, it uses two auxiliary mutually recursive functions \MustName{} and \CanName{} \marginlink{Esterel.Semantics.MustCan}{-2mm}
to compute respectively the set of signals that \emph{must} and \emph{can} be emitted within an instant using only the information contained in the event~$E$.
They intuitively represent the information we can gather from the body of~$p$ \emph{by making no assumptions on the status of the declared signal~$s$}, that is, by setting its status to~$⊥$ (unknown).
This restriction ensures that the justification of the status of signal~$s$ does not rely on itself.

The union of events is now defined pointwise using Scott's disjunction on $\{⊥,+,-\}$, that is: $+ \vee x = x \vee + = +$,  $- \vee ⊥ = ⊥ \vee - = ⊥ \vee ⊥ = ⊥$, and $- \vee - = -$ ; in essence, non emission must be agreed upon everywhere.

As most of the rules of the logical semantics do not deal with signal statuses, we do not need to modify them.
Moreover, the rules for ``$\Semit s$'', ``$\Sawimm s$'' and ``$\Sif s p q$'' can also be kept unmodified, as they only refer to statuses~$+$ and~$-$.
In fact, only the two rules for signal declaration (rules {\SigP} and {\SigM} rules on Figure~\ref{fig:LBS}) need to be reworked,
where handling $⊥$ is deferred to two auxiliary functions, $\MustName$ and $\CanName$.
Thus, the CBS contains exactly the same rules as the LBS, except for the {\SigP} and {\SigM} rules which are replaced by the following {\CSigP} and {\CSigM} rules:

\begin{prooftree}
  \AXC{$\CBS p {\addEvent s + E} {E'} k {p'}$}
  \AXC{$s ∈ \Must p {\addEvent s {⊥} E}$}
  \RightLabel{\CSigP}
  \BIC{$\CBS{\Ssignaldecl s p} E {\restrictEvent{E'} s} k {\deltafun k {\Ssignaldecl s {p'}}}$}
  \DP
  \hfill
  \AXC{$\CBS p {\addEvent s - E} {E'} k {p'}$}
  \AXC{$s ∉ \Can + p {\addEvent s {⊥} E}$}
  \RightLabel{\CSigM}
  \BIC{$\CBS{\Ssignaldecl s p} E {\restrictEvent{E'} s} k {\deltafun k {\Ssignaldecl s {p'}}}$}
\end{prooftree}

\subsection{The \CanName{} and \MustName{} functions}
\label{sec:can-must}

The auxiliary functions \MustName{} and \CanName{} are respectively an under- and an over-approximation of the final set of emitted signals.
They coincide on most statements, differing only on four of them: ``$\Sif s p q$'', ``$\Sawimm s$'', ``$\Ssequence p q$'', and ``$\Ssignaldecl s p$''.
For instance, when the status of~$s$ is~$⊥$ in a presence test $\Sif s p q$, we do not know which branch to execute so that nothing \emph{must} be done; on the contrary, both branches \emph{can} be executed, hence the difference between \MustName{} and \CanName.

Technically, we need to compute approximations for both signals and completion codes at the same time because the sequential composition~$\Ssequence p q$ needs to know if~$p$ must/can terminate to decide if~$q$ has to be considered.
The computed signal and completion code sets are denoted by adding a subscript to the \MustName{} and \CanName{} functions: $s$ for signals and $k$ for completion codes, as in $\MustS p E$ for signals and $\MustK p E$ for completion codes.
We usually drop the subscript as the ambiguity is easily resolved from the context. 
Furthermore, in the case of \MustName, the set of completion codes is either empty or a singleton because of determinism, whereas it can be an arbitrary non-empty set for \CanName.

When $\Must{\Ssignaldecl s p} E$ concludes that the status of the local signal~$s$ must be~$+$, we want to propagate~$s^+$ inside the evaluation of~$p$ to get more precise information.
This can only be done when we are sure that $\Ssignaldecl s p$ is executed, not merely that it \emph{could} be executed as \CanName{} expresses.
Therefore, \CanName{} needs to carry a $+$ tag telling whether the statement must be executed or a $⊥$ tag if we do not have this information.

The $\Kup{}$ and $\Kdown{}$ functions, used for computing the completion codes of statements $\Strap p$ and $\Sshift p$, are extended pointwise to sets of completion codes, that is: $\Kdown K = \left\{ \Kdown k \;\middle|\; k \in K \right\}$ and $\Kup K = \left\{ \Kup k \;\middle|\; k \in K \right\}$.
The maximum of two completion code sets $K_1$ and $K_2$, written $\MAX{K_1}{K_2}$ and used in the $\Sparallel p q$ case, is the set of pointwise maxima: $\left\{ \max(k_1, k_2) ~\middle|~ k_1 \in K_1 \land k_2 \in K_2\right\}$.
We abbreviate by $\Ssignaldecl e X$ the removal of element~$e$ from set~$X$, instead of the more standard but more cumbersome $X \setminus \{e\}$.
Operations on sets are extended pointwise to pairs of sets 
by applying them on each component where meaningful.
For example:
\[
  \begin{array}{r@{\quad=\quad}l}
    \Can{⊥} p E \cup \Can{⊥} q E & \left(\CanS{⊥} p E \cup \CanS{⊥} q E, \CanK{⊥} p E \cup \CanK{⊥} q E\right) \\
    \Ssignaldecl 0 {\Can m p E} & \left(\CanS m p E, \Ssignaldecl 0 {\CanK m p E} \right) \\
    \Ssignaldecl s {\Must p E} & \left(\Ssignaldecl s {(\MustS p E)}, \MustK p E\right) \\
  \end{array}
\]
Here, removing the integer $0$ from a set of signals does not make sense, hence the second example only applies it to the second component of the pair; similarly, removing a signal~$s$ from a set a completion codes does not make sense, hence the third example only applies it to the first component of the pair.
We also extend inclusion pointwise to these pairs by writing $\Must p E \subseteq \Can + p E$ to mean both $\MustS p E \subseteq \CanS + p E$ and $\MustK p E \subseteq \CanK + p E$.

$\MustName$ and $\CanName$ are defined by mutual induction (because of the signal declaration case) over the statement~$p$ (see Figure~\ref{fig:MustCanDef}).

\begin{figure}
  \begin{align*}
    \Must{\Sk k} E &= \Can m {\Sk k} E = (∅, \{ k \}) \\
    \Must{\Semit s} E &= \Can m {\Semit s} E = (\{ s \}, \{ 0 \}) \\
    \Must{\Sawimm s} E &=
    \begin{cases}
      (∅, \{ 0 \}) & \text{if } s^+ ∈ E \\
      (∅, \{ 1 \}) & \text{if } s^- ∈ E \\
      (∅, ∅) & \text{otherwise} \\
    \end{cases} &
    \Can m {\Sawimm s} E &=
    \begin{cases}
      (∅, \{ 0 \}) & \text{if } s^+ ∈ E \\
      (∅, \{ 1 \}) & \text{if } s^- ∈ E \\
      (∅, \{ 0, 1 \}) & \text{otherwise} \\
    \end{cases} \\
    \Must{\Sif s p q} E &=
    \begin{cases}
      \Must p E & \text{if } s^+ ∈ E \\
      \Must q E & \text{if } s^- ∈ E \\
      (∅, ∅) & \text{otherwise} \\
    \end{cases} &
    \Can m {\Sif s p q} E &=
    \begin{cases}
      \Can m p E & \text{if } s^+ ∈ E \\
      \Can m q E & \text{if } s^- ∈ E \\
      \Can{⊥} p E ∪ \Can{⊥} q E & \text{otherwise} \\
    \end{cases} \\
    \Must{\Ssuspend s p} E &= \Must p E &
    \Can m {\Ssuspend s p} E &= \Can m p E \\
    \Must{\Sloop p} E &= \Must p E &
    \Can m {\Sloop p} E &= \Can m p E \\
    \Must{\Strap p} E &= (\MustS p E, \Kdown{(\MustK p E)}) &
    \Can m {\Strap p} E &= (\CanS m p E, \Kdown{(\CanK m p E)}) \\
    \Must{\Sshift p} E &= (\MustS p E, \Kup{(\MustK p E)}) &
    \Can m {\Sshift p} E &= (\CanS m p E, \Kup{(\CanK m p E)}) \\
%
%
  \end{align*}
  \vspace{-4em}
  \begin{align*}
    \Must{\Ssequence p q } E &=
    \begin{cases}
      \Must p E & \text{if } 0 ∉ \MustK p E \\
      (\MustS p E ∪ \MustS q E, \MustK q E) & \text{if } 0 ∈ \MustK p E \\
    \end{cases} \\
    \Can m {\Ssequence p q} E &=
    \begin{cases}
      \Can m p E & \text{if } 0 ∉ \CanK m p E \\
      \left( \Ssignaldecl 0 {\Can m p E} \right) ∪ \Can + q E & \text{if } 0 ∈ \MustK p E \text{ and } m = + \\
      \left( \Ssignaldecl 0 {\Can m p E} \right) ∪ \Can{⊥} q E & \text{otherwise} \\
    \end{cases} \\
    \Must{\Sparallel p q} E &= (\MustS p E ∪ \MustS q E, \MAX{\MustK p E}{\MustK q E}) \\
    \Can m {\Sparallel p q} E &= (\CanS m p E ∪ \CanS m q E, \MAX{\CanK m p E}{\CanK m q E}) \\
    \Must{\Ssignaldecl s p} E &=
    \begin{cases}
      \Ssignaldecl s {\Must p {\addEvent s + E}} & \text{if } s ∈ \MustS p {\addEvent s {⊥} E} \\
      \Ssignaldecl s {\Must p {\addEvent s - E}} & \text{if } s ∉ \CanS + p {\addEvent s {⊥} E} \\
      \Ssignaldecl s {\Must p {\addEvent s {⊥} E}} & \text{otherwise} \\
    \end{cases} \\
    \Can m {\Ssignaldecl s p} E &=
    \begin{cases}
      \Ssignaldecl s {\Can m p {\addEvent s + E}} & \text{if } s ∈ \MustS p {\addEvent s {⊥} E} \text{ and } m = + \\
      \Ssignaldecl s {\Can m p {\addEvent s - E}} & \text{if } s ∉ \CanS m p {\addEvent s {⊥} E} \\
      \Ssignaldecl s {\Can m p {\addEvent s {⊥} E}} & \text{otherwise} \\
    \end{cases}
  \end{align*}
  \caption{Definitions of the functions \MustName{} and \CanName.
    The first component is for signals, the second one for completion codes.}
  \label{fig:MustCanDef}
\end{figure}

\begin{coqrmk}[The \MustName{} and \CanName{} functions in Coq]
There are two small differences between the presentation of \MustName{} and \CanName{} above and their Coq description.
First, as we know that $\MustK p E$ is either a singleton or the empty set, it is more convenient to use an option type to enforce this invariant: in Coq, $\MustK p E$ is either \texttt{Some} $k$ for some $k \in \N$ or \texttt{None}.
Second, the tag~$m$ carried by $\Can m p E$ to express whether~$p$ must be evaluated or not can take values~$+$ (surely executed) and~$\bot$ (maybe executed).
In Coq, we use a Boolean instead with true representing~$+$ and false representing~$\bot$.
\end{coqrmk}

\subsection{Properties of \MustName{}/\CanName{} and of the constructive behavioral semantics}

The constructive restriction ensures that the result is deterministic and that undefined statuses cannot be created during execution: they are only temporary statuses, meant to be used for local signals.
More precisely, if we define a \emph{total event} as a constructive event where no signal is mapped to~$⊥$, we have the following lemmas:
\marginlink[\#CBS_Total]{Esterel.Semantics.CBS}{10mm}
\begin{lemma}[Output events are total]
  \label{thm:preserving-totality}
  For all $p$, $E$, $E'$, $k$, and $p'$, if $\CBS p E {E'} k {p'}$, then~$E'$ is total.
  (Notice that we make no assumption on~$E$).
\end{lemma}
\marginlink[\#CBS_deterministic]{Esterel.Semantics.CBS}{10mm}
\begin{lemma}[Determinism of the constructive semantics]
  For all~$p$,~$E$,~$E_1'$,~$k_1$,~$p_1'$,~$E_2'$,~$k_2$, and~$p_2'$, if $\CBS p E {E_1'}{k_1}{p_1'}$ and $\CBS p E {E_2'}{k_2}{p_2'}$, then $E_1' = E_2'$, $k_1 = k_2$, and $p_1' = p_2'$.
\end{lemma}

As \MustName{} and \CanName{} represents under- and over-approximation of the behavior of statements under the CBS semantics, we have the following inclusions:
\begin{lemma}[Properties of \MustName{} and \CanName]
  For all $m$, $p$, and $E$, we have:
  \begin{itemize}
  \item When~$p$ is known to be executed, \CanName{} is more precise: $\Can + p E \subseteq \Can{⊥} p E$;
    \marginlink[\#Can_true_false_fst]{Esterel.Semantics.MustCan}{-2mm}
  \item Everything that \emph{must} be done \emph{can} be done: $\Must p E \subseteq \Can m p E$;
    \marginlink[\#Must_Can]{Esterel.Semantics.MustCan}{2mm}
  \end{itemize}
\end{lemma}

How large is the gap between the under- and over-approximations described by \MustName{} and \CanName{}?
The answer is that they coincide on constructive statements, that is, on statements that can reduce under the constructive semantics.
Conversely, whenever they are equal, the statement is constructive.
\marginlink[\#CBS_iff_Must_Can]{Esterel.Semantics.CBS}{10mm}
\begin{theorem}
  For all $p$, $E$, there exists a transition $\CBS p E {E'\!} k {p'}$ for some~$E'$,~$k$ and~$p'$ if and only if we have $\Must p E = \Can + p E$.
\end{theorem}
\begin{remark}
  This is only true for loop-safe~$p$~\cite{Berry:ConstructiveSemanticsOfPureEsterel}, that is, for statements in which the Can function on all loop bodies does not contain 0.
  This notion could be refined more, but it does not seem very useful.
\end{remark}

What about non-constructive statements?
The gap can be arbitrarily large in general, for instance if~$s$ does not appear in~$p$, then for any event~$E$, $\MustS{\Ssignaldecl s {(\Sif s p p)}} E = ∅$ whereas $\CanS + {\Ssignaldecl s {(\Sif s p p)}} E = \CanS + p E$.
Nevertheless, the \MustName{} and \CanName{} functions are well-named, as they indeed describe what must/can be observed, even on non-constructive statements:
\marginlink[\#Must_Can_LBS]{Esterel.Proofs.CBS_LBS}{10mm}
\begin{lemma}[Inclusion between LBS and \MustName{} and \CanName{}]
  \label{thm:MustCanLBSLemma}
  For all $p$, $E$, $E'$, $k$, and $p'$, if $\LBS p E {E'} k {p'}$ then we have $\Must p E \subseteq (E', \{ k \}) \subseteq \Can + p E$.
\end{lemma}
Notice that we use the logical behavioral semantics (LBS) here.
If we have $\Must p E \neq \emptyset$, then the statement is actually constructive and these inclusions become equalities.

At this point, the reader may wonder how the non-causal programs P$_1$ and P$_3$ mentioned in the introduction, are rejected.
The technical answer is quite simple: no constructive rule can be applied to them.
Why is that?
Because of the definitions of \CanName{} and \MustName.
For instance, to evaluate P$_1 = \Ssignaldecl s {(\Sif s {\Semit s} {\Snothing})}$ we have to check whether we can apply the rule for local signals, namely \CSigP\! or \CSigM\!.
  To apply \CSigP, we have to compute $\Must{(\Sif s {\Semit s} \Snothing)}{\addEvent s {⊥} E}$, which is $(\emptyset, \emptyset)$ as signal~$s$ has status~$⊥$ and thus, rule \CSigP{} cannot be applied.
To apply \CSigM, we need to compute $\Can + {(\Sif s {\Semit s} \Snothing)}{\addEvent s {⊥} E} = (\{s\}, \{0\})$ and rule \CSigM{} cannot be applied either.
The exact same reasoning applies to P$_3 = \Ssignaldecl s {(\Sif s {\Semit s} {\Semit s})}$.

The reader may also wonder whether these definitions of \MustName{} and \CanName{} accurately represent constructiveness.
In our opinion, what justifies it is the direct correspondence with the constructiveness
of sets of Boolean equations logically defining digital circuits.
More precisely, \cite{MendlerShipleBerry:ConstructiveCircuits} proves that Boolean constructiveness is equivalent to
electrical stability of the result computed by a physical mapping of the equations,
under a reasonable gate and wire delay model.

\subsection{Refinement between the logical and constructive semantics}
\marginlink{Esterel.Semantics.CBS_LBS}{-10mm}
\label{sec:refinement-constructive-coherent}

The constructive semantics is a refinement of the logical one that precludes some unwanted behaviors due to non-causality.
Therefore, we expect to have a theorem stating that any valid transition in the constructive semantics is also valid in the logical one.

In order to state this theorem in Coq, we need to convert constructive events into classical ones.
We name \CtoKName{} the function performing this conversion.
Because of the status~$⊥$ (unknown), there is no canonical way to do so.
Therefore, we choose to restrict the equivalence theorem to \emph{total (constructive) events}, that is, to constructive events~$E$ in which no signal is mapped to~$⊥$.
Thus, the precise statement we prove is the following:

\begin{theorem}
  \label{thm:constructive-refine-coherent}
  \marginlink[\#CBS_LBS]{Esterel.Semantics.CBS_LBS}{-2mm}
  For all $p$, $E$, $E'$, $k$, $p'$, if $\Total E$ and $\CBS p E {E'} k {p'}$ then $\LBS p {\CtoK E} {\CtoK{E'}} k {p'}$.
\end{theorem}
\begin{coqrmk}
  As we already know that the constructive semantics produces total output events, we could replace $\CtoK{E'}$ with $E'$.
  In Coq, it is still required for typing reasons.
\end{coqrmk}
\begin{proof}[Proof of Theorem~\ref{thm:constructive-refine-coherent}]
  By induction on the proof of $\CBS p E {E'} k {p'}$. \\
  As both semantics are the same except for signal declaration, only the \CSigP{} and \CSigM{} rules are interesting; the other ones are handled using expected properties of \CtoKName.
  These two cases rely on Lemma~\ref{thm:MustCanLBSLemma} above which expresses that the functions \MustName{} and \CanName{} indeed compute what their names imply: any signal (\resp completion code) computed by $\Must p E$ is indeed emitted (\resp reached) and any emitted signal or reached completion code indeed belongs to the corresponding set computed by $\Can + p E$.
\end{proof}

\section{The Constructive State Semantics
\texorpdfstring{\marginlink{Esterel.Semantics.CSS}{-2mm}}{}} 
\label{sec:state-semantics}

Even though the constructive semantics is well-suited as the starting semantics of Esterel, it is quite far from the circuit one of~\cite{Berry:ConstructiveSemanticsOfPureEsterel}, mostly because the models are very different.
On the Esterel side, the semantics transforms the source code into a derivative by keeping only the currently active parts of the program, \ie what is left to execute.
On the contrary, circuits are fixed hardware where only the state of registers can be modified between reactions.
To bridge this gap, one can present the constructive semantics in a different way, keeping the source code intact and adding tags on top of it to represent where the execution currently is:
these tags will correspond to a distributed program counter that precisely encodes the states of the hardware registers, and the program text will be preserved, exactly as the circuit wiring.
This leads us the \emph{constructive state semantics} we study now.

Let us consider the following program:
``$\Sseq{\Sseq{\Semit{o_1}}{\Spause}}{\Sif{s}{\Spar{(\Sseq{\Spause}{\Semit{o_2}})}{(\Sseq{\Semit{o_3}}{\Spause})}}{(\Sseq{\Spause}{\Semit{o_4}})}}$''.
In the constructive semantics, executing it in the environment $E = \{ s \}$ would lead to the following execution chain:

\begin{align*}
  \Sseq{\Sseq{\Semit{o_1}}{\Spause}}{\Sif{s}{\Spar{(\Sseq{\Spause}{\Semit{o_2}})}{(\Sseq{\Semit{o_3}}{\Spause})}}{(\Sseq{\Spause}{\Semit{o_4}})}}
  & \CBS{}{\{s\}}{\{o_1\}}{1}
       {\Sseq{\Snothing}{\Sif{s}{(\Spar{\Sseq{\Spause}{\Semit{o_2}}}{\Sseq{\Semit{o_3}}{\Spause}})}{(\Sseq{\Spause}{\Semit{o_4}})}}} \\
  & \CBS{}{\{s\}}{\{o_3\}}{1}
       {\Spar{(\Sseq{\Snothing}{\Semit{o_2}})}{\Snothing}}
  \CBS{}{\{s\}}{\{o_2\}}{0}{\Snothing}
\end{align*}

Here, executed statements keep disappearing until the whole program becomes $\Snothing$.
In the state semantics, we instead keep the program intact and move tags $\state{\;\cdot\;}$ to indicate where execution stops at each reaction.
One way to start the overall execution is to add an activated pause $\state{\Spause}$ in sequence before the program, called the \emph{boot} statement.
\begin{align*}
  \Sseq{\state{\Spause}}{\Sseq{\Sseq{\Semit{o_1}}{\Spause}}{\Sif{s}{\Spar{(\Sseq{\Spause}{\Semit{o_2}})}{(\Sseq{\Semit{o_3}}{\Spause})}}{(\Sseq{\Spause}{\Semit{o_4}})}}}
  &\CSS{}{\{s\}}{\{o_1\}}{1}
       {\Sseq{\Spause}{\Sseq{\Sseq{\Semit{o_1}}{\state{\Spause}}}{\Sif{s}{\Spar{(\Sseq{\Spause}{\Semit{o_2}})}{(\Sseq{\Semit{o_3}}{\Spause})}}{(\Sseq{\Spause}{\Semit{o_4}})}}}} \\
  &\CSS{}{\{s\}}{\{o_3\}}{1}
       {\Sseq{\Spause}{\Sseq{\Sseq{\Semit{o_1}}{\Spause}}{\Sif{s}{\Spar{(\Sseq{\state{\Spause}}{\Semit{o_2}})}{(\Sseq{\Semit{o_3}}{\state{\Spause}})}}{(\Sseq{\Spause}{\Semit{o_4}})}}}} \\
  &\CSS{}{\{s\}}{\{o_2\}}{0}
       {\Sseq{\Spause}{\Sseq{\Sseq{\Semit{o_1}}{\Spause}}{\Sif{s}{\Spar{(\Sseq{\Spause}{\Semit{o_2}})}{(\Sseq{\Semit{o_3}}{\Spause})}}{(\Sseq{\Spause}{\Semit{o_4}})}}}} \\
\end{align*}
In particular, one can directly observe two crucial invariants of the execution of Esterel statements:
\begin{enumerate*}[label=\textit{(\roman*)}]
  \item execution only stops on $\Spause$ and $\Sawimm s$ statements, which are the only statements that can carry a $\state{\;\cdot\;}$ tag;
  \item Consider the two branches of a presence test or of a sequential composition.
    Only one of them can be executed in a given reaction, unlike for parallel composition \Tpar{$p$}{$q$} that precisely triggers parallel threads.
\end{enumerate*}
These two properties are enforced by the very definition of states below.

\subsection{Formal definition of the state semantics}

Formally, states (written $\state p, \state q$) are defined by the following grammar:
\marginlink[\#state]{Esterel.Definitions}{0mm}
\begin{trivlist}
\item
  \begin{tabular}{@{\hspace{\offset}}l@{\qquad}r@{}r@{}l@{\sep}l@{\sep}l@{}}
    States & $\state p, \state q$ &
    $\quad := \quad$ & $\state 1$ \\
    &&& $\state{\,\Sawimm s\,}$ \\
    &&& $\Sif s {\state p} q$ & $\Sif s p {\state q}$ \\
    &&& $\Ssuspend s {\state p}$ \\
    &&& $\Strap{\state p}$ \\
    &&& $\Sshift{\state p}$ \\
    &&& $\Ssequence{\state p} q$ & $\Ssequence p {\state q}$ \\
    &&& $\Sloop{\state p}$ \\
    &&& $\Sparallel{\state p} q$ & $\Sparallel p {\state q}$ & $\Sparallel{\state p}{\state q}$ \\
    &&& $\Ssignaldecl s {\state p}$ \\
  \end{tabular}
\end{trivlist}
The only elementary states are the active $\Spause$ (\Tpause) and the active $\Sawimm s$ (\Tawimm{$s$}) statements, written $\state{\Spause}$ and $\state{\Sawimm s}$, and states propagate through all other compound statements of the language.
The decorated statements $\state 1$ and $\!\state{\,\Sawimm s\,}\!$ represent the points where execution should restart in the next instant, in other words, the aforementioned distributed program counter since execution always stops only on ``$1$'' or ``$\Sawimm s$'' statements in Kernel Esterel.
For example, the states of ``$\Ssequence{\Spause}{\Spause}$'' are ``$\Ssequence{\state{\Spause}}{\Spause}$'' and ``$\Ssequence{\Spause}{\state{\Spause}}$'', never ``$\Ssequence{\state{\Spause}}{\state{\Spause}}$'' as two states cannot be in sequence ; the only state of ``$\Ssequence{(\Sparallel{\Semit{s_1}}{\Spause})}{\Semit{s_2}}$'' is ``$\Ssequence{(\Sparallel{\Semit{s_1}}{\state{\Spause}})}{\Semit{s_2}}$'' as there is only one occurrence of $\Spause$ or $\Sawimm s$.
The ABROi program of Section~\ref{sec:kernel-statements} has four possible states: waiting on both signals~$A$ and~$B$, waiting on only one of them, waiting on the \id{halt} ($\Sloop{\Spause}$) statement.

The constructive state semantics is a rewording of the constructive semantics where we replace the derivative statement~$p'$ inside the rule $\CBS p E {E'} k {p'}$ by a derivative \emph{term}~$\term{p'}$.
  Such terms\marginlink[\#term]{Esterel.Definitions}{0mm} $\term p$, $\term q$ are either active states $\state p$, $\state q$ or inactive statements $p$, $q$.
  When a derivative term~$\term{p'}$ is a state~$\state{p'}$, it means that execution has paused within~$\state{p'}$ for the current reaction and can be resumed later.
  On the contrary, when this term~$\term{p'}$ is a statement~$p'$, it means that execution is over, either by normal termination or raising a trap.
Technically, the state semantics is still a rewriting semantics where the tags are moved around, but the core idea is that \emph{the underlying statement never changes}.
Execution of an instant starts with $p$ and yields a term~$\term p$.
From one reaction to the next, one resumes from $\state p$ but not from~$p$ itself, so that there is no implicit toplevel loop restarting~$p$.

Tagged statements (states) are called \emph{active}, and pauses are \emph{activated} in a reaction if they are executed and not killed until the end of the reaction.
Because Esterel has parallelism, there can be multiple active and activated statements inside a program: all active `$1$'s and `$\Sawimm s$'s are used to resume execution and some `$1$'s and `$\Sawimm s$'s are activated during the execution (possibly the same ones).

There are two kinds of state semantics: the \emph{start} semantics for untagged statements used when we start execution and the \emph{resumption} semantics for active states used when we resume execution.
We distinguish them by subscripts.
\[
  \text{\emph{Start} semantics}  \quad \CSSs p E {E'} k {\term{p'}}
  \qquad \qquad
  \text{\emph{Resumption} semantics} \quad \CSSr{\state p} E {E'} k {\term{p'}}
\]

To define these constructive state semantics, we need to extend
\MustName{} and \CanName{} to states.
An easy way to do this is to use an expansion function \expandName{}\marginlink[\#Sexpand]{Esterel.Definitions}{0mm} removing the already executed parts of a state, thus translating a state~$\state p$ into its
corresponding statement in the constructive semantics.
Therefore, we need to express ``what is left to execute'' inside a state.
Following~\cite{Berry:ConstructiveSemanticsOfPureEsterel}, we define the expansion (written \expandName{}) of a state~$\state p$ into a statement, erasing the already executed parts to return the parts that still need to be executed inside~$\state p$.
\begin{align*}
  \expand{\activatedpause} &:= \Snothing
  & \expand{\Sloop{\state p}} &:= \expand{\state p}; \Sloop p
  & \expand{\Sparallel{\state p}{\state q}} &:= \Sparallel{\expand{\state p}}{\expand{\state q}} \\
  \expand{\state{\Sawimm s}} &:= \Sawimm s
  & \expand{\Strap{\state p}} &:= \Strap{\expand{\state p}}
  & \expand{\Sparallel{\state p} q} &:= \Sparallel{\expand{\state p}}{\Snothing} \\
  \expand{\Sif s {\state p} q} &:= \expand{\state p}
  & \expand{\Sshift{\state p}} &:= \Sshift{\expand{\state p}}
  & \expand{\Sparallel p {\state q}} &:= \Sparallel{\Snothing}{\expand{\state q}} \\
  \expand{\Sif s p {\state q}} &:= \expand{\state q}
  & \expand{\Ssequence{\state p} q} &:= \expand{\state p}; q
  & \expand{\Ssignaldecl s {\state p}} &:= \Ssignaldecl s {\expand{\state p}} \\
  \expand{\Ssuspend s {\state p}} &:= \Ssequence{\Sawimm{\lnot s}}{\Ssuspend s {\expand{\state p}}}
  & \expand{\Ssequence p {\state q}} &:= \expand{\state q} \\
\end{align*}
Recalling that a term $\term p$ is either an active state $\state p$ or an inert statement $p$, this function is extended to terms $\term p$ by setting $\expand p = \Snothing$: in an inactive statement, nothing is pending execution.

\begin{remark}
  For the $\Spar{\state p} q$ and $\Spar p {\state q}$ cases, keeping the seemingly useless $\Spar \Snothing \_$ is actually technically useful to avoid a mismatch between the derivatives of both semantics.
  Indeed, this simplifies the equivalence between the constructive and state semantics: if we instead set $\expand{\Spar{\state p} q} = \expand{\state p}$, the parallel statement disappears and we would need to prove that $\Spar p \Snothing$ and $p$ are equivalent, which would require introducing bisimulation.
\end{remark}

With this expansion function \expandName, we let $\Must{\state p} E$ be $\Must{\expand{\state p}} E$ and
similarly for \CanName, so that we can directly reuse all the already
proved properties about \MustName{} and \CanName{} and have for free
the following equivalences:
\marginlink[\#sMust_sCan_expand_Must_Can]{Esterel.Semantics.StateMustCan}{5mm}
\[
  s ∈ \MustS{\state p} E \iff s ∈ \MustS{\expand{\state p}} E \hspace{3em}
  s ∈ \CanS m {\state p} E \iff s ∈ \CanS m {\expand{\state p}} E
\]



The rules of the state semantics are given in Figure~\ref{fig:start-rules} for the start semantics rules and in Figure~\ref{fig:resumption-rules} for the resumption semantics rules.

\begin{figure}
\begin{prooftree}
  \AXC{}
  \RightLabel{nothing}
  \UIC{$\CSSs \Snothing E {∅} 0 \Snothing$}
  \DP
  \hspace{\rulehspace}
  \AXC{}
  \RightLabel{pause}
  \UIC{$\CSSs \Spause E {∅} 1 \activatedpause$}
  \DP
  \hspace{\rulehspace}
  \AXC{$k \geq 2$}
  \RightLabel{exit}
  \UIC{$\CSSs{\Sexit k} E {∅}{k}{\Sexit k}$}
  \DP \\[\rulevspace]

  \AXC{}
  \RightLabel{emit}
  \UIC{$\CSSs{\Semit s} E {\SingletonEvent s} 0 {\Semit s}$}
  \DP \\[\rulevspace]

  \AXC{$s^+ ∈ E$}
  \RightLabel{awimm${}^+$}
  \UIC{$\CSSs{\Sawimm s} E {∅} 0 {\Sawimm s}$}
  \DP
  \hspace{\rulehspace}
  \AXC{$s^- ∈ E$}
  \RightLabel{awimm${}^-$}
  \UIC{$\CSSs{\Sawimm s} E {∅} 1 {\state{\Sawimm s}}$}
  \DP \\[\rulevspace]

  \AXC{$s^+ ∈ E$}
  \AXC{$\CSSs p E {E'} k {\term{p'}}$}
  \RightLabel{then}
  \BIC{$\CSSs{\Sif s p q} E {E'} k {\Sif s {\term{p'}} q}$}
  \DP
  \hspace{\rulehspace}
  \AXC{$s^- ∈ E$}
  \AXC{$\CSSs q E {E'} k {\term{q'}}$}
  \RightLabel{else}
  \BIC{$\CSSs{\Sif s p q} E {E'} k {\Sif s p {\term{q'}}}$}
  \DP \\[\rulevspace]

  \AXC{$\CSSs p E {E'} k {\term{p'}}$}
  \RightLabel{suspend}
  \UIC{$\CSSs{\Ssuspend s p} E {E'} k {\Ssuspend s {\term{p'}}}$}
  \DP \\[\rulevspace]

  \AXC{$k \neq 0$}
  \AXC{$\CSSs p E {E'} k {\term{p'}}$}
  \RightLabel{loop}
  \BIC{$\CSSs{\Sloop p} E {E'} k {\Sloop{\term{p'}}}$}
  \DP \\[\rulevspace]

  \AXC{$\CSSs p E {E'} k {\term{p'}}$}
  \RightLabel{trap}
  \UIC{$\CSSs{\Strap{\strut p}} E {E'}{↓k}{\Strap{\term{p'}}}$}
  \DP
  \hspace{\rulehspace}
  \AXC{$\CSSs p E {E'} k {\term{p'}}$}
  \RightLabel{shift}
  \UIC{$\CSSs{\Sshift p} E {E'}{\Kup k}{\Sshift{\term{p'}}}$}
  \DP \\[\rulevspace]

  \AXC{$k \neq 0$}
  \AXC{$\CSSs p E {E'} k {\term{p'}}$}
  \RightLabel{seq$_k$}
  \BIC{$\CSSs{\Ssequence p q} E {E'} k {\Ssequence{\term{p'}} q}$}
  \DP
  \hspace{\rulehspace}
  \AXC{$\CSSs p E {E_p'} 0 {p}$}
  \AXC{$\CSSs q E {E_q'} k {\term{q'}}$}
  \RightLabel{seq$_0$}
  \BIC{$\CSSs{\Ssequence p q} E {E_p' ∪ E_q'} k {\Ssequence p {\term{q'}}}$}
  \DP \\[\rulevspace]

  \AXC{$\CSSs p E {E_p'}{k_p}{\term{p'}}$}
  \AXC{$\CSSs q E {E_q'}{k_q}{\term{q'}}$}
  \RightLabel{parallel}
  \BIC{$\CSSs{\Sparallel p q} E {E_p' ∪ E_q'}{\max(k_p,k_q)}{\Sparallel{\term{p'}}{\term{q'}}}$}
  \DP \\[\rulevspace]

  \AXC{$s ∈ \Must p {\addEvent s {⊥} E}$}
  \AXC{$\CSSs p {\addEvent s + E}{E'} k {\term{p'}}$}
  \RightLabel{\CSigP}
  \BIC{$\CSSs{\Ssignaldecl s p} E {\restrictEvent{E'} s} k {\Ssignaldecl s {\term{p'}}}$}
  \DP
  \hfill
  \AXC{$s ∉ \Can + p {\addEvent s {⊥} E}$}
  \AXC{$\CSSs p {\addEvent s - E}{E'} k {\term{p'}}$}
  \RightLabel{\CSigM}
  \BIC{$\CSSs{\Ssignaldecl s p} E {\restrictEvent{E'} s} k {\Ssignaldecl s {\term{p'}}}$}
\end{prooftree}
\caption{Start rules of the state semantics.}
\label{fig:start-rules}
\end{figure}

\begin{figure}
\begin{prooftree}
  \AXC{}
  \RightLabel{pause}
  \UIC{$\CSSr \activatedpause E {∅} 0 \Spause$}
  \DP \\[\rulevspace]

  \AXC{$s^+ ∈ E$}
  \RightLabel{awimm${}^+$}
  \UIC{$\CSSr{\state{\Sawimm s}} E {∅} 0 {\Sawimm s}$}
  \DP
  \hspace{\rulehspace}
  \AXC{$s^- ∈ E$}
  \RightLabel{awimm${}^-$}
  \UIC{$\CSSr{\state{\Sawimm s}} E {∅} 1 {\state{\Sawimm s}}$}
  \DP \\[\rulevspace]

  \AXC{$\CSSr{\state p} E {E'} k {\term{p'}}$}
  \RightLabel{then}
  \UIC{$\CSSr{\Sif s {\state p} q} E {E'} k {\Sif s {\term{p'}} q}$}
  \DP
  \hspace{\rulehspace}
  \AXC{$\CSSr{\state q} E {E'} k {\term{q'}}$}
  \RightLabel{else}
  \UIC{$\CSSr{\Sif s p {\state q}} E {E'} k {\Sif s p {\term{q'}}}$}
  \DP \\[\rulevspace]

  \AXC{$s^+ ∈ E$}
  \RightLabel{suspend${}^+$}
  \UIC{$\CSSr{\Ssuspend s {\state p}} E {∅} 1 {\Ssuspend s {\state p}}$}
  \DP
  \hspace{\rulehspace}
  \AXC{$s^- ∈ E$}
  \AXC{$\CSSr{\state p} E {E'} k {\term{p'}}$}
  \RightLabel{suspend${}^-$}
  \BIC{$\CSSr{\Ssuspend s {\state p}} E {E'} k {\Ssuspend s {\term{p'}}}$}
  \DP \\[\rulevspace]

  \AXC{$\CSSr{\state p} E {E'} k {\term{p'}}$}
  \AXC{$k \neq 0$}
  \RightLabel{loop$_k$}
  \BIC{$\CSSr{\Sloop{\state p} \;} E {E'} k {\Sloop{\term{p'}}}$}
  \DP
  \hspace{\rulehspace}
  \AXC{$\CSSr{\state p} E {E'} 0 {p}$}
  \AXC{$\CSSs p E {E'} k {\term{p'}}$}
  \AXC{$k \neq 0$}
  \RightLabel{loop$_0$}
  \TIC{$\CSSr{\Sloop{\state p}\;} E {E'} k {\Sloop{\term{p'}}}$}
  \DP \\[\rulevspace]

  \AXC{$\CSSr{\state p} E {E'} k {\term{p'}}$}
  \RightLabel{trap}
  \UIC{$\CSSr{\Strap{\state p}} E {E'}{↓k}{\Strap{\term{p'}}}$}
  \DP
  \hspace{\rulehspace}
  \AXC{$\CSSr{\state p} E {E'} k {\term{p'}}$}
  \RightLabel{shift}
  \UIC{$\CSSr{\Sshift{\state p}} E {E'}{\Kup k}{\Sshift{\term{p'}}}$}
  \DP \\[\rulevspace]

  \AXC{$\CSSr{\state p} E {E'} k {\term{p'}}$}
  \AXC{$k \neq 0$}
  \RightLabel{seq$_k$}
  \BIC{$\CSSr{\Ssequence{\state p} q} E {E'} k {\Ssequence{\term{p'}} q}$}
  \DP
  \hfill
  \AXC{$\CSSr{\state p} E {E_p'} 0 p$}
  \AXC{$\CSSs q E {E_q'} k {\term{q'}}$}
  \RightLabel{seq$_0$}
  \BIC{$\CSSr{\Ssequence{\state p} q} E {E_p' ∪ E_q'} k {\Ssequence p {\term{q'}}}$}
  \DP
  \hfill
  \AXC{$\CSSr{\state q} E {E_q'} k {\term{q'}}$}
  \RightLabel{seq$_q$}
  \UIC{$\CSSr{\Ssequence p {\state q}} E {E_q'} k {\Ssequence p {\term{q'}}}$}
  \DP \\[\rulevspace]

  \AXC{$\CSSr{\state p} E {E_p'}{k_p}{\term{p'}}$}
  \AXC{$\CSSr{\state q} E {E_q'}{k_q}{\term{q'}}$}
  \RightLabel{both}
  \BIC{$\CSSr{\Sparallel{\state p}{\state q}} E {E_p' ∪ E_q'}{\max(k_p,k_q)}{\Sparallel{\term{p'}}{\term{q'}}}$}
  \DP \\[\rulevspace]

  \AXC{$\CSSr{\state p} E {E_p'}{k_p}{\term{p'}}$}
  \RightLabel{left}
  \UIC{$\CSSr{\Sparallel{\state p} q} E {E_p'}{k_p}{\Sparallel{\term{p'}}{ q}}$}
  \DP
  \hspace{\rulehspace}
  \AXC{$\CSSr{\state q} E {E_q'}{k_q}{\term{q'}}$}
  \RightLabel{right}
  \UIC{$\CSSr{\Sparallel p {\state q}} E {E_q'}{k_q}{\Sparallel p {\term{q'}}}$}
  \DP \\[\rulevspace]

  \AXC{$s ∈ \Must{\expand p} {\addEvent s {⊥} E}$}
  \AXC{$\CSSr{\state p} {\addEvent s + E}{E'} k {\term{p'}}$}
  \RightLabel{\CSigP}
  \BIC{$\CSSr{\Ssignaldecl s {\state p}} E {\restrictEvent{E'} s} k {\Ssignaldecl s {\term{p'}}}$}
  \DP
  \hfill
  \AXC{$s ∉ \Can + {\expand p} {\addEvent s {⊥} E}$}
  \AXC{$\CSSr{\state p} {\addEvent s - E}{E'} k {\term{p'}}$}
  \RightLabel{\CSigM}
  \BIC{$\CSSr{\Ssignaldecl s {\state p}} E {\restrictEvent{E'} s} k {\Ssignaldecl s {\term{p'}}}$}
\end{prooftree}
\caption{Resumption rules of the state semantics.}
\label{fig:resumption-rules}
\end{figure}

\begin{remark}
  \label{rmk:classical-state-semantics}
  The transformation done to turn the constructive semantics into the constructive state semantics can also be done on the logical semantics, yielding a logical state semantics as done in~\cite{Berry:ConstructiveSemanticsOfPureEsterel}.
  The only difference with the constructive version would be the local signal rules.
\end{remark}

We illustrate the execution of the ABROi program in the CSS semantics in Figure~\ref{fig:ABROi-CSS}.

\begin{figure}
  \[
  \begin{array}{r@{\;}c@{\;}l}
    \Ssequence{\state{\Spause}}{ABROi} & = &
    \Ssequence{\state{\Spause}}{\Sloop{
      (\Strap{
          \Sparallel{
            \big(\Ssequence{\Spause}{
              \Sloop{(\Sif{R}{\Sexit{2}}{\Spause})}}
            \big)}
                  {(\Ssequence{
                     (\Ssequence{
                        \Ssequence{
                          (\Sparallel{\strut\Sawimm{A}}
                                     {\Sawimm{B}})}
                          {\Semit{O}}}
                        {(\Sloop{\Spause})})
                      }{\Sexit 2})}}
    )}} \\
    & \CSSr{}{\{ B \}}{\emptyset} 1 {} &
        \Ssequence{\Spause}{
          \Sloop{
            (\Strap{
              \Sparallel{
                \big(\Ssequence{\state{\Spause}}{
                  \Sloop{(\Sif{R}{\Sexit{2}}{\Spause})}}
                \big)}
                        {(\Ssequence{
                            (\Ssequence{
                              \Ssequence{
                                (\Sparallel{\state{\Sawimm{A}}}
                                           {\Sawimm{B}})}
                                {\Semit{O}}}
                            {(\Sloop{\Spause})})
                          }{\Sexit 2})}}
            )}} \\
    & \CSSr{}{\{ A, B \}}{\{ O \}} 1 {} &
        \Ssequence{\Spause}{
          \Sloop{
            (\Strap{
              \Sparallel{
                \big(\Ssequence{\Spause}{
                  \Sloop{(\Sif{R}{\Sexit{2}}{\state{\Spause}})}}
                \big)}
                        {(\Ssequence{
                            (\Ssequence{
                              \Ssequence{
                                (\Sparallel{\Sawimm{A}}
                                           {\Sawimm{B}})}
                                {\Semit{O}}}
                            {(\Sloop{\state{\Spause}})})
                          }{\Sexit 2})}}
            )}} \\
    & \CSSr{}{\{ B \}}{\emptyset} 1 {} &
        \Ssequence{\Spause}{
          \Sloop{
            (\Strap{
              \Sparallel{
                \big(\Ssequence{\Spause}{
                  \Sloop{(\Sif{R}{\Sexit{2}}{\state{\Spause}})}}
                \big)}
                        {(\Ssequence{
                            (\Ssequence{
                              \Ssequence{
                                (\Sparallel{\Sawimm{A}}
                                           {\Sawimm{B}})}
                                {\Semit{O}}}
                            {(\Sloop{\state{\Spause}})})
                          }{\Sexit 2})}}
            )}} \\
    & \CSSr{}{\{ R \}}{\emptyset} 1 {} &
        \Ssequence{\Spause}{
          \Sloop{
            (\Strap{
              \Sparallel{
                \big(\Ssequence{\state{\Spause}}{
                  \Sloop{(\Sif{R}{\Sexit{2}}{\Spause})}}
                \big)}
                        {(\Ssequence{
                            (\Ssequence{
                              \Ssequence{
                                (\Sparallel{\state{\Sawimm{A}}}
                                           {\state{\Sawimm{B}}})}
                                {\Semit{O}}}
                            {(\Sloop{\Spause})})
                          }{\Sexit 2})}}
            )}} \\
    & \CSSr{}{\{ A, B, R \}}{\{ O \}} 1 {} &
        \Ssequence{\Spause}{
          \Sloop{
            (\Strap{
              \Sparallel{
                \big(\Ssequence{\state{\Spause}}{
                  \Sloop{(\Sif{R}{\Sexit{2}}{\Spause})}}
                \big)}
                        {(\Ssequence{
                            (\Ssequence{
                              \Ssequence{
                                (\Sparallel{\Sawimm{A}}
                                           {\Sawimm{B}})}
                                {\Semit{O}}}
                            {(\Sloop{\state{\Spause}})})
                          }{\Sexit 2})}}
            )}} \\
  \end{array}
  \]
  \caption{Execution of ABROi in the CSS.}
  \label{fig:ABROi-CSS}
\end{figure}

Finally, we can now prove that our initial idea of preserving the underlying statement actually works.
Let us define an erasing function \baseName{} that converts a term into its untagged underlying statement, which amounts to erasing active tags by turning all active ``$\state{\Spause}$''s and ``$\state{\Sawimm s}$''s into inactive ones.
Then, the following lemma holds:
\begin{lemma}[Invariance of the base statement]
  For all $p$, $E$, $E'$, $k$, $p'$, if $\CSSs p E {E'} k {\overline{p'}}$ then $p = \base{\term{p'}}$. \marginlink[\#sCSS_base]{Esterel.Semantics.CSS}{0mm} \\
  For all $\state p$, $E$, $E'$, $k$, $p'$, if $\CSSr{\state p} E {E'} k {\term{p'}}$ then $\base{\state p} = \base{\term{p'}}$. \marginlink[\#rCSS_base]{Esterel.Semantics.CSS}{0mm}
\end{lemma}
Being a rewording of the constructive semantics, the state semantics enjoys the same properties, in particular the state start and resumption semantics are deterministic.
\marginlink[\#sCSS_deterministic]{Esterel.Semantics.CSS}{-5mm}

\begin{coqrmk}[The constructive state semantics in Coq]
  The functions $\expandName$ and $\baseName$ are written respectively \id{expand} and \id{base}.
Furthermore, in order to have a more self-contained definition of the state semantics, the Coq development defines the \MustName{} and \CanName{} functions on states from scratch.
Then, the equivalences $s ∈ \MustS{\state p} E \iff s ∈ \MustS{\expand{\state p}} E$ and $s ∈ \CanS m {\state p} E \iff s ∈ \CanS m {\expand{\state p}} E$ are proved rather than used as definitions.
From these equivalences, the properties of \MustName{} and \CanName{} on states are imported from the versions on statements.
\end{coqrmk}

\subsection{Equivalence between the constructive behavioral semantics and the constructive state semantics}
\label{sec:equiv-constructive-state}

To compare the constructive semantics and the state one, we need the expansion function \expandName{} to relate the meaning of the derivatives of both semantics.
Then, we can state the equivalence between the constructive behavioral semantics and the constructive state semantics as follows:
\begin{theorem}[Equivalence between the constructive and state semantics]
  \label{thm:equivalence-constructive-state-semantics}
  \begin{align*}
    \text{For all } p, E, E', k, \term{p'}, \quad& \CSSs p E {E'} k \term{p'} \implies \CBS p E {E'} k {\expand{\term{p'}}}
    \marginlink[\#sCSS_CBS]{Esterel.Semantics.CBS_CSS}{0mm} \\
    \text{For all } \state p, E, E', k, \term{p'}, \quad& \CSSr{\state p} E {E'} k \term{p'} \implies \CBS{\expand{\state p}} E {E'} k {\expand{\term{p'}}}
    \marginlink[\#rCSS_CBS]{Esterel.Semantics.CBS_CSS}{0mm} \\
    \text{For all } p, E, E', k, p', \quad& \CBS p E {E'} k {p'} \implies \exists \term{p''}, \CSSs p E {E'} k \term{p''} \land p' = \expand{\term{p''}}
    \marginlink[\#CBS_sCSS]{Esterel.Semantics.CBS_CSS}{0mm} \\
    \text{For all } \state p, E, E', k, p', \quad& \CBS{\expand{\state p}} E {E'} k {p'} \implies \exists \term{p''}, \CSSr{\state p} E {E'} k \term{p''} \land p' = \expand{\term{p''}}
    \marginlink[\#CBS_rCSS]{Esterel.Semantics.CBS_CSS}{0mm} \\
  \end{align*}
\end{theorem}

\begin{proof}
  The proofs of these four implications all go by induction on the derivation tree in the source semantics.
  They are straightforward, they simply require to match the source semantics rules by case distinction on~$p'$ and~$k$.
\end{proof}

\begin{coqrmk}
  The actual Coq statement and proof are slightly different because we need to insert $\deltafun k \_$ and use the fact that $\deltafun k {\deltafun k p} = \deltafun k p$.
\end{coqrmk}

\subsection{Going beyond a state}

In the circuit translation of~\cite{Berry:ConstructiveSemanticsOfPureEsterel}, the state semantics provides a perfect match between states $\state p$ and active circuit states at the end of a clock cycle, that is, those in which at least one register is 1.
We can extend this correspondence to untagged statements $p$ and inactive circuits.
In particular, the registers in the circuit exactly correspond to ``$\Spause$'' (\Tpause) and ``$\Sawimm s$'' (\Tawimm{$s$}) statements.\footnote{Technically, the circuit translation of ``$\Ssuspend s p$'' also contains a register (the $\Sawimm{\lnot s}$ appearing in the CBS rule). We do not add it explicitly to the state semantics because it is only active whenever $p$ is.}
Nevertheless, the computation of semantics states is abstract: we go from one state to the next, without explaining precisely how information flows between these states.
Furthermore, the state semantics is still using recursively the \MustName{} and \CanName{} functions to compute the statuses of local signals.
This is very different from the computation of wire values in digital circuits: instead of performing a two-step recursive evaluation using \MustName/\CanName{}, a digital circuit propagates electrical values in a forward way, following causality chains.

To solve these issues, we now introduce the \emph{microstep semantics}, detailing how computation works \emph{within} an instant as a \emph{sequence} of microsteps, implementing a progressive evolution from one state to the next.
This is the semantic equivalent of the propagation of electrical fronts in a digital circuit, which is internally asynchronous but where each step and the full reaction have predicable time and result.


\section{Microstep Semantics
\texorpdfstring{\marginlink{Esterel.Semantics.Microstep}{-2mm}}{}}
\label{sec:microstep-semantics}

With the microstep semantics, our goal is to have an Esterel semantics as precise as the generated circuit, meaning that its resolution should be (roughly) at the gate level, and that it should avoid the mutually recursive definitions of \MustName{} and \CanName.
Furthermore, we want it to be defined directly on the source code, so that the equivalence with the previous semantics is easier to express.
Because of that, like the state semantics, the microstep semantics tags Esterel programs and execution consists in moving tags.
Like digital circuits in which electricity propagates within a clock cycle, we restrict the microstep semantics to a single reaction, meaning that there will be no rule to move past an electrical register
and that execution is performed starting either from the boot of the generated circuit (in the starting instant) or from active ``$\Spause$'' (\Tpause) or ``$\Sawimm s$'' (\Tawimm{$s$}) statements (in resumption instants) to either termination or reaching an activated ``$\Spause$'' or ``$\Sawimm s$'' statement.
This matches the Register Transfer Level (RTL) evaluation in circuit design.


\subsection{Microstate Definition}

\subsubsection{Relation with the circuit translation}
The core idea in designing the tags propagated by the microsteps is taken from a simplification of the circuit translation of~\cite{Berry:ConstructiveSemanticsOfPureEsterel}.
Actually, the microstep semantics can be seen as a simulation of this circuit translation on top of the Kernel Esterel statements.
As Kernel Esterel does not feature data, the key point is to transfer control.
In the circuit translation, this is done through a structural translation respecting the circuit interface depicted in Figure~\ref{fig:circuit-interface}.
\begin{figure}[ht]
  \centering
  \includegraphics[width=0.2\linewidth]{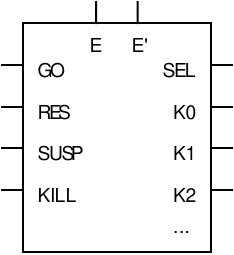}
  \caption{Circuit Translation Interface, taken from~\cite{Berry:ConstructiveSemanticsOfPureEsterel}.}
  \label{fig:circuit-interface}
\end{figure}

This interface has:
\begin{itemize}
  \item input and output events for signals ($E$ and $E'$ respectively);
  \item four input wires (Go, Res, Susp, Kill);
  \item a finite number of output wires (Sel and the K$_i$).
\end{itemize}
All are single wires except for $E$ and $E'$ which are bundles of wires (one per signal).
Their meaning is as follows:
\begin{itemize}
  \item Go is set to start the execution of the (currently inactive) circuit;
  \item Res (for Resume) is set to continue execution of the (currently active) circuit;
  \item Susp (for Suspend) is set to freeze the circuit registers for the current instant, but the combinational logic is still executed, in particular signals may be emitted;
  \item Kill is used to reset the circuit into its inactive state, that is, reset all registers to 0;
  \item Sel represents whether the circuit is active or not, that is, whether one of its registers is 1;
  \item the K$_i$ represent the completion code in one-hot encoding: completion code $k$ is represented by wire K$_k$ being one and all other wires K$_i$ ($i \neq k$) being zero.
\end{itemize}

Looking at Figure~\ref{fig:circuit-translation}, we can see that among these wires of the circuit translation of~\cite{Berry:ConstructiveSemanticsOfPureEsterel}, only Go and Res are used to propagate control within an instant while the Susp and Kill wires are only used as state register inputs, thus never to emit a signal or compute a completion code.
  Therefore, the effect of the latter only becomes visible at the next instant.
  In our semantics, the same effect is obtained by the use of auxiliary functions (see Section~\ref{sec:-microstep-intuition}), which makes it possible to omit the Susp and Kill wires for simplicity.
  The Sel wire denotes which parts of a circuit are active, implementing the $\state{\,\cdot\,}$ tag of the state semantics.
  We must keep it because it is used in the control propagation of $\Ssuspend s {\state p}$ (see Figure~\ref{fig:circuit-translation-suspend}, p. \pageref{fig:circuit-translation-suspend}).
  Since its value does not evolve during an instant, we implement its effect at the start of each instant, independently of the circuit's control or signal inputs.
  Thus, the reader may think of Sel as a constant input and we keep it separate from the input and output colors (see Figure~\ref{fig:def-microstate}), which evolve within an instant.

Remember that a statement can only return a single integer completion code; here it is technically simpler to encode this integer by the one-hot encoding represented by the  K$_i$.
This is a classical technique in circuits.

\subsubsection{Microstate Interface}
From the previous analysis, we design a \emph{microstate interface}\marginlink[\#microstate]{Esterel.Semantics.Microstate}{-2mm} to annotate the source command $p$: the notation is $\Sin \, p \, \Sout$, defined as follows.

Inputs $\Sin$ consist in the Go and Res wires (as Scott Booleans that can take values $\bot$, $+$, or $-$); signals are read from an (external) event $E$ but no output event $E'$ is produced (instead it will be read from the evaluated microstate); outputs $\Sout$ consist in the completion code represented as either $\black k$ with $k$ an integer, meaning that wire K$_k$ is 1 and all other wires are 0, or $\white K$ with $K$ a finite set of integers, meaning that wires K$_i$ with $i \in K$ are $\bot$ and all other wires are 0 (and in particular, no wire is 1).
This completion code representation directly denotes the 1-hot encoding, thus avoiding maintaining it as an invariant or dealing with impossible cases.

In fact, $\black k$ represents the fact that the~$k$-th wire is~$+$ (\MustName{} propagation) whereas $\white K$ represents a finite set~$K$ of completion code wires whose value is still unknown, the others being~$-$ (\CanName{} propagation). 
The information about Sel is assumed to be already known so we do not need to consider the $\bot$ case and can use a simple Boolean.
We use the $\white{K \setminus i}$ notation to denote setting wire~$i$ to 0 inside $\white K$ instead of the more standard but cumbersome $\white{K \setminus \{i\}}$.
We abbreviate $val = b$ as $val^b$, both when used as an equality or as an assignment.

These $\Sin$ and $\Sout$ annotations are called \emph{colors} and are recursively inserted in the bodies and branches of compound microstates.
For example, $\Ssuspend s p$ becomes $\Sin (\Ssuspend s {(\Sin p \Sout)}) \Sout$.
We usually drop parentheses when ambiguity can be resolved from the context.
The formal definition of microstates is given in Figure~\ref{fig:def-microstate}.
\begin{figure}[htb]
  \[
  \begin{array}{r@{\quad}c@{\quad}l@{\qquad}l}
    sel   & := & + \quad\mid\quad -    & \text{\textbf{sel Boolean}} \\[1em]
    cb    & := & \bot \quad\mid\quad + \quad\mid\quad - & \text{\textbf{constructive Boolean}} \\[1em]
    \Sin  & := & \{ Go: cb \;; Res: cb \}  & \text{\textbf{input color}} \\[1em]
    \Sout & := &                      & \text{\textbf{output color}} \\
          & \mid & \black k \hfill k \in \N & \text{$k$-th wire is $1$} \\
          & \mid & \white K \hfill K \subset_{\mathit{f\!in}} \N & \text{wires outside $K$ are $0$} \\[1em]
    mstate& := & sel \; \Sin \; mstmt \; \Sout & \text{\textbf{microstate}} \\[1em]
    mstmt & :=   & \Snothing \\
          & \mid & \Spause \\
          & \mid & \Semit s \hfill s \text{ signal} \\
          & \mid & \Sawimm s \hfill s \text{ signal} \\
          & \mid & \Sif s {mstate}{mstate} \qquad s \text{ signal} \\
          & \mid & \Ssuspend s {mstate} \hfill s \text{ signal} \\
          & \mid & \Strap{mstate} \\
          & \mid & \Sexit k \hfill k \geq 2 \\
          & \mid & \Sshift{mstate} \\
          & \mid & \Ssequence{mstate}{mstate} \\
          & \mid & \Sloop{mstate} \\
          & \mid & \Sparallel{mstate}{mstate} \\
          & \mid & \Ssignaldecl s {mstate} \hfill s \text{ signal} \\
  \end{array}
  \]
  \caption{Definition of microstates.}
  \label{fig:def-microstate}
\end{figure}

\subsubsection{Notations}
\label{sec:input-color-notations}
We color inputs as follows to have a visual representation of the various values of Go, Res and Sel (whether the square is filled or not).
\begin{center}
  \begin{tabular}{l@{\;\;}|@{\;\;}ccc@{\quad}l}
    \hline
    Colors   & Sel & Go  & Res & Intuition \\
    \hline
    \quad$\Sin$   & $?$ & $?$ & $?$ & We know nothing \\
    \quad$\SWin$  & $-$ & $?$ & $?$ & We know only Sel$⁻$ \\
    \quad$\go$    & $-$ & $+$ & $?$ & Starting an inactive statement \\
    \quad$\nogo$  & $-$ & $-$ & $?$ & Keeping inactive an inactive statement \\
    \quad$\SBin$  & $+$ & $?$ & $?$ & We know only Sel$⁺$ \\
    \quad$\res$   & $+$ & $?$ & $+$ & Resuming an active statement \\
    \quad$\nores$ & $+$ & $?$ & $-$ & Not executing an active statement \\
    \hline
  \end{tabular}
\end{center}
When we have Sel$^-$ (that is, in the $\SWin$, $\go$, and $\nogo$ cases), Theorem~\ref{thm:micro-ignores-Res} will formally justify that the value of Res is irrelevant.
Similarly, an invariant of the circuit translation, called  \texttt{input\_invariant}$(Sel, \Sin)$\marginlink[\#input_invariant]{Esterel.Proofs.ValidColoring}{0mm} and defined as $Go(\Sin)⁺ \implies Sel⁻$, ensures that Sel and Go can never be true simultaneously.\footnote{This requires careful handling of loops to maintain. For instance, the circuit translation of $\Sloop{\Spause}$ will have both Sel and Go wires set to 1 after the first cycle.}
Intuitively, only inert statements can be started, active states are resumed instead.
Thus, the ``$?$''s for Go in the last three cases are actually either $\bot$ or $-$ but they cannot be $+$.
Nevertheless, even for active microstate (having Sel$^+$) for which Go can only become Go$^-$, having Go$^-$ still provide more information than Go$^\bot$ and may be necessary for some microstates.
For instance, in the microstate built from $\Sparallel{\state p} q$, even though the overall statement has Sel$^+$, we need Go$^-$ to know that $q$ is not executed and that its completion code will become $\whiteAll$.


We define two colored input notations as convenient shorthands for several cases:
\begin{itemize}
  \item $\gores$: either $\go$ or $\res$, that is, $Go⁺ \lor (Sel⁺ \land Res⁺)$ \\
    The intuition is that the microstate is executed, either started fresh (if inactive) or resumed (if active).
    The \texttt{input\_invariant} invariant expresses that $Go⁺$ entails $Sel^-$, so that we do not need to add $Sel^-$ in the left disjunct.
  \item $\nogores$: either $\nogo$ or $\nores$, that is, $Go^- \land (Sel^- \lor Res^-)$ \\
    The intuition is that the microstate is not executed: neither started nor resumed.
    This definition is a bit more restrictive than $(Go^- \land Sel^-) \lor (Res^- \land Sel^+)$, which is the direct translation of $\nogo$ or $\nores$, because we also require $Go^-$ in the $Sel^+$ case.
    Nevertheless, in that case we know that $Go$ can only become $Go^-$ and having $Go^-$ may be useful as discussed previously.
\end{itemize}

We let $\inC p$, $\outC p$, $Go(p)$, and $Res(p)$ denote respectively the input color, output color, Go wire, and Res wire of a microstate~$p$; and let $Go(\Sin)$ and $Res(\Sin)$ denote the Go and Res wires of an input color $\Sin$.
To change the input color of a microstate~$p$ into $\Sin$, we write $\changeI{\Sin} p$.
We extend this notation to single wires: $\changeI{Go⁺} p$, $\changeI{Res(q)} p$.
For a signal~$s$, we define $E(s)$ to be the value of~$s$ in~$E$ if $s \in \dom(E)$ and $\bot$ otherwise.

\subsection{Intuition of the Microstep Semantics}
\label{sec:-microstep-intuition}

A synchronous digital circuit works by propagating electric potentials through wires and logical gates until all wires have a stable~$0$ or~$1$ value.
Similarly, our microstep semantics works by propagating control information through Esterel statements until reaching some maximal state of information, that is, a state where all Go and Res wires are no longer~$\bot$ and where all completion codes are either $\whiteAll$ (no execution, thus no completion code) or $\black k$ for some integer~$k$ (executed and returning completion code~$k$).
After the microstep execution chain completes, the resulting state $\state{p'}$ can be read from this maximal state of information.
This state is unique as we prove the microstep semantics to be confluent (Theorem~\ref{thm:microstep-confluence}, p.~\pageref{thm:microstep-confluence}).
To measure information and account for its propagation, we use a Scott ordering that we define now.

\paragraph*{Scott ordering on microstates}
\label{sec:scott-order-microstate}
The Scott order $p \le q$\marginlink[\#Mle]{Esterel.Semantics.Microstate}{-2mm} between two microstates~$p$ and~$q$ intuitively means that~$p$ contains no more information (is not more defined) than~$q$.
It is defined recursively, by requiring that~$p$ and~$q$ have the same base statement and Sel values and that any input or output color in~$p$ is no larger (contains no more information) than the corresponding one in~$q$.
On input colors, we use a component-wise ordering of the Scott Booleans Go and Res (that is, $\bot < +$ and $\bot < -$).

The definition of Scott ordering on output colors stems from the intuition that $\black k$ (resp., $\white K$) represents the fact that \MustName{} is $k$ (resp., \CanName{} is $K$).
\begin{center}
  \begin{tabular}[t]{l@{$\;\le\;$}l@{\quad:=\quad}l@{\qquad}l}
    $\white K$ & $\white L$ & $L \subseteq K$ & with more information, more wires are set to $0$ \\
    $\white K$ & $\black k$ & $k \in K$ & the wire to $1$ must be among the ones that are not $0$ \\
    $\black k$ & $\black l$ & $k = l$ & only a single wire can be $1$ \\
    $\black k$ & $\white K$ & False & $\black k$ cannot contain less information than $\white K$ \\
  \end{tabular}
\end{center}
We require the Sel values to be the same to enforce that $p \le q$ entails that~$p$ and~$q$ have the same starting term, in particular the same base statement.
Indeed, on microstates from different terms or different instants, the information available in both has no reason to be compatible and the comparison is not meaningful.

We define a \emph{total} microstate as a microstate where the information is maximal, that is, no input color still contains $\bot$ and all output colors are either $\whiteAll$ or $\black k$ for some~$k$.

\paragraph*{Connection between states and microstates}
The microstep semantics can be related to the state semantics as follows:
starting from a term $\term p$ (that is, an inactive statement~$p$ or an active state~$\state p$) we build a starting microstate representing the state of computation at the beginning of the instant.
Then, after setting its input color, we let the microstep semantics make this microstate evolve until reaching maximal information.
Finally, we convert the final microstate back into a term.
This is shown on Figure~\ref{fig:state-microstate}, where dashed arrows are function application and full arrows represent the microstep semantics.

\begin{figure}[htbp]
  \begin{center}
    \begin{tikzpicture}[every node/.style={minimum width=2em}]
      \node (state) {$\state p$};
      \node[right =5em of state] (micro1) {$\SBin p \Sout$};
      \node[right =4em of micro1] (microstart1) {$\res p \Sout$};
      \node[right =5em of microstart1] (microend1) {$\res p' \black k$};
      \node[right =5em of microend1] (term1) {$\term{p'}$};
      \draw[thick, dashed, ->] (state) -- (micro1) node[midway, above] {{\small \fromstateName}};
      \draw[thick, dashed, ->] (micro1) -- (microstart1) node[midway, above] {{\small \texttt{set\_gr}}};
      \draw[thick, ->] (microstart1) -- (microend1) node[midway, below] {$E$};
      \path (microend1.north west) ++ (0,-0.1) node {$*$};
      \draw[thick, dashed, ->] (microend1) -- (term1) node[midway, above] {{\small \totermName}};
      \node[above =2em of state] (stmt) {$p$};
      \node[above =2em of micro1] (micro2) {$\SWin p \Sout$};
      \node (microstart2) at (microstart1|-stmt) {$\go p \Sout$};
      \node (microend2) at (microend1|-stmt) {$\go p' \black k$};
      \node (term2) at (term1|-stmt) {$\term{p'}$};
      \draw[thick, dashed, ->] (stmt) -- (micro2) node[midway, above] {{\small \fromstmtName}};
      \draw[thick, dashed, ->] (micro2) -- (microstart2) node[midway, above] {{\small \texttt{set\_gr}}};
      \draw[thick, ->] (microstart2) --(microend2) node[midway, below] {$E$};
      \path (microend2.north west) ++ (0,-0.1) node {$*$};
      \draw[thick, dashed, ->] (microend2) -- (term2) node[midway, above] {{\small \totermName}};
    \end{tikzpicture}
  \end{center}
  \caption{Links between the state and microstep semantics.}
  \label{fig:state-microstate}
\end{figure}

The difference between \fromstmtName{} and \fromstateName\marginlink[\#from_stmt]{Esterel.Semantics.Microstate}{-2mm} only lies in how Sel is computed: in \fromstmtName, it is always Sel$^-$ whereas in \fromstateName{} the active parts of a state~$\state p$ have Sel$^+$.
In particular, the overall structure is the same (the one of $p$) and the output colors are also identical.
These output colors represent the statically-computed possible completion codes for each statement and sub-statement and thus, give the number of completion wires in the final circuit.
Notice that the \fromstmtName{} and \fromstateName{} functions perform the computation of the Sel wire.

The \totermName{} function converts a microstate back into a term.
It is well-defined only for microstates with maximal information, since otherwise we do not know which parts are active or not.
For example, $\go \Tpause \black 1$ should become $\state{\Spause}$ whereas $\nogo \Tpause \whiteAll$ should become $\Spause$, so that we cannot translate $\Sin \Tpause \Sout$ in a meaningful way.
This function also performs the effect of the Susp and Kill wires of~\cite{Berry:ConstructiveSemanticsOfPureEsterel} by either freezing subterms (that is, reverting back to the state at the beginning of the instant) or killing a pausing parallel branch (that is, replacing it with its base statement) whenever the other branch raises a trap.

\begin{remark}
As a microstate does not contain information about suspension, the \totermName{} function must take the input event~$E$ in order to decide whether a suspension $\Ssuspend s p$ is triggered or not.
A different solution could be to add a \emph{Susp} component to the control part of input colors, to track the value of the \emph{Susp} wire in the circuit translation of~\cite{Berry:ConstructiveSemanticsOfPureEsterel}.
\end{remark}

Unlike the previous semantics, the microstep semantics does not produce a completion code~$k$ nor an output event~$E'$.
  Although this would be technically possible, it is of little practical interest because microsteps are too fine-grained.
  Indeed, in most microstep rules, the output event and the completion code are undefined, that is, no signal would be emitted and the completion code would be $\bot$.
Furthermore, the completion code~$k$ and the output event~$E'$ can be read off the final (and total) microstate.
The completion code is simply given by the output color and the output event is computed by a function \toeventName\marginlink[\#to_event]{Esterel.Semantics.Microstate}{0mm} by scanning the \Temit statements of the microstate:
\label{def:to_event}
\begin{enumerate}[label=\textit{(\roman*)}]
  \item if there is an executed emitter of~$s$, that is, a subterm $\go \Semit s \black 0$,~$s$ is emitted so we set $s$'s status to $+$;
  \item if all emitters of~$s$ are not executed, that is, they are all $\nogo \Semit s \whiteAll$,~$s$ is not emitted and its status is $-$;
  \item otherwise, some emitter of~$s$ is neither executed nor not executed, that is, it is neither $\go \Semit s \black 0$ nor $\nogo \Semit s \whiteAll$ (hence it is $\Sin \Semit s \white{\{ 0 \}}$), and the status of~$s$ is~$\bot$.
\end{enumerate}
To distinguish these three cases, we can ignore input colors and only look at output colors: $\black 0$ means case
\begin{enumerate*}[label=\textit{(\roman*)}]
  \item ; $\whiteAll$ means case
  \item ; $\white{\{ 0 \}}$ means case
  \item.
\end{enumerate*}
The \toeventName{} function can be applied to any microstate, not only final and total ones.
In particular, it is used to compute the status of local signals (see Section~\ref{sec:micro-local-signal}), thus replacing the \MustName{} and \CanName{} functions.

\subsection{Formal Definition of the Microstep Semantics
\texorpdfstring{\marginlink{Esterel.Semantics.Microstep}{-2mm}}{}}

This (long) section details the 51 rules of the microstep semantics and their meaning.
We give in Figure~\ref{fig:circuit-translation} the circuit translation of~\cite{Berry:ConstructiveSemanticsOfPureEsterel}, as it is a strong inspiration of the microstep rules presented in this section.

\begin{figure}
  \begin{minipage}{0.6\textwidth}
    \begin{subfigure}{\textwidth}
      \includegraphics[width=\textwidth]{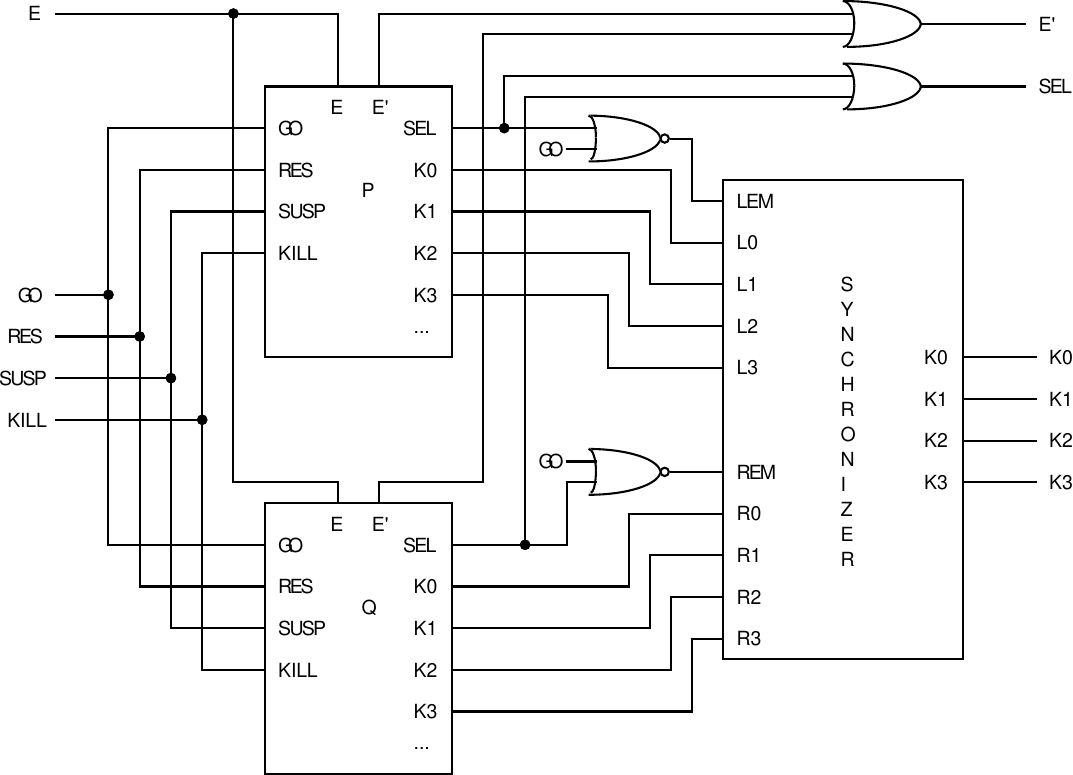}
      \caption{The $\Spar p q$ statement (without synchronizer).}
    \end{subfigure} \\[1em]

    \begin{subfigure}{\textwidth}
      \includegraphics[width=\textwidth]{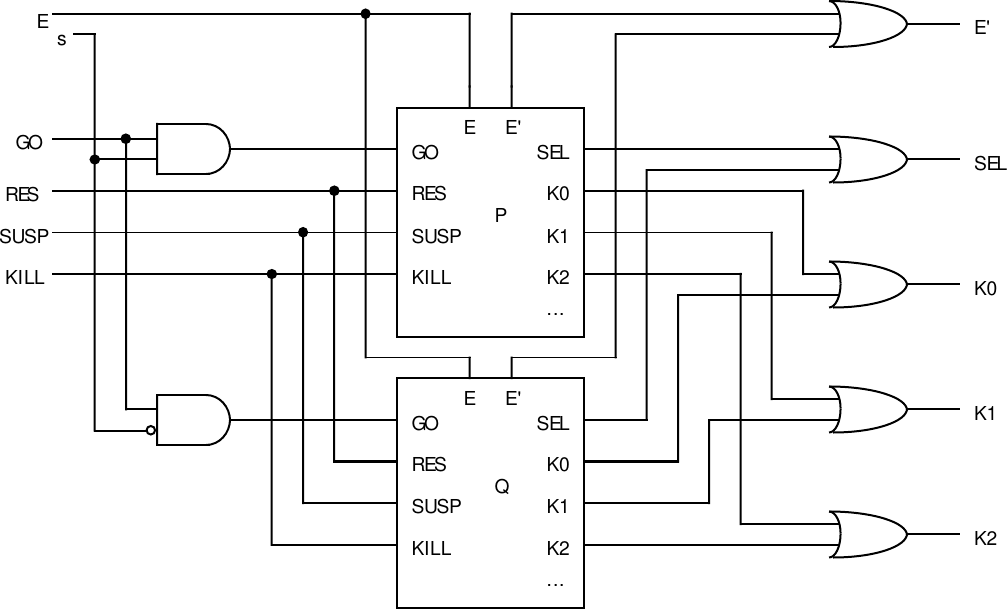}
      \caption{The $\Sif s p q$ statement.}
    \end{subfigure} \\[1em]

    \begin{subfigure}{\textwidth}
      \includegraphics[width=\textwidth]{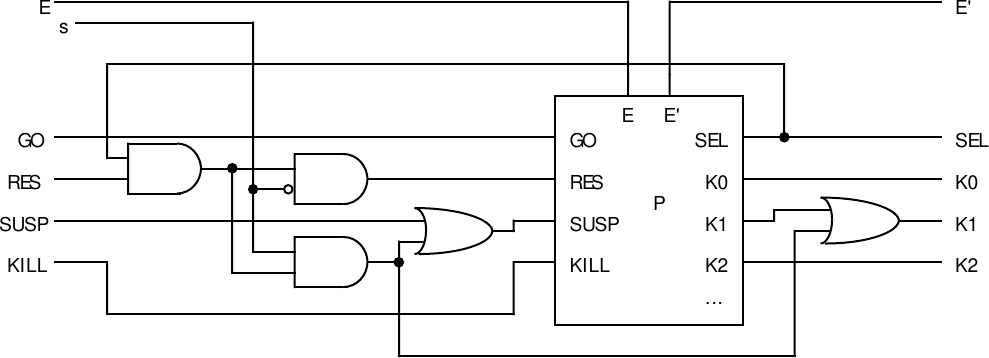}
      \caption{The $\Ssuspend s p$ statement.}
      \label{fig:circuit-translation-suspend}
    \end{subfigure}
  \end{minipage}
  \hfill
  \begin{minipage}{0.35\textwidth}
    \begin{subfigure}{0.45\textwidth}
      \includegraphics[width=\textwidth]{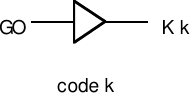}
      \caption{The $\Sexit k$ statement.}
    \end{subfigure}
    \hfill
    \begin{subfigure}{0.47\textwidth}
      \includegraphics[width=\textwidth]{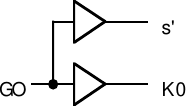}
      \caption{The $\Semit s$ statement.}
    \end{subfigure} \\[1em]

    \begin{subfigure}{0.9\textwidth}
      \includegraphics[width=\textwidth]{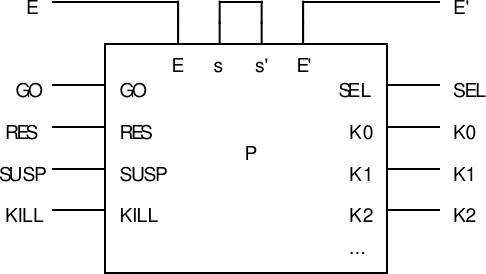}
      \caption{The $\Ssignaldecl s p$ statement.}
    \end{subfigure} \\[1em]

    \begin{subfigure}{0.9\textwidth}
      \includegraphics[width=\textwidth]{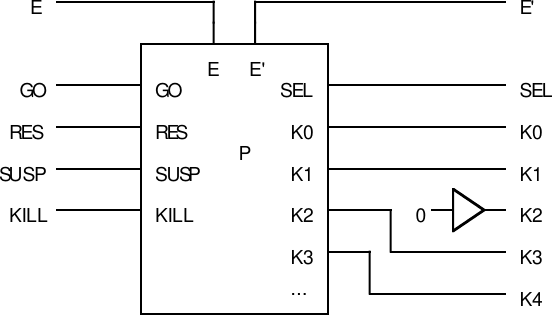}
      \caption{The $\Sshift p$ statement.}
    \end{subfigure} \\[1em]

    \begin{subfigure}{\textwidth}
      \includegraphics[width=\textwidth]{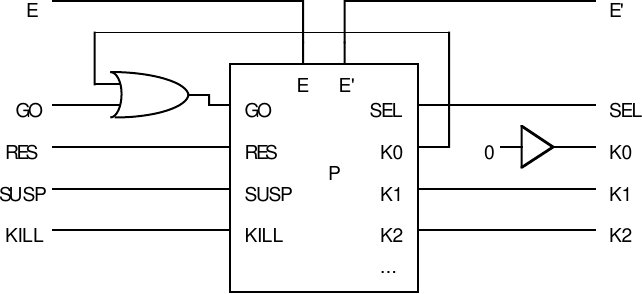}
      \caption{The $\Sloop p$ statement.}
    \end{subfigure} \\[1em]

    \begin{subfigure}{\textwidth}
      \includegraphics[width=\textwidth]{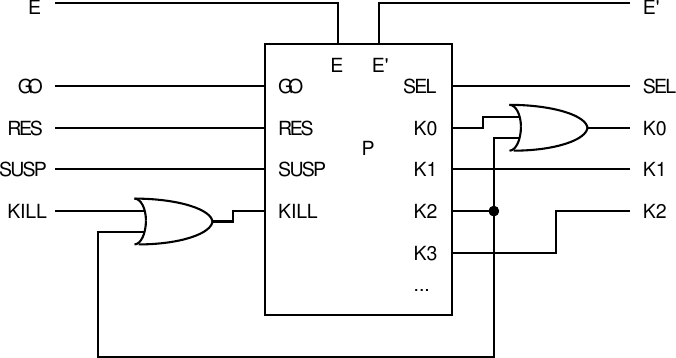}
      \caption{The $\Strap p$ statement.}
    \end{subfigure}
  \end{minipage}
  \vspace{1em}

  \begin{subfigure}{0.5\textwidth}
    \includegraphics[width=\textwidth]{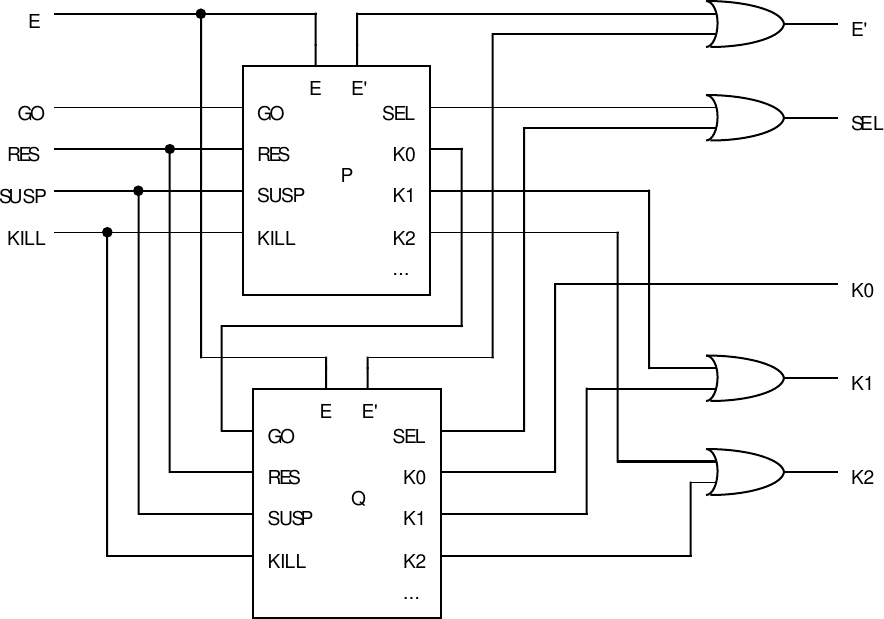}
    \caption{The $\Sseq s p$ statement.}
  \end{subfigure}
  \hfill
  \begin{subfigure}[t]{0.45\textwidth}
    \includegraphics[width=\textwidth]{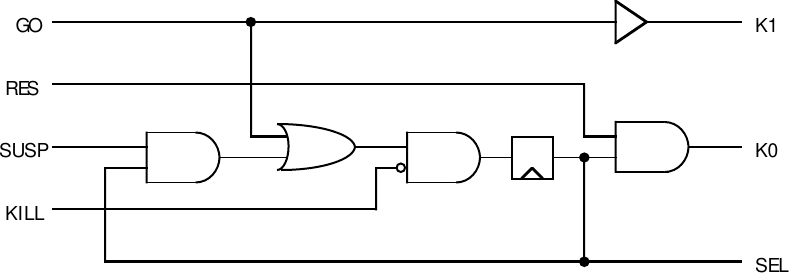}
    \caption{The $\Spause$ statement.}
  \end{subfigure}

  \caption{The Esterel circuit translation, as given by~\cite{Berry:ConstructiveSemanticsOfPureEsterel}.
    The missing $\Sawimm s$ statement can be recovered as a macro: $\Strap{\Sloop{(\Sif s {\Sexit 2} \Spause)}}$.}
  \label{fig:circuit-translation}
\end{figure}

\subsubsection{Elementary microstates}

The statements ``\Snothing'', ``\Spause'', ``\Semit{$s$}'', ``\Sexit{$k$}'', and ``\Sawimm{$s$}'' only have input and output colors, and no sub-statement.
Thus, their microstep rules simply compute their output color depending on their input color and, for $\Sawimm s$, the status of~$s$.

\paragraph{The $\Snothing$, $\Sexit k$, $\Semit s$ rules}
These statements being instantaneous, we know that Sel is always false and that in the circuit translation only the Go wire is useful (and actually represented).
Thus, the only possible input colors are $\SWin$, $\go$ and $\nogo$.
If the statement is started (that is, the input is $\go$), we execute the statement.
If the statement is not started (that is, the input is $\nogo$), we do not execute it.
Otherwise, the statement is neither executed nor not executed (that is, the input is $\SWin$, or in the circuit translation the Go wire is $\bot$) and no rule applies.

\begin{prooftree}
  \AXC{}
  \RightLabel{nothingGo}
  \UIC{$\micro{E}{\go \Snothing \white{\{ 0 \}}}{\go \Snothing \black 0}$}
  \DP
  \hspace{\rulehspace}
  \AXC{}
  \RightLabel{nothingNoGo}
  \UIC{$\micro{E}{\nogo \Snothing \white{\{ 0 \}}}{\nogo \Snothing \whiteAll}$}
  \DP \\[\rulevspace]

  \AXC{}
  \RightLabel{exitGo}
  \UIC{$\micro{E}{\go \Sexit n \white{\{ n+2 \}}}{\go \Sexit n \black{n+2}}$}
  \DP
  \hspace{\rulehspace}
  \AXC{}
  \RightLabel{exitNoGo}
  \UIC{$\micro{E}{\nogo \Sexit n \white{\{ n+2 \}}}{\nogo \Sexit n \whiteAll}$}
  \DP \\[\rulevspace]

  \AXC{}
  \RightLabel{emitGo}
  \UIC{$\micro{E}{\go \Semit s \white{\{ 0 \}}}{\go \Semit s \black 0}$}
  \DP
  \hspace{\rulehspace}
  \AXC{}
  \RightLabel{emitNoGo}
  \UIC{$\micro{E}{\nogo \Semit s \white{\{ 0 \}}}{\nogo \Semit s \whiteAll}$}
\end{prooftree}

\paragraph{The $\Spause$ rules}
If the $\Sin \Spause \Sout$ statement is started, that is, $Go(\Sin)^+$, then the output should be $\black 1$. In order to trigger this rule only when it increases information, we add the precondition $\Sout < \black 1$.
\begin{prooftree}
  \AXC{$Go(\Sin)^+$}
  \AXC{$\Sout < \black 1$}
  \RightLabel{pauseGo}
  \BIC{$\micro{E}{\Sin \Tpause \Sout}{\Sin \Tpause \black 1}$}
\end{prooftree}
\textbf{Remark:}
Thanks to \texttt{input\_invariant} (that is, $Go(\Sin)⁺ \implies Sel⁻$), $Go(\Sin)^+$ actually means $\go$.
We use the former version because the circuit is (rightfully) not checking Sel. 

If the \Tpause{} statement is not started, then its completion code cannot be 1:
\begin{prooftree}
  \AXC{$Go(\Sin)^-$}
  \AXC{$1 \in K$}
  \RightLabel{pauseNoGo}
  \BIC{$\micro{E}{\Sin \Tpause \white K}{\Sin \Tpause \white{K \setminus 1}}$}
\end{prooftree}

Similarly, for Res, we have:
\begin{prooftree}
  \AXC{$\Sout < \black 0$}
  \RightLabel{pauseRes}
  \UIC{$\micro{E}{\res \Tpause \Sout}{\res \Tpause \black 0}$}
  \DP
  \hspace{\rulehspace}
  \AXC{$Res(\Sin)^- \lor Sel^-$}
  \AXC{$0 \in K$}
  \RightLabel{pauseNoRes}
  \BIC{$\micro{E}{\Sin \Tpause \white K}{\Sin \Tpause \white{K \setminus 0}}$}
\end{prooftree}

\paragraph{The $\Sawimm s$ rules}

If the statement is executed, depending on the presence or absence of the signal~$s$, the completion code of $\Sawimm s$ will be either 0 or 1:
\begin{prooftree}
  \AXC{$\Sin = \gores$}
  \AXC{$s⁻ \in E$}
  \AXC{$\Sout < \black 1$}
  \RightLabel{awimmM}
  \TIC{$\micro{E}{\Sin \Sawimm s \Sout}{\Sin \Sawimm s \black 1}$}
  \DP
  \hspace{\rulehspace}
  \AXC{$\Sin = \gores$}
  \AXC{$s⁺ \in E$}
  \AXC{$\Sout < \black 0$}
  \RightLabel{awimmP}
  \TIC{$\micro{E}{\Sin \Sawimm s \Sout}{\Sin \Sawimm s \black 0}$}
\end{prooftree}

When the statement is not executed or the signal is present (resp. absent), we know that the rule for the absence (resp. presence) cannot be triggered, hence we can remove the corresponding completion code from the output color:
\begin{prooftree}
  \AXC{$\Sin = \nogores \lor s⁻ \in E$}
  \AXC{$0 \in K$}
  \RightLabel{awimmNoGo0}
  \BIC{$\micro{E}{\Sin \Sawimm s \white K}{\Sin \Sawimm s \white{K \setminus 0}}$}
  \DP
  \hfill
  \AXC{$\Sin = \nogores \lor s⁺ \in E$}
  \AXC{$1 \in K$}
  \RightLabel{awimmNoGo1}
  \BIC{$\micro{E}{\Sin \Sawimm s \white K}{\Sin \Sawimm s \white{K \setminus 1}}$}
\end{prooftree}
These two rules permit to make progress using the status of~$s$, even before knowing whether the statement is executed or not.

\subsubsection{Compound microstates}

Compound microstates contain input and output colors and some sub-microstates.
Their microstep rules are split into three parts\footnote{This distinction is actually violated by a single case: in a sequence $\Tsequence p q$, starting~$q$ is decided depending on the output color of~$p$, so that this rule is neither a start rule (which it should intuitively be) nor an end rule.}:
\begin{itemize}
  \item a \emph{start} part which propagates control: it converts the input color of the compound microstate into input colors for its sub-microstates;
  \item a \emph{context} part where sub-microstates evolve on their own;
  \item an \emph{end} part where the output colors of sub-microstates are combined into an output color for the compound microstate.
\end{itemize}

In SOS-style semantics, the context rules simply express that execution can happen inside a subpart: if $x$ executes into $x'$, then any term containing $x$ can execute into the same term where $x$ is replaced by $x'$.
They have the same role here, except in the $\Ssignaldecl s p$ construction where they additionally update the value of the local signal~$s$.

The level of detail we want to reach is roughly the gate level, but not always. More precisely, we want to express the functional dependency of each wire on the values of other wires, without being tied to a particular implementation.
For instance, in $\Spar p q$, the output color is the maximum of the output colors of $p$ and $q$, but we do not want to commit to a given implementation of this maximum.
This will permit to try other implementations without changing the specification.

\begin{figure}
  \centering
  \includegraphics{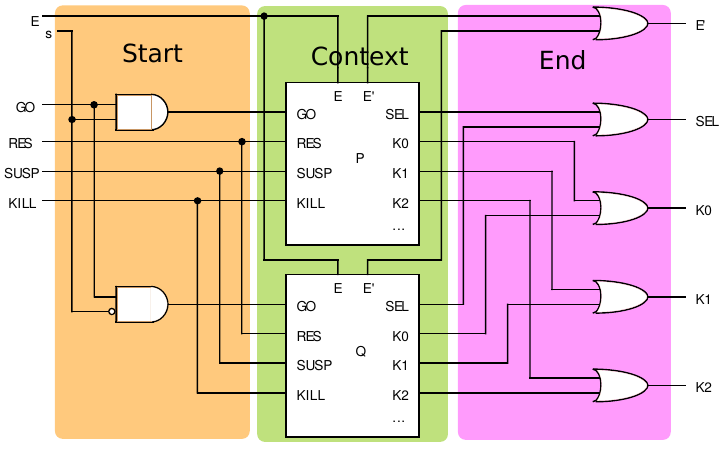}
  \caption{The start, context and end delimitation in the circuit translation of the $\Sif s p q$ statement.}
  \label{fig:if-zones}
\end{figure}

\begin{coqrmk}
  In the Coq formalization, some input rules are split into two: one for Go and another for Res.
  This is merely for convenience: in this case, each input rule only transmits a single input wire (Go or Res), making proofs easier to follow.
  On paper, it seems unnecessary to add more rules, so we keep a single input rule.
  This is the case for rules trapI, shiftI, seqIL, parIL, parIR, signalI.
\end{coqrmk}

\paragraph{The $\Sif s p q$ rules}
The distinction between start, context and end rules is very clear here, see Figure~\ref{fig:if-zones}.
The start part contains two AND gates (one for~$p$ and one for~$q$), and the end part contains a bunch of OR gates for the K$_i$.
(Remember that we ignore the computation of Sel.)
The start rules for~$p$ merely propagate the Res wire from $\Sif s p q$ to~$p$:
\begin{prooftree}
  \AXC{$Res(p) < Res(\Sin)$}
  \RightLabel{ifResL}
  \UIC{$\micro{E}{\Sin (\Sif s p q) \Sout}{\Sin (\Sif s {\changeI{Res(\Sin)} p} q) \Sout}$}
\end{prooftree}
The use of preconditions containing $<$ ensures that a rule can only be triggered if it increases information inside the microstate.

The statement~$p$ is started ($Go(p) = +$) if the overall statement is started ($Go(\Sin (\Sif s p q) \Sout) = +$) and~$s$ is present ($E(s)=+$).
More generally, $Go(p)$ is the conjunction of $Go(\Sin (\Sif s p q) \Sout) = Go(\Sin)$ and $E(s)$.
\begin{prooftree}
  \AXC{$Go(p) < Go(\Sin) \land E(s)$}
  \RightLabel{ifGoL}
  \UIC{$\micro{E}{\Sin (\Sif s p q) \Sout}{\Sin (\Sif s {\changeI{Go^{Go(\Sin) \land E(s)}} p} q) \Sout}$}
\end{prooftree}
The start rules for~$q$ are identical, except that we negate the value of~$s$ in the rule for Go:
\begin{prooftree}
  \AXC{$Go(q) < Go(\Sin) \land \lnot E(s)$}
  \RightLabel{ifGoR}
  \UIC{$\micro{E}{\Sin (\Sif s p q) \Sout}{\Sin (\Sif s p {\changeI{Go^{Go(\Sin) \land \lnot E(s)}} q}) \Sout}$}
  \DP
  \hspace{\rulehspace}
  \AXC{$Res(q) < Res(\Sin)$}
  \RightLabel{ifResR}
  \UIC{$\micro{E}{\Sin (\Sif s p q) \Sout}{\Sin (\Sif s p {\changeI{Res(\Sin)} q}) \Sout}$}
\end{prooftree}

The two context rules permit executing either~$p$ or~$q$ ; the end rule combines the completion codes of~$p$ and~$q$ by taking their union.
\begin{prooftree}
  \AXC{$\micro E p {p'}$}
  \RightLabel{ifCL}
  \UIC{$\micro E {\Sin (\Sif s p q) \Sout}{\Sin (\Sif s {p'} q) \Sout}$}
  \DP
  \hspace{\rulehspace}
  \AXC{$\micro E q {q'}$}
  \RightLabel{ifCR}
  \UIC{$\micro E {\Sin (\Sif s p q) \Sout}{\Sin (\Sif s p {q'}) \Sout}$}
  \DP \\[\rulevspace]

  \AXC{$\Sout < \outC p \cup \outC q$}
  \RightLabel{ifE}
  \UIC{$\micro{E}{\Sin (\Sif s p q) \Sout}{\Sin (\Sif s p q) (\outC p \cup \outC q)}$}
\end{prooftree}

\paragraph{The $\Ssuspend s p$ rules}

When started, the $\Ssuspend s p$ statement starts its sub-statement~$p$.
Thus, Go is simply transmitted to the sub-statement.
When resumed, the $\Ssuspend s p$ statement resumes its sub-statement~$p$ only when~$s$ is absent.
In order not to resume an inactive statement, $Sel(p)$ is also checked.
\begin{prooftree}
  \AXC{$Go(p) < Go(\Sin)$}
  \RightLabel{suspendGo}
  \UIC{$\micro{E}{\Sin \Ssuspend s p \Sout}{\Sin \Ssuspend s {(\changeI{Go(\Sin)} p)} \Sout}$}
  \DP 
  \hspace{\rulehspace}
  \AXC{$Res(p) < \overbrace{Res(\Sin) \land Sel(p) \land \lnot E(s)}^{res}$}
  \RightLabel{suspendRes}
  \UIC{$\micro{E}{\Sin (\Ssuspend s p) \Sout}{\Sin \Ssuspend s {(\changeI{Res^{res}} p)} \Sout}$}
\end{prooftree}

The context rule executes~$p$:
\begin{prooftree}
  \AXC{$\micro E p {p'}$}
  \RightLabel{suspendC}
  \UIC{$\micro E {\Sin (\Ssuspend s p) \Sout}{\Sin (\Ssuspend s {p'}) \Sout}$}
\end{prooftree}

For the end rule, we first define an auxiliary function $\SuspNow{bo}{\Sout}$ to compute the completion code:
\[
  \SuspNow{bo}{\Sout} =
  \begin{cases}
    \black 1                   & \text{if } bo = + \\
    \Sout                      & \text{if } bo = - \\
    \white{\{ 1 \}} \cup \Sout & \text{if } bo = \bot \\
  \end{cases}
\]
Essentially, $bo$ represent suspension of the sub-statement~$p$: if $bo$ is $+$,~$p$ is suspended and the completion code is 1;
if $bo$ is $-$,~$p$ is executed normally;
if $bo$ is $\bot$, it adds the possibility of returning 1 to the completion code of~$p$.
The suspension of~$p$ only happens when~$p$ should have been resumed, that is, $Res(\Sin) \land Sel(p) = +$, but~$s$ was present, that is, $E(s) = +$.
\begin{prooftree}
  \AXC{$\Sout < \SuspNow{\overbrace{Res(\Sin) \land Sel(p) \land E(s)}^{susp}}{\outC p}$}
  \RightLabel{suspendE}
  \UIC{$\micro E {\Sin (\Ssuspend s p) \Sout}{\Sin (\Ssuspend s {p}) \SuspNow{susp}{\outC p}}$}
\end{prooftree}
For simplicity, we introduce a particular case of this rule, where $bo$ is set to~$\bot$:
\begin{prooftree}
  \AXC{$\Sout < \SuspNow{\bot}{\outC p}$}
  \RightLabel{suspendEKO}
  \UIC{$\micro E {\Sin (\Ssuspend s p) \Sout}{\Sin (\Ssuspend s {p}) \SuspNow{\bot}{\outC p}}$}
\end{prooftree}
Without this rule, we cannot ignore the value of~$s$ in~$E$, so that compatibility with ordering on events (last property of Theorem~\ref{thm:microstep-properties}) is lost.
More importantly, the context rule for $\Ssignaldecl s p$ becomes false.

\paragraph{The $\Strap p$ and $\Sshift p$ rules}

The $\Strap p$ and $\Sshift p$ statements only have an effect on the completion code of~$p$, thus the input color is simply transmitted.
\begin{prooftree}
  \AXC{$in(p) < \Sin$}
  \RightLabel{trapI}
  \UIC{$\micro{E}{\Sin \Strap p \Sout}{\Sin \Strap{\changeI \Sin p} \Sout}$}
  \DP
  \hfill
  \AXC{$\micro{E} p {p'}$}
  \RightLabel{trapC}
  \UIC{$\micro{E}{\Sin \Strap{\strut p} \Sout}{\Sin \Strap{p'} \Sout}$}
  \DP
  \hfill
  \AXC{$\Sout < \Kdown{\outC p}$}
  \RightLabel{trapE}
  \UIC{$\micro{E}{\Sin \Strap p \Sout}{\Sin \Strap p \Kdown{\outC p}}$}
  \DP \\[\rulevspace]

  \AXC{$in(p) < \Sin$}
  \RightLabel{shiftI}
  \UIC{$\micro{E}{\Sin (\Sshift p) \Sout}{\Sin (\Sshift{\changeI \Sin p}) \Sout}$}
  \DP
  \hfill
  \AXC{$\micro{E} p {p'}$}
  \RightLabel{shiftC}
  \UIC{$\micro{E}{\Sin (\Sshift p) \Sout}{\Sin (\Sshift{p'}) \Sout}$}
  \DP
  \hfill
  \AXC{$\Sout < \Kup{\outC p}$}
  \RightLabel{shiftE}
  \UIC{$\micro{E}{\Sin (\Sshift p) \Sout}{\Sin (\Sshift p) \Kup{\outC p}}$}
\end{prooftree}

\paragraph{The $\Tsequence p q$ rules}

In the start rules, Res is forwarded to both~$p$ and~$q$ whereas Go is transmitted only to~$p$.
\begin{prooftree}
  \AXC{$in(p) < \Sin$}
  \RightLabel{seqIL}
  \UIC{$\micro E {\Sin (\Ssequence p q) \Sout}{\Sin (\Ssequence {(\changeI \Sin p)} q) \Sout}$}
  \DP
  \hspace{\rulehspace}
  \AXC{$Res(q) < Res(\Sin)$}
  \RightLabel{seqResR}
  \UIC{$\micro E {\Sin (\Ssequence p q) \Sout}{\Sin (\Ssequence p {\changeI{Res(\Sin)} q}) \Sout}$}
\end{prooftree}

The sub-statement~$q$ is started depending on the completion code of~$p$, that is, the value of $Go(q)$ is set to $+$ or $-$ respectively when $\outC p$ is $\black 0$ or cannot become $\black 0$.
\begin{prooftree}
  \AXC{$\outC p = \black 0$}
  \AXC{$Go(q) = \bot$}
  \RightLabel{seqGoR}
  \BIC{$\micro{E}{\Sin (\Ssequence p q) \Sout}{\Sin (\Ssequence p {\changeI{Go⁺} q}) \Sout}$}
  \DP
  \hspace{\rulehspace}
  \AXC{$\outC p \not\le \black 0$}
  \AXC{$Go(q) = \bot$}
  \RightLabel{seqNoGoR}
  \BIC{$\micro{E}{\Sin (\Ssequence p q) \Sout}{\Sin (\Ssequence p {\changeI{Go⁻} q}) \Sout}$}
\end{prooftree}
In the second rule, we know that~$q$ cannot be started as soon as the completion code of~$p$ cannot be~$0$, regardless of other information we might know about $\outC p$.
Hence, we use $\outC p \not\le \black 0$.

As for all statements, the context rules permit execution of sub-statements.
\begin{prooftree}
  \AXC{$\micro E p {p'}$}
  \RightLabel{seqCL}
  \UIC{$\micro E {\Sin (\Ssequence p q) \Sout}{\Sin (\Ssequence{p'} q) \Sout}$}
  \DP
  \hspace{\rulehspace}
  \AXC{$\micro E q {q'}$}
  \RightLabel{seqCR}
  \UIC{$\micro E {\Sin (\Ssequence p q) \Sout}{\Sin (\Ssequence p {q'}) \Sout}$}
\end{prooftree}

For the end rule, we remove~$0$ from the possible completion codes of~$p$, written $\outC p \setminus 0$, before taking the union with the ones from~$q$.
More precisely, if $\outC p$ is $\white K$, we return $\white{K \setminus 0}$, if $\outC p$ is $\black 0$, we return $\whiteAll$, and if $\outC p$ is $\black k$ with $k\neq 0$, we return it unchanged.
\begin{prooftree}
  \AXC{$\Sout < (\outC p \setminus 0) \cup \outC q$}
  \RightLabel{seqE}
  \UIC{$\micro{E}{\Sin (\Sif s p q) \Sout}{\Sin (\Sif s p q) (\outC p \setminus 0) \cup \outC q)}$}
\end{prooftree}

\begin{coqrmk}
  In the Coq development, the operation $\Sout \setminus 0$ is written \texttt{SEQrestrict}$(\Sout)$.
\end{coqrmk}

\paragraph{The $\Sloop p$ rules}

A loop $\Sloop p$ starts its body either when the loop itself starts or when the body finishes an iteration, that is, when $\outC p = \black 0$.
Getting inspiration from the circuit translation, we use an OR gate between the Go wire of $\Sloop p$ and the 0-th component of the output color of~$p$, written $\outC p[0]$.
Resumption is simply propagated.
The end rule propagates the output color of~$p$ after forcing the 0-th component to false.
\begin{prooftree}
  \AXC{$Go(p) < Go(\Sin) \lor \outC p[0]$}
  \RightLabel{loopGo}
  \UIC{$\micro{E}{\Sin \Sloop p \Sout}{\Sin \Sloop{\changeI{Go(\Sin) \lor \outC p[0]} p} \Sout}$}
  \DP
  \hspace{\rulehspace}
  \AXC{$Res(p) < Res(\Sin)$}
  \RightLabel{loopRes}
  \UIC{$\micro{E}{\Sin \Sloop p \Sout}{\Sin \Sloop{\changeI{Res(\Sin)} p} \Sout}$}
  \DP \\[\rulevspace]

  \AXC{$\micro E p p'$}
  \RightLabel{loopC}
  \UIC{$\micro{E}{\Sin \Sloop p \Sout}{\Sin \Sloop{p'} \Sout}$}
  \DP
  \hspace{\rulehspace}
  \AXC{$\Sout < \outC p \setminus 0$}
  \RightLabel{loopE}
  \UIC{$\micro{E}{\Sin \Sloop p \Sout}{\Sin \Sloop p \; (\outC p \setminus 0)}$}
\end{prooftree}

\paragraph{The \texorpdfstring{$\Sparallel p q$}{p || q} rules}
Input and context rules simply propagate control.
\begin{prooftree}
  \AXC{$in(p) < \Sin$}
  \RightLabel{parIL}
  \UIC{$\micro E {\Sin (\Sparallel p q) \Sout}{\Sin (\Sparallel{\changeI \Sin p} q) \Sout}$}
  \DP
  \hspace{\rulehspace}
  \AXC{$in(q) < \Sin$}
  \RightLabel{parIR}
  \UIC{$\micro E {\Sin (\Sparallel p q) \Sout}{\Sin (\Sparallel p {\changeI \Sin q}) \Sout}$}
  \DP \\[\rulevspace]

  \AXC{$\micro E p {p'}$}
  \RightLabel{parCL}
  \UIC{$\micro E {\Sin (\Sparallel p q) \Sout}{\Sin (\Sparallel{p'} q) \Sout}$}
  \DP
  \hspace{\rulehspace}
  \AXC{$\micro E q {q'}$}
  \RightLabel{parCR}
  \UIC{$\micro E {\Sin (\Sparallel p q) \Sout}{\Sin (\Sparallel p {q'}) \Sout}$}
\end{prooftree}

For the end rule, the synchronizer performs a max of the completion codes of~$p$ and~$q$, except when one of them is $\whiteAll$, in which case it returns the other one since this corresponds to the case when one branch is dead and the other one is alive.
\begin{prooftree}
  \AXC{$\Sout < \overbrace{\synchronize{Sel(p)}{Sel(q)}{\outC p}{\outC q}}^{synch}$}
  \RightLabel{parE}
  \UIC{$\micro E {\Sin (\Sparallel p q) \Sout}{\Sin (\Sparallel p q) \mathop{synch}}$}
\end{prooftree}
with
\[
  \synchronize{Sel(p)}{Sel(q)}{\Sout_p}{\Sout_q} =
  \begin{cases}
    \MAX{\Sout_p}{\Sout_q} & \text{when \quad} Sel(p) = Sel(q) \\
    \Sout_p & \text{when \quad} Sel(q) = + \text{\; and\; } Sel(q) = - \\
    \Sout_q & \text{when \quad} Sel(p) = - \text{\; and\; } Sel(q) = + \\
  \end{cases}
\]
and
\[
  \begin{array}{l@{\quad=\quad}l}
    \MAX{\black{k}}{\black{k'}} & \black{\max(k, k')} \\
    \MAX{\white K}{\black k}    & \white{\MAX{K}{\{k\}}} \\
    \MAX{\black k}{\white K}    & \white{\MAX{\{k\}}{K}} \\
    \MAX{\white K}{\white{K'}}  & \white{\MAX{K}{K'}} \\
  \end{array}
\]

\paragraph{The $\Ssignaldecl s p$ rules}
\label{sec:micro-local-signal}
The input and output colors are simply propagated by the start and end rules.
\begin{prooftree}
  \AXC{$\inC p < \Sin$}
  \RightLabel{signalI}
  \UIC{$\micro E {\Sin \Ssignaldecl s p \Sout}{\Ssignaldecl s {(\changeI \Sin \Sin p)} \Sout}$}
  \DP
  \hspace{\rulehspace}
  \AXC{$\Sout < \outC p$}
  \RightLabel{signalE}
  \UIC{$\micro E {\Sin (\Ssignaldecl s p) \Sout}{\Sin (\Ssignaldecl s p) \outC p}$}
\end{prooftree}

For the context rule, we first check the emitters of~$s$ inside~$p$ to deduce the status of~$s$.
In the circuit, this only amounts to wire propagation and a big OR gate to combine all emitters into the signal status.
Here, we use the \toeventName{} function (see the end of Section~\ref{def:to_event}) which scans a microstate and identifies emitters and whether they are executed, not executed or still pending.
\begin{prooftree}
  \AXC{$b := (\toevent{p}{ s^\bot*E})(s)$}
  \AXC{$\micro{\addEvent s b E} p {p'}$}
  \RightLabel{signalC}
  \BIC{$\micro E {\Sin (\Ssignaldecl s p) \Sout}{\Sin (\Ssignaldecl s {p'}) \Sout}$}
\end{prooftree}

\subsubsection{ABROi in the microstep semantics}

\newcommand\marked[1]{\underline{#1}}
As the microstep semantics requires a lot more steps than previous semantics, for the sake of space, we will consider only the first instant (where B is present), which still requires 46 microsteps!
We combine several microsteps at once when this does not hinder understanding too much and, to improve readability, we underline the parts that change and we define two abbreviations:
\begin{align*}
  \text{check} &=
  \Ssequence{\SWin \Spause \white{\{0, 1\}}}
            {\SWin \Sloop{\big( \SWin
                (\Sif{R}{\SWin \Sexit{2} \white{\{2\}}}
                        {\SWin \Spause \white{\{0, 1\}}}
                ) \white{\{0, 1, 2\}}
            \big) } \white{\{1, 2\}}} \\
  \text{body} &= \Ssequence{\SWin \big(
                        \Ssequence{\SWin (
                          \Sparallel{\SWin \Sawimm{A}  \white{\{0, 1\}}}
                                    {\SWin \Sawimm{B} \white{\{0, 1\}}}
                          \,) \white{\{0, 1\}} }
                          {\SWin \Semit{O} \white{\{0\}}}
                        \big) \white{\{0, 1\}}}
                        {\SWin \Sloop{(\SWin \Spause \white{\{1\}})} \white{\{1\}}}
\end{align*}

\[
  \begin{array}{@{}r@{\;}l@{}}
    \text{check} &=
      \Ssequence{\SWin \Spause \white{\{0, 1\}}}
            {\SWin \Sloop{
            \big( \SWin
            (\Sif{R}{\SWin \Sexit{2} \white{\{2\}}}
                    {\SWin \Spause \white{\{0, 1\}}}
            ) \white{\{0, 1, 2\}}
                    \big) } \white{\{1, 2\}}} \\
    &\microexp{2}{\{ B \}}{}{
       \Ssequence{\SWin \Spause \marked{\white{\{1\}}}}
            {\SWin\Sloop{
            \big( \SWin
            (\Sif{R}{\marked{\nogo} \Sexit{2} \white{\{2\}}}
                    {\SWin \Spause \white{\{0, 1\}}}
            ) \white{\{0, 1, 2\}}
                    \big) } \white{\{1, 2\}} }} \\
    &\micro{\{ B \}}{}{
       \Ssequence{\SWin \Spause \white{\{1\}}}
            {\SWin\Sloop{
            \big( \SWin
            (\Sif{R}{\nogo \Sexit{2} \marked{\whiteAll}}
                    {\SWin \Spause \white{\{0, 1\}}}
            ) \white{\{0, 1, 2\}}
                    \big) } \white{\{1, 2\}} }} \\
    &\microexp{2}{\{ B \}}{}{
       \Ssequence{\SWin \Spause \white{\{1\}}}
            {\SWin\Sloop{
            \big( \SWin
            (\Sif{R}{\nogo \Sexit{2} \whiteAll}
                    {\SWin \Spause \white{\{0, 1\}}}
            ) \marked{\white{\{0, 1\}}}
                    \big) } \marked{\white{\{1\}}} }} = \text{check'} \\
    \text{body} &= \Ssequence{\SWin \big(
                        \Ssequence{\SWin (
                          \Sparallel{\SWin \Sawimm{A}  \white{\{0, 1\}}}
                                    {\SWin \Sawimm{B} \white{\{0, 1\}}}
                          \,) \white{\{0, 1\}} }
                          {\SWin \Semit{O} \white{\{0\}}}
                        \big) \white{\{0, 1\}}}
         {\SWin \Sloop{(\SWin \Spause \white{\{1\}})} \white{\{1\}}} \\
         &\microexp{2}{\{ B \}}{}{
             \Ssequence{\SWin \big(
                        \Ssequence{\SWin (
                          \Sparallel{\SWin \Sawimm{A}  \marked{\white{\{1\}}}}
                                    {\SWin \Sawimm{B} \marked{\white{\{0\}}}}
                          \,) \white{\{0, 1\}} }
                          {\SWin \Semit{O} \white{\{0\}}}
                        \big) \white{\{0, 1\}}}
                       {\SWin \Sloop{(\SWin \Spause \white{\{1\}})} \white{\{1\}}}
         } = \text{body'} \\
    \text{ABROi} &=
    \Sloop{
      \left( \SWin \Strap{ \SWin
        \big( \Sparallel{ \SWin \text{check}
           \white{\{1, 2\}}
        }
            {\SWin (\Ssequence{\SWin \text{body} \white{\{1\}}
              }
                     {\SWin \Sexit{2} \white{\{2\}}}
                    ) \white{\{1, 2\}} } \big) \white{\{1, 2\}}
      } \white{\{0, 1\}} \right)
    } \\
    &\microexp{7}{\{ B \}}{}{
    \Sloop{
      \left( \SWin \Strap{ \SWin
        \big( \Sparallel{ \SWin \marked{\text{check'}}
           \white{\{1, 2\}}
        }
            {\SWin (\Ssequence{\SWin \marked{\text{body'}} \white{\{1\}}
              }
                     {\SWin \Sexit{2} \white{\{2\}}}
                    ) \white{\{1, 2\}} } \big) \white{\{1, 2\}}
      } \white{\{0, 1\}} \right)
    }} \\
    &\microexp{2}{\{ B \}}{}{
    \Sloop{
      \left( \SWin \Strap{ \SWin
        \big( \Sparallel{ \SWin \text{check'}
           \marked{\white{\{1\}}}
        }
            {\SWin (\Ssequence{\SWin \text{body'} \white{\{1\}}
              }
                     {\marked{\nogo} \Sexit{2} \white{\{2\}}}
                    ) \white{\{1, 2\}} } \big) \white{\{1, 2\}}
      } \white{\{0, 1\}} \right)
    }} \\
    &\micro{\{ B \}}{}{
    \Sloop{
      \left( \SWin \Strap{ \SWin
        \big( \Sparallel{ \SWin \text{check'}
           \white{\{1\}}
        }
            {\SWin (\Ssequence{\SWin \text{body'} \white{\{1\}}
              }
                     {\nogo \Sexit{2} \marked{\whiteAll}}
                    ) \white{\{1, 2\}} } \big) \white{\{1, 2\}}
      } \white{\{0, 1\}} \right)
    }} \\
    &\microexp{3}{\{ B \}}{}{
    \Sloop{
      \left( \SWin \Strap{ \SWin
        \big( \Sparallel{ \SWin \text{check'}
           \white{\{1\}}
        }
            {\SWin (\Ssequence{\SWin \text{body'} \white{\{1\}}
              }
                     {\nogo \Sexit{2} \whiteAll}
                    ) \marked{\white{\{1\}}} } \big) \marked{\white{\{1\}}}
      } \marked{\white{\{1\}}} \right)
    }} = \text{ABROi'} \\
  \end{array}
\]

Notice that at this point, we have only used the statuses of signals (A, B and R) but the input color of ABROi has not been set yet.
Let us use now the fact that execution is started.
\[
  \begin{array}{@{}r@{\;}l@{}}
    \go \text{ABROi'} \white{\{1\}}
    &\microexp{5}{\{ B \}}{}{
    \go \Sloop{
      \left( \marked{\go} \Strap{ \marked{\go}
        \big( \Sparallel{ \marked{\go} \text{check'}
           \white{\{1\}}
        }
            {\marked{\go} (\Ssequence{\marked{\go} \text{body'} \white{\{1\}}
              }
                     {\nogo \Sexit{2} \whiteAll}
                    ) \white{\{1\}} } \big) \white{\{1\}}
      } \white{\{1\}} \right)
    } \white{\{1\}}
    } \\
    \go \text{check'} \white{\{1\}}
    &\micro{\{ B \}}{}{
      \go \left (
      \Ssequence{\marked{\go} \Spause \white{\{1\}}}
            {\SWin\Sloop{
            \big( \SWin
            (\Sif{R}{\nogo \Sexit{2} \whiteAll}
                    {\SWin \Spause \white{\{0, 1\}}}
            ) \white{\{0, 1\}}
                    \big) } \white{\{1\}} }
    \right) \white{\{1\}}} \\
    &\micro{\{ B \}}{}{
      \go \left (
      \Ssequence{\go \Spause \marked{\black{1}}}
            {\SWin\Sloop{
            \big( \SWin
            (\Sif{R}{\nogo \Sexit{2} \whiteAll}
                    {\SWin \Spause \white{\{0, 1\}}}
            ) \white{\{0, 1\}}
                    \big) } \white{\{1\}} }
    \right) \white{\{1\}}} \\
    &\microexp{3}{\{ B \}}{}{
      \go \left (
      \Ssequence{\go \Spause \black{1}}
            {\marked{\nogo} \Sloop{
            \big( \marked{\nogo}
            (\Sif{R}{\nogo \Sexit{2} \whiteAll}
                    {\marked{\nogo} \Spause \white{\{0, 1\}}}
            ) \white{\{0, 1\}}
                    \big) } \white{\{1\}} }
    \right) \white{\{1\}}} \\
    &\microexp{3}{\{ B \}}{}{
      \go \left (
      \Ssequence{\go \Spause \black{1}}
            {\nogo \Sloop{
            \big( \nogo
            (\Sif{R}{\nogo \Sexit{2} \whiteAll}
                    {\nogo \Spause \marked{\whiteAll}}
            ) \marked{\whiteAll}
                    \big) } \marked{\whiteAll}}
    \right) \white{\{1\}}} \\
    &\micro{\{ B \}}{}{
      \go \left (
      \Ssequence{\go \Spause \black{1}}
            {\nogo \Sloop{
            \big( \nogo
            (\Sif{R}{\nogo \Sexit{2} \whiteAll}
                    {\nogo \Spause \whiteAll}
            ) \whiteAll
                    \big) } \whiteAll}
    \right) \marked{\black{1}}} \\
  \end{array}
\]
\[
  \begin{array}{@{}r@{\;}l@{}}
    \go \text{body'} \white{\{1\}}
    &\microexp{4}{\{ B \}}{}{
      \go \big( \Ssequence{\marked{\go} \big(
                        \Ssequence{\marked{\go} (
                          \Sparallel{\marked{\go} \Sawimm{A}  \white{\{1\}}}
                                    {\marked{\go} \Sawimm{B} \white{\{0\}}}
                          \,) \white{\{0, 1\}} }
                          {\SWin \Semit{O} \white{\{0\}}}
                        \big) \white{\{0, 1\}}}
        {\SWin \Sloop{(\SWin \Spause \white{\{1\}})} \white{\{1\}}}
        \big) \white{\{1\}}
        } \\
    &\microexp{2}{\{ B \}}{}{
      \go \big( \Ssequence{\go \big(
                        \Ssequence{\go (
                          \Sparallel{\go \Sawimm{A} \marked{\black{1}}}
                                    {\go \Sawimm{B} \marked{\black{0}}}
                          \,) \white{\{0, 1\}} }
                          {\SWin \Semit{O} \white{\{0\}}}
                        \big) \white{\{0, 1\}}}
        {\SWin \Sloop{(\SWin \Spause \white{\{1\}})} \white{\{1\}}}
        \big) \white{\{1\}}
        } \\
    &\micro{\{ B \}}{}{
      \go \big( \Ssequence{\go \big(
                        \Ssequence{\go (
                          \Sparallel{\go \Sawimm{A} \black{1}}
                                    {\go \Sawimm{B} \black{0}}
                          \,) \marked{\black{1}} }
                          {\SWin \Semit{O} \white{\{0\}}}
                        \big) \white{\{0, 1\}}}
        {\SWin \Sloop{(\SWin \Spause \white{\{1\}})} \white{\{1\}}}
        \big) \white{\{1\}}
        } \\
    &\micro{\{ B \}}{}{
      \go \big( \Ssequence{\go \big(
                        \Ssequence{\go (
                          \Sparallel{\go \Sawimm{A} \black{1}}
                                    {\go \Sawimm{B} \black{0}}
                          \,) \black{1} }
                          {\marked{\nogo} \Semit{O} \white{\{0\}}}
                        \big) \white{\{0, 1\}}}
        {\SWin \Sloop{(\SWin \Spause \white{\{1\}})} \white{\{1\}}}
        \big) \white{\{1\}}
        } \\
    &\micro{\{ B \}}{}{
      \go \big( \Ssequence{\go \big(
                        \Ssequence{\go (
                          \Sparallel{\go \Sawimm{A} \black{1}}
                                    {\go \Sawimm{B} \black{0}}
                          \,) \black{1} }
                          {\nogo \Semit{O} \marked{\whiteAll}}
                        \big) \white{\{0, 1\}}}
        {\SWin \Sloop{(\SWin \Spause \white{\{1\}})} \white{\{1\}}}
        \big) \white{\{1\}}
        } \\
    &\micro{\{ B \}}{}{
      \go \big( \Ssequence{\go \big(
                        \Ssequence{\go (
                          \Sparallel{\go \Sawimm{A} \black{1}}
                                    {\go \Sawimm{B} \black{0}}
                          \,) \black{1} }
                          {\nogo \Semit{O} \whiteAll}
                        \big) \marked{\black{1}}}
        {\SWin \Sloop{(\SWin \Spause \white{\{1\}})} \white{\{1\}}}
        \big) \white{\{1\}}
        } \\
    &\microexp{2}{\{ B \}}{}{
      \go \big( \Ssequence{\go \big(
                        \Ssequence{\go (
                          \Sparallel{\go \Sawimm{A} \black{1}}
                                    {\go \Sawimm{B} \black{0}}
                          \,) \black{1} }
                          {\nogo \Semit{O} \whiteAll}
                        \big) \black{1}}
        {\marked{\nogo} \Sloop{(\marked{\nogo} \Spause \white{\{1\}})} \white{\{1\}}}
        \big) \white{\{1\}}
        } \\
    &\microexp{2}{\{ B \}}{}{
      \go \big( \Ssequence{\go \big(
                        \Ssequence{\go (
                          \Sparallel{\go \Sawimm{A} \black{1}}
                                    {\go \Sawimm{B} \black{0}}
                          \,) \black{1} }
                          {\nogo \Semit{O} \whiteAll}
                        \big) \black{1}}
        {\nogo \Sloop{(\nogo \Spause \marked{\whiteAll})} \marked{\whiteAll}}
        \big) \white{\{1\}}
        } \\
    &\micro{\{ B \}}{}{
      \go \big( \Ssequence{\go \big(
                        \Ssequence{\go (
                          \Sparallel{\go \Sawimm{A} \black{1}}
                                    {\go \Sawimm{B} \black{0}}
                          \,) \black{1} }
                          {\nogo \Semit{O} \whiteAll}
                        \big) \black{1}}
        {\nogo \Sloop{(\nogo \Spause \whiteAll)} \whiteAll}
        \big) \marked{\black{1}}
        } \\
  \end{array}
\]
We observe that the final microstates for check' and body' are total, so cannot execute further.
Let us call them check'' and body'' respectively.
Finally, we get for ABROi:
\[
  \begin{array}{@{}r@{\;}l@{}}
    \go \text{ABROi'} \white{\{1\}}
    &\microexp{5}{\{ B \}}{}{
    \go \Sloop{
      \left( \marked{\go} \Strap{ \marked{\go}
        \big( \Sparallel{ \marked{\go} \text{check'}
           \white{\{1\}}
        }
            {\marked{\go} (\Ssequence{\marked{\go} \text{body'} \white{\{1\}}
              }
                     {\nogo \Sexit{2} \whiteAll}
                    ) \white{\{1\}} } \big) \white{\{1\}}
      } \white{\{1\}} \right)
    } \white{\{1\}}
    } \\
    &\microexp{24}{\{ B \}}{}{
    \go \Sloop{
      \left( \go \Strap{ \go
        \big( \Sparallel{ \go \marked{\text{check''}}
           \black{1}
        }
            {\go (\Ssequence{\go \marked{\text{body''}} \black{1}
              }
                     {\nogo \Sexit{2} \whiteAll}
                    ) \white{\{1\}} } \big) \white{\{1\}}
      } \white{\{1\}} \right)
    } \white{\{1\}}
    } \\
    &\micro{\{ B \}}{}{
    \go \Sloop{
      \left( \go \Strap{ \go
        \big( \Sparallel{ \go \text{check''}
           \black{1}
        }
            {\go (\Ssequence{\go \text{body''} \black{1}
              }
                     {\nogo \Sexit{2} \whiteAll}
                    ) \marked{\black{1}} } \big) \white{\{1\}}
      } \white{\{1\}} \right)
    } \white{\{1\}}
    } \\
    &\microexp{3}{\{ B \}}{}{
    \go \Sloop{
      \left( \go \Strap{ \go
        \big( \Sparallel{ \go \text{check''}
           \black{1}
        }
            {\go (\Ssequence{\go \text{body''} \black{1}
              }
                     {\nogo \Sexit{2} \whiteAll}
                    ) \black{1} } \big) \marked{\black{1}}
      } \marked{\black{1}} \right)
    } \marked{\black{1}}
    } \\
  \end{array}
\]
From this final microstate, we can read that:
\begin{itemize}
  \item O is not emitted because its only emitter inside body'' is $\nogo \Semit{O} \whiteAll$,
  \item the activated $\Spause$ and $\Sawimm s$ statements are $\go \Spause \black{1}$ inside check'' and $\go \Sawimm{A} \black{1}$ inside body''.
\end{itemize}

\subsection{Properties of the Microstep Semantics
\texorpdfstring{\marginlink[\#micro_lt]{Esterel.Semantics.Microstep}{-2mm}}{}}

As it is lower-level than previous semantics, the reader my wonder which properties of the previous semantics are preserved by the microstep semantics.
In fact, it enjoys most of the properties of other Kernel Esterel semantics we have studied so far:
\begin{theorem}
  \label{thm:microstep-properties}
  The microstep semantics enjoys the following properties:
  \begin{itemize}
    \item The base statement and Sel values are unchanged by execution:
    \marginlink[\#micro_base]{Esterel.Semantics.Microstep}{0mm} \\
      \hfill $\forall E, p, p'\!, \; \micro E p {p'} \implies \base{p'} = \base p \land Sel(p') = Sel(p)$;
    \item The input color is unchanged by execution:
      \marginlink[\#micro_toplevel_input]{Esterel.Semantics.Microstep}{0mm}
      $\forall E, p, p'\!, \; \micro E p {p'} \implies \inC p = \inC{p'}$;
    \item Each microstep increases information:
      \marginlink[\#micro_lt]{Esterel.Semantics.Microstep}{0mm}
      $\forall E, p, p'\!, \; \micro E p {p'} \implies p < p'$; \\
      Corollary: Total microstates cannot execute.
  \end{itemize}
\end{theorem}
\noindent
Since there is no output event~$E'$\!, the event domain preservation property does not make sense.

The reader may have noticed that an important property of previous semantics is missing: determinism.
And in fact, we do lose it, because microsteps are local and may happen in several places in parallel, for instance when starting both branches of $\Spar p q$ (rules parIL and parIR).
This may seem like a serious issue as the Esterel language is meant to program critical systems, in which determinism is often paramount.
Hopefully, we have the next best results: confluence and termination.

\paragraph*{Determinism, confluence, and termination}
The reason why determinism is so important for a semantics is to ensure uniqueness of the execution path and therefore of the result.
\emph{Confluence} is a property that allows to recover uniqueness of the result for non-deterministic semantics: it expresses that for any two execution paths, we can execute them further to reach the same state.
In particular, the final result (if any) must be the same (as it cannot execute further).
When all possible execution paths are finite, Newman's lemma~\cite{Newman-lemma} shows that \emph{local confluence}, where only single step executions need to be considered, is enough.
\begin{center}
    \begin{tikzpicture}[node distance=5em]
      \node (p) {$p$};
      \node[right of= p]  (p1) {$\; p_1$};
      \node[below of= p]  (p2) {$p_2$};
      \node[below of= p1] (p') {$\; p'$};
      \draw[->] (p) -- node[at end] {\raisebox{0.3em}{$\;^*$}} (p1);
      \draw[->] (p) -- node[at end] {$\;\;\;_*$} (p2);
      \draw[->, dashed] (p1) -- node[at end] {$\;\;\;_*$} (p');
      \draw[->, dashed] (p2) -- node[midway, label=below:Confluence] {} node[at end] {\raisebox{0.3em}{$\;^*$}} (p');
    \end{tikzpicture}
    \hspace{10em}
    \begin{tikzpicture}[node distance=5em]
      \node (p) {$p$};
      \node[right of= p] (p1) {$p_1$};
      \node[below of= p] (p2) {$p_2$};
      \node[below of= p1] (p') {$p'$};
      \draw[->] (p) -- (p1);
      \draw[->] (p) -- (p2);
      \draw[->, dashed] (p1) -- node[at end] {$\;\;\;_*$} (p');
      \draw[->, dashed] (p2) -- node[midway, label=below:Local Confluence] {} node[at end] {\raisebox{0.3em}{$\;^*$}} (p');
    \end{tikzpicture}
\end{center}

For Esterel, termination of the microstep semantics is quite easy to prove as this semantics is meant to mimic the circuit execution in which every execution step sets the value of (at least) one more wire.
And indeed, if we define the $\N$-valued measure \texttt{Mmeasure}\marginlink[\#Mmeasure]{Esterel.Semantics.Microstate}{0mm} over microstates counting the number of wires still having a value~$\bot$, we can prove that the measure strictly decreases with every microstep.
Thus, from any microstate, only a finite number of microsteps are possible.
\begin{remark}
This measure is actually a different way to define an order between microstates but it is coarser than the Scott ordering of Section~\ref{sec:scott-order-microstate} as it allows comparison between any two microstates, including microstates not having the same underlying base statement or where information evolve in different locations.
\end{remark}

On the other hand, local confluence does not hold for arbitrary microstates.
For example, the (meaningless) microstate ``$\Sif s {\go \Snothing \white{\{ 0 \}}}{\go \Spause  \white{\{ 1 \}}}$'' can produce output color $\black 0$ and $\black 1$ depending on which branch we execute, and cannot be made confluent.
Nevertheless, the execution of an Esterel program will never exhibit such meaningless microstates.
Thus, we need to restrict the proof of local confluence to microstates satisfying a well-formation invariant.
Once we have that invariant called $\VC E p$ (see below), we can prove that the microstep semantics is confluent:
\begin{theorem}[Confluence]
  \marginlink{Esterel.Proofs.MicroConfluence}{-2mm}
  \label{thm:microstep-confluence}
  The microstep semantics is locally confluent from any well-formed starting microstate.
  As it is also terminating, by Newman's lemma~\cite{Newman-lemma}, it is then confluent (from any well-formed starting microstate).
\end{theorem}

\paragraph*{The \VC E p predicate\marginlink[\#valid_coloring]{Esterel.Proofs.ValidColoring}{-2mm}}
\label{sec:valid-coloring}

The predicate $\VC E p$ represents the fact that microstates appearing along the evaluation of an Esterel program are not arbitrary: they satisfy some invariants coming from the circuit semantics and from the circuit translation of Esterel.
More precisely, it expresses that
\begin{enumerate}
  \item the input and output colors inside the microstate~$p$ are coherent with the information they have access to, that is, the output value of a gate contains no more information than its input values allow (circuit semantics invariant);
  \item Sel and Go are exclusive and so are branches of a test or a sequence (Esterel circuit translation invariant);
  \item all signals used in~$p$ are declared in~$E$ (well-formation invariant).
\end{enumerate}
The last point is added only to ensure the well-formation condition $\VD E p$ without which the execution of an Esterel program does not makes sense.
In particular, $\VC E p$ then implies $\VD E {\base p}$.\marginlink[\#valid_coloring_valid_dom]{Esterel.Proofs.ValidColoring}{-2mm}

The set of constraints for a microstate $\Sin p \Sout$ is defined by case analysis over~$p$ and is given in Figure~\ref{fig:valid-coloring}.
The interested reader is invited to look at the Coq code, which is more explicit and commented and less likely to contain a mistake.
\newcommand\spand{\quad\!\land\quad\!}
\begin{figure}
  The $\VC E {\Sin p \Sout}$ property is defined by induction over~$p$.
  For readability, we omit two parts in these cases:
  \begin{itemize}
    \item all cases contain \texttt{input\_invariant}$(E, p)$ and
    \item all compound microstates contain recursive calls for sub-microstates. \\
      This recursive call is explicitly shown for $\Ssignaldecl s p$ because of the change in the environment~$E$.
  \end{itemize}
\[
  \begin{array}{@{}l@{\quad}c@{\quad}l@{}}
    \Snothing        & \mapsto & sel = false \spand \WIRE 0 \Sin \Sout \\[0.5em]
    \Spause          & \mapsto & \white{\{0, 1\}} \le \Sout \spand \PAUSE{sel}{\Sin}{\Sout} \\[0.5em]
    \Sexit k         & \mapsto & sel = false \spand \WIRE{k+2}{\Sin}{\Sout} \\[0.5em]
    \Semit s         & \mapsto & s ∈ E \spand sel = false \spand \WIRE 0 \Sin \Sout \\[0.5em]
    \Sawimm s        & \mapsto & s ∈ E \spand \white{\{0, 1\}} \le \Sout \\
                     &         & \land\quad (\Sin = \nogores \implies \Sout \le \whiteAll)
                                 \spand (\Sin \neq \nogores \land s^\bot \in E \implies \Sout = \white{\{0, 1\}}) \\
                     &         & \land\quad (\Sin = \gores \land s^+ \in E \implies \Sout \le \black 0)
                                 \spand (\Sin \neq \gores \land \Sin \neq \nogores \land s^+ \in E \implies \Sout \le \white{\{0\}}) \\
                     &         & \land\quad (\Sin = \gores \land s^- \in E \implies \Sout \le \black 1)
                                 \spand (\Sin \neq \gores \land \Sin \neq \nogores \land s^- \in E \implies \Sout \le \white{\{1\}}) \\[0.5em]
    \Strap p         & \mapsto & \inC p \le \Sin \spand sel = Sel(p) \, \spand ↓\outC{\fromstmt{\base p}} \le \Sout \le ↓\outC p \\[0.5em]
    \Sshift p        & \mapsto & \inC p \le \Sin \spand sel = Sel(p) \, \spand ↑\outC{\fromstmt{\base p}} \le \Sout \le ↑\outC p \\[0.5em]
    \Ssuspend s p    & \mapsto & s ∈ E \spand sel = false \spand Go(p) \le Go(\Sin) \spand Res(p) \le Res(\Sin) \band Sel(p) \band \lnot E(s) \\
                     &         & \land\quad \{1\} \cup \outC{\fromstmt{\base p}} \le \Sout \le \SuspNow{Res(\Sin) \band Sel(p) \band E(s)}{\outC p} \\[0.5em]
    \Sif s p q       & \mapsto & s^b \in E \spand sel = Sel(p) \mathbin{||} Sel(q) \spand Sel(p) \band Sel(q) = false \\
                     &         & \land\quad Go(p) \le Go(\Sin) \band b \spand Go(q) \le Go(\Sin) \band \lnot b \\
                     &         & \land\quad Res(p) \le Res(\Sin) \spand Res(q) \le Res(\Sin) \\
                     &         & \land\quad \outC{\fromstmt{\base p}} \cup \outC{\fromstmt{\base p}} \le \Sout \le \outC p \cup \outC q \\[0.5em]
    \Ssequence p q   & \mapsto & sel = Sel(p) \mathbin{||} Sel(q) \spand Sel(p) \band Sel(q) = false \\
                     &         & \land\quad \inC p \le \Sin \spand Go(q) \le (\outC p)[0] \spand Res(p) \le Res(\Sin) \\
                     &         & \land\quad (\outC{\fromstmt{\base p}} \setminus 0) \cup \outC{\fromstmt{\base q}} \le \Sout \le (\outC p \setminus 0) \cup \outC q \\[0.5em]
    \Sparallel p q   & \mapsto & sel = Sel(p) \mathbin{||} Sel(q) \spand \inC p \le \Sin \spand \inC q \le \Sin \\
                     &         & \land\quad \outC{\fromstmt{\base p}} \cup \outC{\fromstmt{\base p}} \le \Sout \le \mathrm{synch}_{Sel(p), Sel(q)}(\outC p, \outC q) \\[0.5em]
    \Ssignaldecl s p & \mapsto & sel = Sel(p) \spand \inC p \le \Sin \spand \outC{\fromstmt{\base p}} \le \Sout \le \outC p \\
                     &         & \VC{\addEvent s {\toevent{p}{\addEvent s \bot E}} E}{p}\\
  \end{array}
\]
where
\begin{itemize}
  \item $\WIRE k \Sin \Sout$ expresses in terms of input/output relation that $Go(\Sin)$ is directly fed into the $k$-th component of $\Sout$, which is the case for the $\Snothing$, $\Sk k$, and $\Semit s$ statements:
    \begin{center}
    $\WIRE k \Sin \Sout =
    \begin{cases}
      \Sout = \white{\{k\}} & \text{when\quad} Go(\Sin) = \bot \\
      \Sout \leq \whiteAll & \text{when\quad} Go(\Sin) = - \\
      \Sout \leq \black k & \text{when\quad} Go(\Sin) = + \\
    \end{cases}$
    \end{center}
  \item $\PAUSE{Sel}{\Sin}{\Sout}$ expresses the input/output relation for the $\Spause$ (\Tpause) statement:
    \begin{center}
    $\PAUSE{Sel}{\Sin}{\Sout} =
    \begin{cases}
      \Sout \leq \black 1 & \text{when\quad} Go(\Sin) = + \\
      \Sout \leq \black 0 & \text{when\quad} Res(\Sin) = + \text{ and } Sel = + \\
      \Sout \leq \whiteAll & \text{when\quad} \Sin = \nogores \\
      \Sout \leq \white{\{0\}} & \text{when\quad} Go(\Sin) = - \text{ and } Sel = + \text{ and } Res(\Sin) \neq + \\
      \Sout \leq \white{\{1\}} & \text{when\quad} Go(\Sin) = \bot \text{ and } (Res(\Sin) = - \text{ or } Sel = -) \\
      \Sout \leq \white{\{0, 1\}} & \text{otherwise} \\
    \end{cases}$
    \end{center}
\end{itemize}
  \caption{The $\VC E p$ invariant.}
  \label{fig:valid-coloring}
\end{figure}

For illustration purposes, we detail the case $\Sif s p q$ and how it should be understood:
\begin{itemize}
  \item the input color must satisfy its Esterel invariant:
    \hfill $\texttt{input\_invariant} \, sel \, \Sin$;
  \item the signal~$s$ must exist in the environment~$E$; let~$b$ be its value:
    \hfill $s^b \in E$;
  \item $\Sif s p q$ is active iff one of~$p$ and~$q$ is:
    \hfill $sel = Sel(p) \mathbin{||} Sel(q)$;
  \item both~$p$ and~$q$ cannot be active at the same time:
    \hfill $Sel(p) \band Sel(q) = false$;
  \item $Go(p)$ and~$G(q)$ are computed from $\Sin$ and $b$:
    \hfill $Go(p) \le  Go(\Sin) \band b$ and $Go(q) \le  Go(\Sin) \band \lnot b$;
  \item the rest of the input color is transmitted to~$p$ and~$q$:
    \hfill $Res(p) \le Res(\Sin)$ and $Res(q) \le Res(\Sin)$;
  \item $out$ contains at least the static information:
    \hfill $\outC{\fromstmt{\base p}} \cup \outC{\fromstmt{\base p}} \le \Sout$;
  \item \dots and at most the information of~$p$ and~$q$:
    \hfill $\Sout \le \outC p \cup \outC q$;
  \item recursively,~$p$ and~$q$ must be well-colored:
    \hfill $\VC E p$ and $\VC E q$.
\end{itemize}

To ensure that the $\VC E p$ property is indeed an invariant, we prove that it holds at the start of execution and is preserved by the microstep semantics:

\begin{theorem}[Invariant of the circuit translation] \hfill
  \label{thm:microstep-invariants}
  \begin{itemize}
    \item For any well-formed statement~$p$, $\VC E {\fromstmt{p}}$ holds\marginlink[\#from_stmt_invariant]{Esterel.Proofs.ValidColoring}{0mm} ;
    \item For any well-formed state~$\state p$, $\VC E {\fromstate{\state p}}$ holds\marginlink[\#from_state_invariant]{Esterel.Proofs.ValidColoring}{0mm} ;
    \item The property \texttt{valid\_coloring} is preserved by the microstep semantics\marginlink[\#valid_coloring_inductive]{Esterel.Proofs.ValidColoring}{0mm}.
  \end{itemize}
\end{theorem}

This invariant directly entails the following properties of the control flow in the microstep semantics, which are rather intuitive and strengthen our confidence in its definition:
\begin{itemize}
  \item Control is never created, only propagated:
    \begin{itemize}
      \item $\VC E {\Sin p \black k} \implies \Sin = \gores$; \marginlink[\#valid_coloring_Black_is_exec]{Esterel.Proofs.ValidColoring}{0mm}
      \item $\VC E {\Sin p \whiteAll} \implies \Sin = \nogores$; \marginlink[\#valid_coloring_empty_is_not_exec]{Esterel.Proofs.ValidColoring}{0mm}
    \end{itemize}
  \item Branches in a test and a sequence are not simultaneously active:
    \begin{itemize}
      \item$\VC E {\Sin (\Sif s p q) \Sout} \implies \inC p \neq \gores \lor \inC q \neq \gores$;
          \marginlink[\#valid_coloring_if_exclusive]{Esterel.Proofs.ValidColoring}{0mm}
      \item $\VC E {\Sin (\Ssequence p q) \Sout} \implies \outC p \setminus 0 \neq \black{k} \lor \inC q \neq \gores$;
          \marginlink[\#valid_coloring_seq_out_exclusive]{Esterel.Proofs.ValidColoring}{0mm}
    \end{itemize}
\end{itemize}

\paragraph*{Useful properties for optimization and reasoning}
Using our microstep semantics, we can prove properties of control flow which can lead to better simulation or reasoning, as well as justify some design choices.
We illustrate this on two examples.

First, dead code does not create a completion code:

\begin{lemma}[Execution of non-executed microstates]
  \label{thm:no-exec}
  A non-executed microstate gets output color $\whiteAll$. \marginlink[\#microsteps_dead_stmt_strong]{Esterel.Semantics.Microstep}{0mm} Formally, for all~$p$,~$\Sout$ and~$E$:
  \begin{itemize}
    \item $\micros E {\nogo p \Sout} {\nogo p \whiteAll}$ ;
    \item For $\nores$ (that is, when Sel$⁺$), we merely have $\micros E {\nores p \Sout} {\nores p' \white K}$, as the Go value may be required to eliminate some completion codes (\eg, 1 for pause).
  If we additionally have Go$^-$, then we recover $\micros E {\nores p \Sout} {\nores p \whiteAll}$.
  \end{itemize}
\end{lemma}
\noindent
This property is the (partial) converse of the fact that control is never created, only propagated and can speed up reasoning about $\Sif s p q$ statements.
Indeed, it means that there is no need to simulate the branch not taken in a branching test, we already know that all wires, both control and completion codes, will eventually be set to 0.



Second, when the starting microstate is not active (that is, it has Sel$^-$), the Res wire has no impact on the execution of the current reaction.
In other words, in inactive microstates, only Go matters, Res does not.
\begin{lemma}[Invariance by Res]
  \label{thm:micro-ignores-Res}
  \marginlink[\#micro_surface_ignores_Res]{Esterel.Semantics.SurfaceIgnoresRes}{0mm}
  In an inactive microstate, the Res wire does not matter.
  In other words, an inactive microstate is invariant by changes to the Res wires: any two equal (up to Res values) well-formed microstates have the same executions (up to Res value computation).

  If we denote by $\equiv_{Res}$ the equality of microstates up to Res wires, we can write it:
  \[
  \noindent
    \begin{array}{l@{\,}l}
      \forall p, q, p',
      & \VC E p \implies \!\VC E q \implies \\
      & Sel(p) = - \implies p \equiv_{Res} q \implies \\
      & \micro E {p}{p'} \implies p' \equiv_{Res} q \lor \Big( \exists q', \micro E q {q'} \land  p' \equiv_{Res} q' \Big)~.
    \end{array}
  \]
\end{lemma}
We do not need the hypothesis $Sel(q) = -$ because $p \equiv_{Res} q$ entails $Sel(p) = Sel(q)$.
The disjunction in the conclusion can be understood as follows: either the microstep from $p$ to $p'$ computed a Res value, in which case we have $p' \equiv_{Res} p \equiv_{Res} q$ (left disjunct), or it computed something else and $q$ can mimic it (right disjunct).
This result justifies why the notation introduced in Section~\ref{sec:input-color-notations} do not consider $Res$ when $Sel = -$.

\subsection{Refinement between the constructive state semantics and the microstep semantics}
\label{sec:refinement-state-microstep}

In order to reduce the proof burden and focus on control and completion codes, some parts of the circuit are not modeled in the microstep semantics, namely the computation of signal values and the synchronizer for the parallel statement.
In the first case, it amounts to a big OR gate over the Go component of the input color of all emitters.
In the second case, we work with the specification of the synchronizer rather than a given implementation, thus permitting to reason about what it \emph{should} do and allowing us to swap and compare implementations.

Unlike the previous semantics which tackle reactions, the microstep one deals with steps \emph{within} one reaction.
This discrepancy makes it harder to relate both semantics because the microstep semantics is much more precise and thus can distinguish states that cannot be expressed in the state semantics and whose translations in the state semantics are identical.
Nevertheless, the biggest obstacle when connecting the constructive state semantics and the microstep semantics comes from the \CanName{} and \MustName{} functions.

The state semantics uses the auxiliary functions \CanName{} and \MustName{} to compute the value of a local signal then performs the evaluation using this value.
Thus, the body of the statement is evaluated twice: first partially to get the value of the local signal and later from scratch again to compute the full step using the value of the local signal.
On the contrary, in the microstep semantics, once the value of the local signal is computed, we continue the evaluation of the microstate using this value \emph{without restarting from scratch}, exactly as a circuit does.
This is the main difference with the microstep semantics presented in the \emph{Compiling Esterel} book~\cite{CompilingEsterel}.

Therefore, for the simulation we need to translate the condition for applying a rule for local signals, $s \in \MustS p E$ or $s \notin \CanS + p E$, into the corresponding condition in the microstep semantics, that is: there exists a microstep sequence $\micros E p {p'}$ for some~$p'$ emitting~$s$ or some~$p'$ surely not emitting it.
This is actually the hardest part about the simulation proof, which makes sense as it is also an essential difference between the semantics.
Because the definitions of \MustName{} and \CanName{} are mutually recursive with recursive calls for $\Sseq p q$ requiring information about completion codes and recursive calls for $\Sif s p q$ changing the tag $+$ or $\bot$ on \CanName, we actually need six mutually recursive properties: three for signals and three for completion codes.
\begin{lemma}[Interpretation of \MustName/\CanName{} as microsteps]
  \label{thm:MustCan-micro}
  \marginlink[\#Must_Can_to_event_start]{Esterel.Proofs.MicroMustCan}{0mm}

  For any statement~$p$ and any event~$E$ containing all signals used by~$p$, \ie satisfying the $\VD E p$ predicate, we have the following properties:
  \begin{itemize}
    \item Any signal that must be emitted is indeed emitted after enough microsteps: \\
      $\forall s, s \in \MustS p E \implies \exists p', \micros E {\changeI{\go}{\fromstmt{p}}}{p'} \land s^+ \in \toevent{p'}{E}$ ;
    \item Any signal that cannot be emitted is indeed not emitted after enough microsteps: \\
      $\forall s, s \notin \CanS + p E \implies \exists p', \micros E {\changeI{\go}{\fromstmt{p}}}{p'} \land s^- \in \toevent{p'}{E}$ ;
    \item Any signal that can never be emitted is indeed not emitted no matter what the input color is: \\
      $\forall s, s \notin \CanS{\bot} p E \implies \forall \Sin, \exists p', \micros E {\changeI{\Sin}{\fromstmt{p}}}{p'} \land s^- \in \toevent{p'}{E}$ ;
    \item If the completion code must be $k$, then it is indeed $k$ after enough microsteps: \\
      $\forall k, k \in \MustK p E \implies \exists p', \micros E {\changeI{\go}{\fromstmt{p}}}{p'} \land \outC{p'} = \black k$ ;
    \item After enough microsteps, the possible completion codes are a subset of what they can be: \\
      $\exists p', \micros E {\changeI{\go}{\fromstmt{p}}}{p'} \land \outtoC{\outC{p'}} \subseteq \CanK + p E$ ;
    \item Regardless of input color, after enough microsteps, the possible completion codes are a subset of what they can be when not knowing if the statement is evaluated or not: \\
      $\forall \Sin, \exists p', \micros E {\changeI{\Sin}{\fromstmt{p}}}{p'} \land \outtoC{\outC{p'}} \subseteq \CanK{\bot} p E$.
  \end{itemize}
  where $\outtoCsymbol$ is a function converting an output color to a set of completion codes, defined by
  $\outtoC{\black k} = \{ k \}$ and $\outtoC{\white K} = K$.

  Six similar properties hold when resuming an active state rather than starting an inactive statement.\marginlink[\#Must_Can_to_event_resume]{Esterel.Proofs.MicroMustCan}{0mm}
\end{lemma}

Once we can translate \MustName{} and \CanName{} into a sequence of microsteps, we can prove the simulation theorem: any step in the constructive state semantics can be performed by a sequence of steps in the microstep semantics.
  Furthermore, the final microstate is total and produces the same completion code and the same output event: 
\begin{theorem}
  \label{thm:micro-refine-CSSmicro}
  For all $p$, $E$, $E'$, $k$, and $\term{p'}$, if $\CSSs p E {E'} k {\term{p'}}$ then there exists a total microstate $p''$ such that
  \marginlink[\#sCSS_microsteps]{Esterel.Proofs.CSS_Micro}{5mm}
  \[
    \micros E {\changeI{\go}{(\fromstmt{p})} \,} {p''}, \quad
    \toterm{p''}{E} = \term{p'}, \quad
    out(p'') = k \quad \text{and} \quad \toevent{p''}{E} = E'.
  \]
  For all $\state p$, $E$, $E'$, $k$, and $\term{p'}$, if $\CSSr{\state p} E {E'} k {\term{p'}}$ then there exists a total microstate $p''$ such that
  \marginlink[\#rCSS_microsteps]{Esterel.Proofs.CSS_Micro}{7mm}
  \[
    \micros E {\changeI{\res}{(\fromstate{\state p})} \,} {p''}, \quad
    \toterm{p''}{E} = \term{p'}, \quad
    out(p'') = k \quad \text{and} \quad \toevent{p''}{E} = E'.
  \]
\end{theorem}

\begin{proof}
  The proof amounts to finding an evaluation order leading from the state at the beginning of the instant to the one at the end of the instant.
  This order corresponds to input-to-output evaluation in circuits: first uses start rules, then context ones,\footnote{For $\Tsequence p q$, the rule transmitting the output of~$p$ to the input of~$q$ is naturally used between the evaluations of~$p$ and~$q$.} and finally end ones.
  Lemma~\ref{thm:MustCan-micro} and confluence are critical for the $\Ssignaldecl s p$ case.
  We use Lemma~\ref{thm:no-exec} to handle the parts of the terms which are not executed.
\end{proof}

Notice that this theorem is not an equivalence between the semantics, only a simulation.
The converse direction cannot hold because the microstep semantics is local whereas the state one is global.
For example, if we take a valid statement~$p$ and a non constructive one~$q$ (say, $\Ssignaldecl s {(\Sif s {\Snothing}{\Semit s})}$), then $\Spar p q$ cannot execute under the state semantics whereas the microstep one can execute~$p$ independently of~$q$, so that the final microstate would correspond to some $\Spar{p'} q$ in the state semantics.





\section{Properties and interpreters of Esterel semantics in Coq} 
\label{sec:coq-proofs}

\subsection{Summary of the properties of Esterel semantics}
\label{sec:Esterel-semantics-properties}

Most semantics satisfy a common body of properties, inherent to the Kernel Esterel language, that we recall now.
These properties are presented below only for the constructive semantics but can be adapted to other semantics in a straightforward way (see the Coq files of each semantics).
Table~\ref{tab:Esterel-semantics-properties} gives an overview of which properties are satisfied by which semantics.
\newcommand\yes{Yes}
\newcommand\no{No}
\newcommand\NA{N/A}

\begin{table}
  \centering
  \begin{tabular}{l|cccc}
    \hline
    \hfill Semantics  &  Logical   & Constructive &  Constructive   & Microstep \\
    Properties        & Semantics  &  Semantics   & State Semantics & Semantics \\
    \hline
    Deterministic     &    \no     &     \yes     &      \yes       &    \no    \\
    Confluent         &    \no     &     \yes     &      \yes       &    \yes   \\
    Total output      &    \yes    &     \yes     &      \yes       &    \NA    \\
    Inactive derivative &  weak    &     weak     &     strong      &    \NA    \\
    Statement invariance & \no     &     \no      &      \yes       &    \yes   \\
    \hline
  \end{tabular}

  \caption{Properties of the various Esterel Semantics.}
  \label{tab:Esterel-semantics-properties}
\end{table}
\begin{description}
  \item[Determinism]
    \marginlink[\#CBS_deterministic]{Esterel.Semantics.CBS}{0mm}
    The semantics is deterministic:
    \[
      \forall p \, E \, E_1' \, E_2' \, k_1 \, k_2 \, p_1' \, p_2', \quad
      \CBS p E {E_1'}{k_1}{p_1'} \implies \CBS p E {E_2'}{k_2}{p_2'} \implies
      p_1' = p_2' \land E_1' = E_2' \land k_1 = k_2
    \]
    Notice that this property does not hold for the logical semantics!
    For example, the program $\Ssignaldecl s {(\Sif{s}{\Semit s}{\Snothing})}$ presented in the introduction is non-deterministic.
    The microstep semantics is also non-deterministic because execution can happen in various places of a microstate.
    Nevertheless, it is confluent so that in practice this non-determinism is not really an issue.
  \item[Total output]
    \marginlink[\#CBS_Total]{Esterel.Semantics.CBS}{0mm}
    Output events are always total, that is, they contain only mappings to $+$ or $-$ and none to~$\bot$.
    This is trivially true for the logical semantics, as there is no~$\bot$ involved.
    \[
      \forall p \, E \, E' \, k \,p', \CBS p E {E'} k {p'} \implies \Total{E'} \quad \text{with} \quad \Total{E'} := \forall s \in \dom(E'), E'(s) \neq \bot
    \]
    This property does not make sense for the microstep semantics as there is no output event.
  \item[Inactive derivative]
    \marginlink[\#CBS_inert_derivative]{Esterel.Semantics.CBS}{0mm}
    If the completion code is not~$1$, then the derivative is~$\Snothing$:
    \[
      \forall p \, E \, E' \, k \,p', \CBS p E {E'} k {p'} \implies k \neq 1 \implies p' = \Snothing
    \]
    This result means that if a statement is not pausing, then its execution has finished (either by normal termination or raising an exit).
    This result is the motivation for introducing the $\deltafun k p$ function.
  \item[Underlying statement invariance]
    \marginlink[\#sCSS_base]{Esterel.Semantics.CSS}{0mm}
    The underlying statement does not change during execution:
    \[
      \forall p \, E \, E' \, k \,p', \CSSr p E {E'} k {p'} \implies \base{p} = \base{p'}
    \]
    This property is true only for the state and microstep semantics.
\end{description}

Of course, not all properties are shared between all semantics, some properties are still specific to each semantics.
For example, the result ``Inactive derivative'' above can be strengthened into an equivalence for the constructive state semantics, as the derivative of $\Spause$ (\Tpause) is no longer an inactive statement but the active state $\state{\Spause}$.

\begin{lemma}[Strong Inactive derivative]
  \label{thm:inactive-derivative-state}
    \marginlink[\#sCSS_inert_derivative]{Esterel.Semantics.CSS}{0mm}
    For all $p$, $E$, $E'$, $k$, and $p'$, if $\CSSs p E {E'} k {\term{p'}}$ then we have $k \neq 1 \iff \term{p'} = p$. \\
    \marginlink[\#rCSS_inert_derivative]{Esterel.Semantics.CSS}{0mm}
    For all $\state p$, $E$, $E'$, $k$, and $p'$, if $\CSSr{\state p} E {E'} k {\term{p'}}$ then we have $k \neq 1 \iff \term{p'} = \base{\state p}$.
\end{lemma}



\subsection{Interpreters}

For all semantics, we have defined interpreters in Coq and proven their correctness and completeness with respect to the corresponding SOS semantics.
These interpreters can be used as an executable semantics of Kernel Esterel programs.
They take as input a program (either as a statement, a state, or a microstate depending on the interpreter) and an input event; they output either nothing (\texttt{None}) when the input program cannot execute, or the output event, completion code and derivative of a possible execution step.

As an illustration, below is the correctness statement for the interpreter of the constructive behavioral semantics.
The same type of result holds for both variants (start and resumption) of the constructive state semantics.
\begin{theorem}[Correctness of the CBS interpreter]
  \marginlink[\#CBS_interp_correct]{Esterel.Proofs.CBS}{0mm}
  For all $p$, $E$, $E'$, $k$, and $p'$, if all signals of~$p$ belong to~$E$ (that is, $\VD p E$),
  then we have $\texttt{CBS\_interp}(p, \, E) = \texttt{Some} (E', k, p') \iff \CBS p E {E'} k {p'}$.
\end{theorem}
\begin{proof}
  All such proofs are straightforward (if sometimes long) and follow the structure of the interpreter or of the semantics in the premise.
  See the Coq files of each semantics for details.
\end{proof}

Since the microstep semantics may execute in several parts of a microstate, we need to choose an evaluation order (even though the actual order does not matter because of confluence).
We use the start-context-end order: first start rules, then context rules, then end rules.
Because of this choice, the completeness result no longer holds (one could pick a different order that the interpreter does not follow) but instead we have a weaker result: if a microstep execution step is possible, then the interpreter is not stuck.
\begin{lemma}[Correctness and partial completeness of the microstep interpreter]
  \[
  \begin{array}{ll}
    \forall p, E, p', \;
    \texttt{micro\_interp}(p, E)= \texttt{Some}\, p' \implies \micro E p p'
    \marginlink[\#micro_interp_correct]{Esterel.Proofs.Microstep}{0mm} \\
    \forall p, E, p', \;
      \micro E p p' \implies \texttt{micro\_interp}(p, E) \neq None
    \marginlink[\#micro_interp_complete]{Esterel.Proofs.Microstep}{0mm} \\
  \end{array}
  \]
\end{lemma}

We can recover a complete interpreter by using the iterated microstep semantics, that is, the reflexive transitive closure of the microstep semantics.
Then, all microstates can execute (thanks to reflexivity) so we can drop \texttt{Some}/\texttt{None} and have a simpler correctness statement:
\[
  \forall p, E, \micros E p {\texttt{micros\_interp}(p, E)}
  \quad \land \forall p'\!, \; \lnot  \Big(\micro E {\texttt{micros\_interp}(p, E)} p' \Big)
  \marginlink[\#micros_interp_correct]{Esterel.Proofs.Microstep}{0mm}
\]
Confluence ensures that $\texttt{micros\_interp}(p, E)$ is the final state of any execution chain starting from~$p$ in the environment~$E$.

Finally, the interpreter for the logical semantics is only correct but not complete as the logical semantics is non deterministic and non confluent.
Indeed, given a statement~$p$ that can non-deterministically pause or terminate and a statement~$q$ that cannot execute, then on $\Sseq p q$ the interpreter will choose to make~$p$ terminate and fail to execute~$q$ making the overall $\Sseq p q$ statement non executable.
On the contrary, the LBS semantics allows~$p$ to pause, which results in a valid reaction as~$q$ is not considered.
For instance, one may take~$p$ to be $\Ssignaldecl s {(\Sif s {\Semit s} \Spause)}$ and~$q$ to be $\Ssignaldecl s {(\Sif s \Snothing {\Semit s})}$.
\begin{remark}
It would be possible to create a complete interpreter for the LBS semantics by returning the set of possible executions instead of just choosing one.
This seems to be of little practical interest as the LBS semantics is not very useful anyway.
\end{remark}

\subsection{Proof effort}

Table~\ref{tab:proof-effort} presents the size of the Coq formalization.
In the semantics definitions, we can observe that the first three semantics (LBS, CBS, CSS) are defined with their main properties in a few hundred lines, the CSS being a bit bigger because it features two sets of rules (start and resumption).
On the other hand, the microstep semantics is twice as big, and even the definition of microstates is thrice bigger than the one of statements and states.
Half of these numbers are due to the interpreters and their proofs of correctness.

Among proofs, the largest ones are by far local confluence (1400 l.) and the link between \MustName{}/\CanName{} and microsteps for the start and resumption cases (1500 l. each).
They are used in the simulation proof between the CSS and the microstep semantics, which is comparably short (450 l. for each of the start and resume variants).
The proof of local confluence is done very naively by a big case analysis on each possible execution steps, in total 280 cases to consider.
It is very likely that it could be considerably shortened by a smarter decomposition: half of the cases are symmetric of the other half and many rules simply commute.

\begin{table}[btp]
  \begin{tabular}{rrr@{\qquad}l}
    \hline
     Spec  &  Proof & Comme\rlap{nts} & \qquad File \\ \hline
      141  &   104  &    18  &  Util/Coqlib.v \\
       28  &   130  &    25  &  Util/MapSig.v \\
      274  &   651  &     9  &  Util/MapList.v \\
       45  &     0  &     5  &  Util/Notations.v \\
      171  &   201  &    23  &  Util/Events.v \\
      198  &   251  &    38  &  Util/SemanticsCommon.v \\
      857  &  1337  &   118  &  \textbf{Total for Util/} \\ \hline
      262  &   200  &    40  &  Semantics/LBS.v \\
       60  &   605  &   120  &  Semantics/MustCan.v \\
      163  &   515  &    94  &  Semantics/CBS.v \\
       52  &   160  &    42  &  Semantics/StateMustCan.v \\
      318  &   812  &   115  &  Semantics/CSS.v \\
       58  &    85  &    11  &  Semantics/InputColor.v \\
      248  &   213  &    43  &  Semantics/OutputColor.v \\
      552  &   740  &    89  &  Semantics/Microstate.v \\
      396  &   380  &   139  &  Semantics/Microstep.v \\
      158  &   983  &   118  &  Semantics/Microsteps.v \\
     2267  &  4693  &   811  &  \textbf{Total for Semantics/} \\ \hline
       58  &   348  &    53  &  Proofs/CBS\_LBS.v \\
        9  &   274  &    87  &  Proofs/CBS\_CSS.v \\
      245  &  1101  &   193  &  Proofs/ValidColoring.v \\
       23  &  1382  &    81  &  Proofs/MicroConfluence.v \\
       17  &    27  &     0  &  Proofs/MicrostepsFacts.v \\
       80  &  3032  &   211  &  Proofs/MicroMustCan.v \\
       11  &   821  &    97  &  Proofs/CSS\_Micro.v \\
       98  &   200  &    16  &  Proofs/SurfaceIgnoresRes.v \\
      541  &  7185  &   738  &  \textbf{Total for Proofs/} \\ \hline
      311  &    98  &    29  &  Definitions.v \\
     3976  & 13313  &   1696 &  \textbf{Total} \\ \hline
  \end{tabular}
  \caption{Size of the Coq formalization, as given by \texttt{coqwc}, ordered by dependency within each folder.}
  \label{tab:proof-effort}
\end{table}



\section{Conclusion and future work} 
\label{sec:conclusion}

This paper has presented new work on the Coq-based formal
verification of the Kernel Esterel semantics chain presented in the
web draft book~\cite{Berry:ConstructiveSemanticsOfPureEsterel} written
by the second author 20 years ago, augmented by the definition of a
new fine-grain operational semantics developed and formally verified
by the first author.  Unlike the coarser-grain operational semantics
introduced by Dumitru Potop in~\cite{CompilingEsterel}, this new
operational semantics closely mimics the circuit translation
introduced in~\cite{Berry:ConstructiveSemanticsOfPureEsterel}, which
was the basis of both the academic Esterel v5 compiler to C and the
industrial Esterel v7 compiler to hardware circuits or
software code. Its main advantage compared to circuits is its SOS
inductive logical rules presentation, which makes the formal proof of its
adequacy vs. the constructive semantics much simpler with current Coq
technology.

Let us present a short history of this work. When he completed the
book~\cite{Berry:ConstructiveSemanticsOfPureEsterel}, the second author
did not even include sketches of correctness proofs, for two reasons.
First, he was completely involved in the industrial
development and applications of the Esterel v7 compiler, which handles
a much more expressive and modular language than v5. The v7 language
and compiler have been successfully used by several hardware and
software companies, before being taken over by Synopsys (they are not
available any more). Second, he thought that only machine-checked
correctness proofs would be fully convincing, but the formal verifiers
were not yet mastered enough in the beginning of the 2000’s for such
an endeavor.

The real game changer in the semantics/compiling field was the
Coq-based construction and formal verification of the CompCert
Coq-verified compiler by Xavier Leroy and his
team~\cite{Leroy-Compcert-Coq}. This large-scale work was itself
inspired by the Coq-based formal proof of a major mathematical theorem
by Gonthier and his team~\cite{Gonthier2013-OddOrderTheorem} (Gonthier
himself had played a major role in the early development of Esterel).
Our formal proofs, technically developed by the first author, followed
a similar track. No mistakes were found in the book’s claims, apart
from the occasional typo, although it did not even contain informal
correctness proofs.

Nevertheless, it must be noted that this paper does not finish
the job, for three reasons: first, it limits the kernel language to
loop-free programs; second, the new operational semantics is not yet
formally related to the Boolean circuits generation described
in~\cite{Berry:ConstructiveSemanticsOfPureEsterel}, which is the heart of
the v5 and v7 compilers. We briefly detail these points below.
Third, but less importantly, we did not yet study the
full Esterel language, which also
supports data definitions and calculations. As shown by the way
the v5 and v7 compilers handle data, this only requires adding data dependencies to
the signal dependency structure, which should raise no difficult
issues.

\subsection{Future work: statement and signal reincarnation}
\label{sec:reincarnation}

We chose to postpone the study of loops because they lead to a subtle
phenomenon: \emph{statement and signal reincarnation}. When embedded
in possibly nested loop statements, a given Esterel statement can be
reincarnated (\ie re-executed) several times in the same instant,
but in a clean way: a new incarnation can only occur when the previous
one has terminated. Consequently, a locally scoped signal may take
several simultaneous incarnation statuses in the same instant when its
declaration is embedded in loops. But these are taken in strict
succession of scope entering and exiting, as are the reincarnations of
its declaration statements; this makes only one signal incarnation
status visible at a time in any statement.
Thus, several incarnations can appear but each will disappear before the next one appears.
Reincarnation occurs naturally when
programming, and raises no difficult problems in practice.

Three quite different technical solutions have been developed to deal
with reincarnation. The first naive one consists of syntactically
duplicating the body of each loop in an inside-out way, which
completely suppresses reincarnation since Esterel requires that the
body of a loop cannot terminate instantaneously.
In $p\!*$ the loop body $p$ may terminate and be restarted at the same instant.
But in $(p;p)\!*$, since $p$ cannot terminate instantaneously, the
first and second $p$'s cannot be executed in the same instant.
But this trivial solution is impractical since inside-out copying
may exponentially increase the source code size.

In~\cite{Berry:ConstructiveSemanticsOfPureEsterel}, the second author
used a more semantic way based on \emph{incarnation indices} that tag
in a unique way each occurrence of a given statement or signal in
execution proof trees. Since restarting a statement $p$ can only
execute its \emph{surface}, \ie the statements instantaneously
reachable when $p$ starts, such indices only lead to worst-case
proof tree and generated circuit sizes that are quadratic compared to the statement
size. Furthermore, the size increase is only quasi-linear in practice for
most useful programs.

In~\cite{TardieuDeSimone:LoopsInEsterel}, Tardieu introduced an
alternative approach based on a semantically-guided partial syntactic
duplication of statements, which copies the surfaces of source
statements also in a quadratic worst-case and most often quasi-linear way.
The price to pay is the addition of a {\Tgotopause} statement, \ie a
{\Tgoto} statement limited to {\Tpause} targets. Both approaches are
compatible, since Tardieu’s approach can also be viewed as a way of
duplicating statements based on a static conservative calculation of
possible incarnation indices.

It is not yet clear which approach should be preferred to complete our
Coq-based verification. Incarnation indices have a simple and clean
algebraic structure, which should make them easy to add to the current
semantic rules. But the proofs will become heavier since they will
require the addition of an extra decreasing incarnation index-related
ordinal. Tardieu's approach may look simpler at first glance, because
the proof might require less changes, but it has another drawback: the
addition of {\Tgotopause} breaks locality of programs because a
{\Tgotopause} may activate a {\Tpause} arbitrarily far away, making
SOS semantics and inductive reasoning less adequate.

\subsection{Relating the operational semantics to the circuit translation}

The relation between the microstep semantics and the circuit semantics is the only missing step to prove the compiler correct: by combining the simulation results of this paper and this last missing simulation, one gets a proof of semantic preservation between the original Esterel program seen in the constructive semantics and its translation as a synchronous digital circuit seen in the circuit semantics.

Let us now finally analyze the technical difficulty of relating SOS
semantics to circuits.
Intuitively, the circuit generated by an Esterel program by the translation
of~\cite{Berry:ConstructiveSemanticsOfPureEsterel} can be viewed as a
folding of all possible proof trees of the new microstep operational semantics
into a single and possibly cyclic graph of Boolean gates and wires (note
that static cycles yield no dynamic electrical problems for
constructive circuits,
see~\cite{MendlerShipleBerry:ConstructiveCircuits}). Registers
implement {\Tpause} statements, with value 1 when a pause is
active in a given program state or 0 when it is inactive. Specific
gate clusters represent proof-tree nodes, and
specific wire groups represent proof-tree branches, taking value sets
that precisely encode the input and output colors attached to statements in the new
operational semantics. In theory, making formal this informal
correspondence should raise no real problem. But, as the saying goes: "In
theory, theory and practice are the same and in practice, they are
not".
The unfortunate practical problem is that finding a circuit representation in Coq well-suited to this formalization work is critical to its success.

\bibliographystyle{plainurl}
\bibliography{esterel} 

\end{document}